\documentclass[a4paper,12pt]{article}

\usepackage{amsfonts}
\usepackage{amsmath}
\usepackage{amssymb}
\usepackage{amsthm}
\usepackage{natbib}
\usepackage{ifpdf}
\usepackage{multicol}
\usepackage{dcolumn}
\usepackage{bm}
\usepackage{multirow}
\usepackage{lscape}
\usepackage{setspace}
\usepackage{footmisc}
\makeatother

\usepackage{jf}
\usepackage{changepage}
\usepackage[latin1]{inputenc}
\usepackage[titletoc]{appendix}

\usepackage{caption,subfig}
\usepackage{bibunits} 
\defaultbibliography{Rmeasure}
\defaultbibliographystyle{jf}

\renewcommand\thetable{\Roman{table}}
\newcolumntype {d}[1]{D{.}{.}{#1}} \newcolumntype {.}{D{.}{.}{-1}}

\usepackage{color}
\definecolor {orange}{rgb}{1,0.5,0}

\usepackage{indentfirst}

\usepackage{xr-hyper}

\usepackage[hyperfootnotes=true]{hyperref}

\makeatletter
\newcommand*{\addFileDependency}[1]{
  \typeout{(#1)}
  \@addtofilelist{#1}
  \IfFileExists{#1}{}{\typeout{No file #1.}}
}
\makeatother

\newcommand*{\myexternaldocument}[1]{%
    \externaldocument{#1}%
    \addFileDependency{#1.tex}%
    \addFileDependency{#1.aux}%
}

\myexternaldocument{JFtemplate_internetappendix}

\usepackage{etoolbox}

\makeatletter

\patchcmd{\NAT@citex}
  {\@citea\NAT@hyper@{%
     \NAT@nmfmt{\NAT@nm}%
     \hyper@natlinkbreak{\NAT@aysep\NAT@spacechar}{\@citeb\@extra@b@citeb}%
     \NAT@date}}
  {\@citea\NAT@nmfmt{\NAT@nm}%
   \NAT@aysep\NAT@spacechar\NAT@hyper@{\NAT@date}}{}{}

\patchcmd{\NAT@citex}
  {\@citea\NAT@hyper@{%
     \NAT@nmfmt{\NAT@nm}%
     \hyper@natlinkbreak{\NAT@spacechar\NAT@@open\if*#1*\else#1\NAT@spacechar\fi}%
       {\@citeb\@extra@b@citeb}%
     \NAT@date}}
  {\@citea\NAT@nmfmt{\NAT@nm}%
   \NAT@spacechar\NAT@@open\if*#1*\else#1\NAT@spacechar\fi\NAT@hyper@{\NAT@date}}
  {}{}

\makeatother

\hypersetup{colorlinks=true,raiselinks=true,breaklinks=false,citecolor=blue}

\newtheorem{lemma}{LEMMA}
\newtheorem{theorem}{THEOREM}

\newtheorem{corollary}{COROLLARY}

\newtheorem{proposition}{PROPOSITION}

\date{}

\usepackage{titlesec}
\titleformat{\section}{\singlespace\centering\normalfont\Large\bfseries}{\centering \thesection .}{1em}{}
\titleformat{\subsection}{\singlespace\normalfont\large\itshape}{\thesubsection .}{1em}{}
\titleformat{\subsubsection}{\singlespace\centering\normalfont\itshape}{\thesubsubsection .}{0.5em}{}

\renewcommand{\thesubsection}{\Alph{subsection}}
\renewcommand{\thetable }{\Roman{table}}
\usepackage{lipsum}

\usepackage{amssymb,amsmath,amsfonts,eurosym,geometry,graphicx,caption,color,setspace,sectsty,comment,footmisc,caption,natbib,pdflscape,array}

\usepackage{mathrsfs}
\usepackage{epsfig,latexsym}

\usepackage{dcolumn}
\usepackage{color}
\usepackage{pstricks}
\usepackage{epstopdf}
\usepackage{bbm}
\usepackage{mathtools}
\usepackage{enumerate}
\usepackage{array}

\usepackage{xcolor}
\usepackage{framed}
\definecolor{shadecolor}{rgb}{1,0.945098,0}

\usepackage{bm}
\usepackage{dsfont}
\usepackage{booktabs}
\usepackage{tabularx}
\usepackage{multirow}
\usepackage{colortbl}
\usepackage{cleveref}
\allowdisplaybreaks

\usepackage{longtable}
\usepackage{multirow}

\usepackage{float}
\usepackage{cleveref}
\Crefname{equation}{}{}

\usepackage{lscape} 
\usepackage{url}

\usepackage{amsmath}
\usepackage{booktabs}
\usepackage{amssymb}

\usepackage{multirow}
\usepackage{color}

\usepackage{booktabs}
\usepackage{multirow}

\usepackage{geometry} 
\usepackage{lipsum}

\DeclareMathOperator{\E}{\mathbb{E}}
\newcommand{\q}{{\mathbb{Q}}}

\newcommand{\QT}{{\mathbb{Q}^T}}
\newcommand{\QS}{{\mathbb{Q}^S}}

\newcommand{\X}{F}
\newcommand{\Y}{G}
\newcommand{\Lc}{{\mathcal{L}}}
\newcommand{\LT}{{L_T^*}}
\newcommand{\Ls}{{L_s^*}}

\newcommand{\LH}{{L_H^*}}

\newcommand{\LsB}{{L_s}}

\newcommand{\LHB}{{L_H}}

\newcommand{\LsP}{{L_s^P}}

\newcommand{\LHT}{{\mathcal{L}_T^*(H)}}
\newcommand{\LHs}{{\mathcal{L}_s^*(H)}}
\newcommand{\LLHs}{{\mathcal{L}_s^*(H)}}
\newcommand{\LHt}{{\mathcal{L}_t^*(H)}}
\newcommand{\LHH}{{\mathcal{L}_H^*(H)}}
\newcommand{\LHTB}{{\mathcal{L}_T(H)}}
\newcommand{\LHsB}{{\mathcal{L}_s(H)}}

\newcommand{\LHTo}{{\mathcal{L}_{1T}^*(H)}}
\newcommand{\LHto}{{\mathcal{L}_{1t}^*(H)}}
\newcommand{\LHHo}{{\mathcal{L}_{1H}^*(H)}}
\newcommand{\LHso}{{\mathcal{L}_{1s}^*(H)}}

\newcommand{\LHsoP}{{\mathcal{L}_{1s}^P(H)}}

\newcommand{\LHTt}{{\mathcal{L}_{2T}^*(H)}}
\newcommand{\LHtt}{{\mathcal{L}_{2t}^*(H)}}
\newcommand{\LHst}{{\mathcal{L}_{2s}^*(H)}}

\newcommand{\LHstB}{{\mathcal{L}_{2s}(H)}}

\newcommand{\LHstP}{{\mathcal{L}_{2s}^P(H)}}

\newcommand{\GuH}{{\gamma}_u(H)}
\newcommand{\GsH}{{\gamma}_s(H)}

\newcommand{\p}{{\mathbb{P}}}
\newcommand{\pp}{{\mathbb{P}}}
\newcommand{\Qem}{{\mathbb{Q}^*}}
\newcommand{\Ree}{{\mathbb{R}^*}}
\newcommand{\Reeo}{{{\mathbb{R}}^{*}_1}}
\newcommand{\Reet}{{{\mathbb{R}}^{*}_2}}
\newcommand{\RTo}{{{\mathbb{R}}^{T}_1 }}
\newcommand{\RTohat}{{{\hat{\mathbb{R}}}^{T}_1 }}
\newcommand{\RTt}{{{\mathbb{R}}^{T}_2 }}
\newcommand{\RSo}{{{\mathbb{R}}^{S}_1 }}

\newcommand{\R}{{\mathbb{R}}}

\newcommand{\Rtt}{{\mathbb{R}_2}}

\newcommand{\RT}{{\mathbb{R}^{T} }}
\newcommand{\RS}{{\mathbb{R}^{S} }}
\newcommand{\RR}{{\mathbb{R}}}
\newcommand{\Ro}{{\mathbb{R}}}

\newcommand{\h}{{\mathbb{R}}}

\newcommand{\A}{{\mathbb{A}}}
\newcommand{\D}{{\mathrm{d}}}
\newcommand{\ue}{\mathrm{e}}
\newcommand{\F}{{\mathcal{F}}}
\newcommand{\I}{{\mathds{1}}}

\newcommand{\ZTone}{{W_{1s}^\RTo}}
\newcommand{\ZTtwo}{{W_{2s}^\RTo}}

\newcommand{\Wp}{{W_s^\p}}
\newcommand{\Wpo}{{W^\p}}
\newcommand{\WT}{{W_s^\RTo}}

\newcommand{\Wq}{{W_s^\q}}
\newcommand{\Wqo}{{W^\q}}

\newcommand{\Wr}{{W_s^\R}}
\newcommand{\Wro}{{W^\R}}
\newcommand{\Wpu}{{W_u^\p}}
\newcommand{\Wqu}{{W_u^\q}}

\newcommand\inv[1]{#1\raisebox{1.15ex}{$\scriptscriptstyle-\!1$}}

\usepackage{bm}
\usepackage{mathtools}
\newcommand{\mymat}[1]{{#1}}

\newcommand{\mybmm}[1]{\bm{#1}}

\usepackage{amsfonts}
\usepackage{mathtools}
\usepackage{soul}

\newcommand\mat[1]{\mathcal{#1}}

\title{\textbf{A Theory of Equivalent Expectation Measures for Contingent Claim Returns}}
\author{SANJAY K. NAWALKHA and XIAOYANG ZHUO\thanks{
Sanjay K. Nawalkha (\href{mailto:nawalkha@isenberg.umass.edu}{nawalkha@isenberg.umass.edu}) is with the Isenberg School of Management, University of Massachusetts, MA, USA. Xiaoyang Zhuo (\href{mailto:zhuoxy@bit.edu.cn}{zhuoxy@bit.edu.cn}) is with the School of Management and Economics, Beijing Institute of Technology, Beijing, China. We are grateful for the helpful comments of two anonymous referees, the Editor (Stefan Nagel), and especially the Associate Editor. We are also very grateful to Peter Carr and Darrell Duffie for their discussant comments at the BQE Lecture series seminar (NYU Tandon) and the CISDM Annual Research conference 2020 (UMass, Amherst), respectively. We gratefully acknowledge the comments of Fousseni Chabi-Yo, Jingzhi Huang, Jens Jackwerth, Nikunj Kapadia, Hossein Kazemi, Zhan Shi, Longmin Wang, Yongjin Wang, Hao Zhou, and the seminar participants at UMass (Amherst), BIT, Tsinghua PBCSF, FERM, and Asian Quantitative Finance Seminar. Xiaoyang Zhuo is grateful for support by the National Natural Science Foundation of China (No. 72001024), and Beijing Institute of Technology Research Fund Program for Young Scholars. We would like to dedicate this paper to the memory of Peter Carr, a giant in the field of financial engineering and derivatives. We have read \emph{The Journal of Finance}'s disclosure policy and have no conflicts of interest to disclose.} }
\begin{document}
\begin{bibunit}

\date{May 20, 2022}
\maketitle

\vspace{-1.5em}
{\hspace{7em}\large{\centering (\emph{Journal of Finance}, forthcoming)}}
\setcounter{page}{0}
\thispagestyle{empty}

\begin{abstract}
\noindent This paper introduces a dynamic change of measure approach for computing the analytical solutions of expected future prices (and therefore, expected returns) of contingent claims over a finite horizon. The new approach constructs hybrid probability measures called the ``equivalent expectation measures" (EEMs), which provide the physical expectation of the claim's future price until before the horizon date,  and serve as pricing measures on or after the horizon date. The EEM theory can be used for empirical investigations of both the cross-section and the term structure of returns of contingent claims, such as Treasury bonds, corporate bonds, and financial derivatives.
\end{abstract}


\newpage

\noindent The change of measure concept lies at the heart of asset pricing. Almost half-a-century ago, the no-arbitrage derivations of the call option price by \cite{black1973pricing} and \cite{Merton1973Theory} led to the discovery of the standard change of measure approach, which assigns risk-neutral or equivalent martingale probabilities to all future events.\footnote{See \cite{cox1976valuation} and \cite{harrison1979martingales} for the development of the change of measure approach.} Extending this static approach which changes the pricing measure only at the current time, we develop a \emph{dynamic} change of measure approach which changes the pricing measure at any given future horizon date $H$ between the current time $t$ and the claim's expiration date $T$. Using the new dynamic change of measure approach, this paper develops a theory for computing expected future prices, and therefore, expected returns of contingent claims over any finite horizon $H$, significantly broadening the range of applications of the change of measure approach in finance.

The theory developed in this paper can be used to answer the following kinds of questions. What is the expected return of a 10-year Treasury bond over the next two months under the $A_1(3)$ affine model of \cite{dai2000specification}?
What is the expected return of a 5-year, A-rated corporate bond over the next quarter under the stationary-leverage ratio model of \cite{collin2001credit}? What is the expected return of a 3-month call option on a stock over the next month under the SVJ model of \cite{pan2002jump} or under the CGMY model of \cite{carr2002fine}? What is the expected return of a 5-year interest rate cap over next six months, under the QTSM3 interest rate model of \cite{Ahn2002Quadratic}? While the analytical solutions of the current prices \emph{exist} under all of the above-mentioned models, the analytical solutions of the expected future prices---which are necessary inputs for computing the expected returns---over an arbitrary finite horizon $H$, \emph{do not exist} in the finance literature.

Estimating expected returns becomes difficult, if not impossible, for most buy-side firms and other buy-and-hold investors---who do not continuously rebalance and hedge their portfolios---without the analytical solutions of expected future prices of the financial securities in their portfolios. This problem is faced not just by sophisticated hedge fund managers who hold complex contingent claims, but also by numerous asset managers at mutual funds, banks, insurance companies, and pension funds, who hold the most basic financial securities, such as the U.S. Treasury bonds and corporate bonds. In an insightful study, \cite{becker2015reaching} illustrate this problem in the context of insurance companies, whose asset managers reach for higher (promised) yields, and not necessarily higher expected returns. They find that the ``yield-centered" holdings of insurance companies are related to the business cycle, being most pronounced during economic expansions. More generally, it is common to use metrics based upon yield, rating, sector/industry, optionality, and default probability, to make portfolio holding decisions by fixed income asset managers. Since expected returns are not modeled \emph{explicitly}, these metrics provide at best only rough approximations, and at worst counter-intuitive guidance about the ex-ante expected returns as shown by \cite{becker2015reaching}. This is quite unlike the equity markets, where much effort is spent both by the academics and the portfolio managers for obtaining the conditional estimates of the expected returns on stocks. Part of this widespread problem in the fixed income markets and the financial derivatives markets is due to the lack of a simple and parsimonious framework that can guide the conditional estimation of finite-horizon expected returns. This paper fills this gap in the finance theory.

The intuition underlying our theory can be exemplified with the construction of a simple discrete-time example to compute the expected future price of a European call option. We begin by constructing a hybrid equivalent probability measure $\Ro$, which remains the physical measure $\pp$ until before time $H$, and becomes the risk-neutral measure $\q$ on or after time $H$, for a specific future horizon $H$ between time $t$ and time $T$. By construction the $\Ro$ measure provides the physical expectation of the claim's time-$H$ future price until before time $H$, and serves as the pricing (or the equivalent martingale) measure on or after time $H$. Using this new measure, we show how to construct a binomial tree to obtain the expected future price the European call option $C$, which matures at time $T$ with strike price $K$, written on an underlying asset price process $S$. Given $C_T = \max (S_T - K, 0) = (S_T - K)^+$ as the terminal payoff from the call option and using a constant interest rate $r$, the future price of the call option at time $H$ can be computed under the $\q$ measure, as follows:
\begin{align*}
	C_H = \E_H^\q\left[\ue^{-r(T-H)}(S_T - K)^+\right].
\end{align*}
Taking the physical expectation of the future call price at the current time $t \leq H$, gives\footnote{We use $\E[\cdot]$ and $\E^\pp[\cdot]$ interchangeably to denote expectation under the physical measure $\pp$. In addition, we denote by $\E_t$ the conditional expectation with respect to the filtration $\left\{\mathcal{F}_t\right\}_{t\geq 0}$, that is,
	$\E_t[\cdot]=\E[\cdot|\mathcal{F}_t]$.}
\begin{align}\label{eq:callpriceat_t}
	\E_t[C_H] = \E_t^\pp\left[\E_H^\q\left[\ue^{-r(T-H)}(S_T - K)^+\right]\right].
\end{align}
Suppose we can force the construction of the new equivalent probability measure $\Ro$ by satisfying the following two conditions:
\begin{enumerate}[i)]
	\item the $\pp$ transition probabilities at time $t$ of all events until time $H$ are equal to the corresponding $\Ro$ transition probabilities of those events, and
	\item the $\q$ transition probabilities at time $H$ of all events until time $T$ are equal to the corresponding $\Ro$ transition probabilities of those events.
\end{enumerate}

Using this construction, the law of iterated expectations immediately gives:
\begin{align}\label{eq:callpriceat_t_under R}
	\E_t[C_H] = \E_t^\Ro\left[\ue^{-r(T-H)}(S_T - K)^+\right].
\end{align}

Note that for each value of the horizon date $H$---from the current time $t$ until the claim's maturity date $T$---equation \eqref{eq:callpriceat_t_under R} uses a different $\Ro$ measure specific to that horizon. The stochastic process under the $\Ro$ measure evolves under the physical measure $\pp$ until before time $H$, and under the risk-neutral measure $\q$ on or after time $H$.

Now, consider a discrete 2-year binomial tree illustrated in Figure \ref{fig:Rtree}, in which the current time $t=0$, $H=1$ years, $T=2$ years, and the 1-year short rate is a constant.


\begin{figure}[h!]
	\begin{center}
		\subfloat[\label{fig:subfig:Rtree}Stock price ]{	\includegraphics[width=0.53\textwidth]{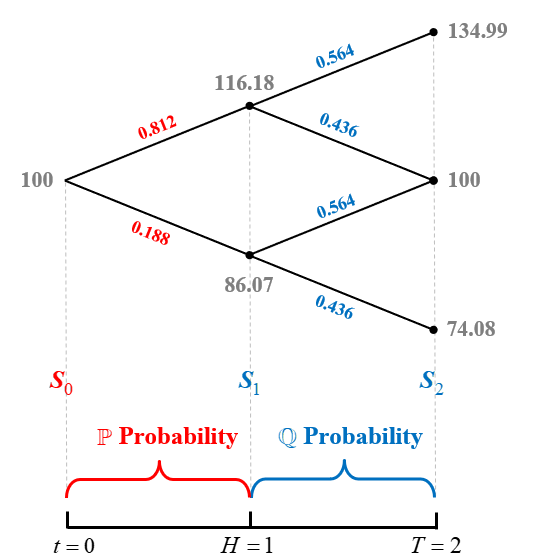}}
		\subfloat[\label{fig:subfig:Ctree}Expected call price]{	\includegraphics[width=0.53\textwidth]{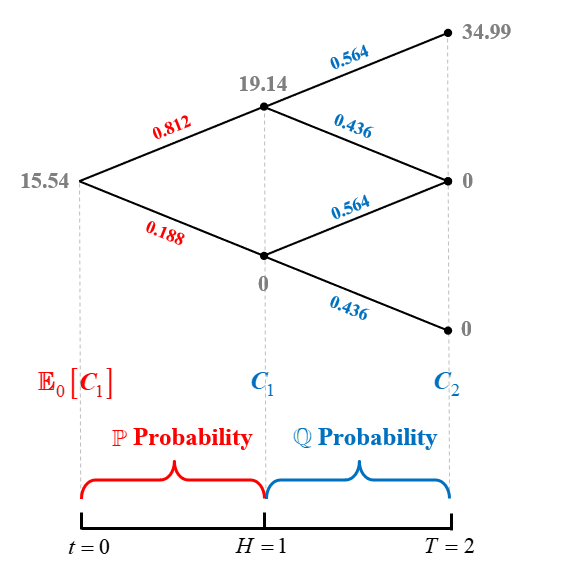}}
		\caption{\label{fig:Rtree}\textbf{The binomial tree under the $\mathbb{R}$ measure.} This figure shows a 2-year binomial tree structure for the stock price and the expected call price under the $\mathbb{R}$ measure.  The $\mathbb{R}$ measure follows physical $\mathbb{P}$ distribution until before $H=1$ years, and risk-neutral $\mathbb{Q}$ distribution beginning at time $H=1$ years. The following parameters are used: $S_t=100$, $K=100$, $\mu=0.1$, $r=0.03$, $\sigma=0.15$, $\Delta t =1$. The up and down node values of the future stock price evolution are computed using the binomial-tree method given by \cite{cox1979option}. The up and down $\mathbb{R}$ probabilities are the same as the corresponding up and down $\mathbb{P}$ probabilities at time 0, equal to 0.812 and 0.188, respectively; and the up and down $\mathbb{R}$ probabilities are the same as the corresponding up and down $\mathbb{Q}$ probabilities at time 1, equal to 0.564 and 0.436, respectively. }
	\end{center}
\end{figure}

To construct the discrete 2-year binomial tree for the stock price process under the $\Ro$ measure, assume that the current price $S_0$ is \$100, the stock's annualized expected return $\mu$ is 0.1, the annualized risk-free rate $r$ is 0.03, the annualized stock return volatility $\sigma$ is 0.15, and the discrete interval $\Delta t$ = 1 year. The stock prices at different nodes are identical under all three measures, $\pp$, $\q$, and $\Ro$, and are given by the binomial model of \cite{cox1979option}: the stock price either moves up by the multiplicative factor $u=\ue^{\sigma \sqrt{\Delta t}}=1.1618$, or moves down by the multiplicative factor $d=\ue^{-\sigma \sqrt{\Delta t}}=0.8607$ over each time step, as shown in Figure \ref{fig:subfig:Rtree}.

The $\Ro$ measure is constructed by satisfying conditions i) and ii) given above. To satisfy condition i), the next-period up and down $\Ro$ probabilities are equal to the corresponding up and down $\pp$ probabilities, respectively, at time $t=0$, and are calculated as:
\begin{align*}
	p^u = \frac{\ue^{\mu\Delta t}-d}{u-d}= \frac{\ue^{0.1\times 1}-0.8607}{1.1618-0.8607} = 0.812, \ \  p^d = 1-p^u = 0.188.
\end{align*}

To satisfy condition ii), the next-period up and down $\Ro$ probabilities are equal to the corresponding up and down $\q$ probabilities, respectively, at time $H=1$, and are calculated as:
\begin{align*}
	q^u = \frac{\ue^{r\Delta t}-d}{u-d}= \frac{\ue^{0.03\times 1}-0.8607}{1.1618-0.8607}= 0.564, \ \ q^d = 1- q^u = 0.436.
\end{align*}

Now, consider the computation of the expected future price $\E_0[C_1]$ of a 2-year European call option written on this stock with a strike price $K$ = \$100. The expected future call price can be computed using the $\Ro$ measure as shown in Figure \ref{fig:subfig:Ctree}.

The option prices $C_2$ at the terminal nodes of the tree are calculated by the payoff function $\max\left(S_2-K, 0\right)$. The option prices $C_1$ at time $H=1$ year are computed using risk-neutral discounting \cite[see][]{cox1979option}. Specifically, the option value $C_1$ at the up node is
\begin{align*}
	C_1^{u} & = \ue^{-r\Delta t}\left( q^u\times \$34.99+q^d\times\$0 \right) \\
	&= \ue^{-0.03\times 1}\left(0.564\times \$34.99+0.436\times\$0 \right)=\$19.14.
\end{align*}
Similar calculations give $C_1^{d}$ = \$0 at the down node. Then at the current time $t = 0$, the expected future call price $\E_0[C_1]$ is calculated by taking the expectation of call prices $C_1^{u}$ and $C_1^{d}$, using the up and down $\pp$ probabilities, as follows:
\begin{align*}
	\E_0[C_1] & = p^u\times C_1^{u} +p^d\times C_1^{d}  \\
	& = 0.812\times \$19.14 +0.188\times\$0 =\$15.54.
\end{align*}

The above 2-year binomial tree can be generalized to multiple periods (of any length) with any specific horizon $H$ between the current time and the option expiration date, by using the one-period up and down $\Ro$ transition probabilities as the corresponding $\q$ transition probabilities for all periods between time $H$ and time $T$, and as the corresponding $\pp$ transition probabilities for all periods between time $t$ and time $H$.

The above discrete-time example used the money market account as the numeraire and the risk-neutral measure $\q$ as the corresponding equivalent martingale measure for valuation on or after the horizon date $H$. The theory underpinning the above example can be generalized to continuous-time models that use alternative numeraires for the valuation of contingent claims. The mathematical intuition underlying this new theory can be understood as follows.

Given a numeraire $G$ with the corresponding equivalent martingale measure (EMM) $\Qem$ and the Radon-Nikod\'{y}m derivative process \emph{$L^*$}, we show how the expected future prices of contingent claims can be computed using one of two new measures, $\Ree$ and $\Reeo$, with Radon-Nikod\'{y}m derivative processes $\mathcal{L}^*$  and $\mathcal{L}_1^*$, respectively. Each of the new measures is pieced together using two different Radon-Nikod\'{y}m derivative processes, which break from each other at the horizon date $H$, so that the measure provides the physical expectation of the claim's time-$H$ future price until before time $H$, and serves as the pricing (or the equivalent martingale) measure on or after time $H$. This simple and novel approach is shown to be very useful in providing formulas of the expected future prices of contingent claims, which are of the same general form as the corresponding valuation formulas of these claims, requiring that only the forms of the underlying price processes be similar under the physical measure and the pricing measure.

Unlike the traditional EMMs, discounted prices are martingales under the new measures, $\Ree$ and $\Reeo$, only from time $H$ to time $T$, and not from time $t$ to time $H$. Since these new measures provide the physical expectation of the claim's time-$H$ future price until before time $H$, we denote these measures as \emph{equivalent expectation measures (EEMs)}. The EMMs of \cite{harrison1979martingales} and \cite{geman1995changes} are obtained as special cases of the corresponding EEMs, when $t = H$.

Of course, the dynamic change of measure is not the only method for computing the expected future price of a contingent claim. Two other methods which can be potentially used to compute the expected future price of a claim are the \emph{stochastic discount factor}  \cite[e.g., see][]{hansen1987role, cochrane1996cross} method and the \emph{numeraire change} method \cite[e.g., see][]{long1990numeraire, becherer2001numeraire,karatzas2007numeraire}. The future price of the call option at time $H$, can be expressed under these two methods as follows:
\begin{align}
	\label{eq:callpriceatHunderP2}
	C_H &=\E_H^\pp\left[M_T/M_H (S_T - K)^+\right] \\
	\label{eq:callpriceatHunderP3}
	&=\E_H^\pp\left[N_H/N_T (S_T - K)^+\right], \tag{\theequation \textit{a}}
\end{align}
where $M$ is the pricing kernel, $M_T/M_H$ is the future stochastic discount factor (SDF) at time $H$, and $N$ is the numeraire portfolio using which makes all of the discounted prices $\pp$ martingales as shown by \cite{long1990numeraire}. Taking iterated expectations under the $\pp$ measure, the time $t$ expectation of the future call price is given as
\begin{align}
	\label{eq:callpriceattunderP2}
	\E_t[C_H] &=\E_t^\pp\left[M_T/M_H (S_T - K)^+\right] \\
	\label{eq:callpriceattunderP3}
	&=\E_t^\pp\left[N_H/N_T (S_T - K)^+\right]. \tag{\theequation \textit{a}}
\end{align}

Equations \eqref{eq:callpriceat_t_under R}, \eqref{eq:callpriceattunderP2}, and \eqref{eq:callpriceattunderP3} provide three distinct but related methods for computing the expected future price of the call option. Section \ref{sec:model} generalizes the dynamic change of method given by equation \eqref{eq:callpriceat_t_under R} to hold with stochastic interest rates and alternative numeraire assets. While the generalized dynamic change of method is mathematically equivalent to the \emph{future} stochastic discount factor method given by equation \eqref{eq:callpriceattunderP2}, and the \emph{dynamic} change of numeraire method given by equation \eqref{eq:callpriceattunderP3}, the former method has a significant advantage over the latter two methods for the derivations of the analytical solutions of the expected future prices under virtually all contingent claim models that admit an analytical solution to their current price using an EMM.

Though this paper provides the first formal theory of the dynamic change of measure method, this is not the first paper to propose a change of measure at a future date. \citeauthor{rubinstein1984simple}'s (\citeyear{rubinstein1984simple}) formula for computing the expected return of a call option under the \cite{black1973pricing} model, and the valuation formulas for a forward-start option by \cite{rubinstein1991pay} and \cite{kruse2005pricing}, also use a change of measure
at a future date. %
These papers combine two pricing measures or EMMs, iteratively---one measure from time $H$ until time $T$ for doing valuation at time $H$, and another measure from the current time $t$ until time $H$ for computing the current expectation---which is a technique that works under a limited set of models. In contrast, our paper creates a new class of hybrid measures called EEMs that ensure the derivations of analytical solutions to the expected future prices of contingent claims under virtually all models in the finance literature that admit an analytical solution to their current price using the existing EMMs.

Given the central role of expected returns in empirical asset pricing and in investment management, the additional flexibility provided by the analytical solutions of the expected future prices of contingent claims over different future horizons can significantly enhance our understanding of the investment opportunity set available in the Treasury bond market, corporate bond market, and virtually all of the financial derivatives in the equity, interest rate, foreign exchange, commodity, credit risk, and other areas of finance. A general econometric implication of the EEM theory is that the analytical solutions of the expected returns over different horizons can allow applications of various empirical estimation methods, such as linear regression, the Markov Chain Monte Carlo method, the generalized method of moments \cite[see][]{hansen1982large}, and others for studying both the cross-section and the term structure of contingent claims returns in different markets.

For example, future research can evaluate the return prediction ability of different structural models for corporate bonds using model-implied\footnote{For example, see \cite{Merton1974ON}, \cite{longstaff2005corporate}, \cite{Leland1996Optimal}, and \cite{collin2001credit}.} analytical solutions of the expected returns, extending the empirical investigations of \cite{gebhardt2005cross}, \cite{lin2011liquidity}, \cite{bai2019common}, \cite{chung2019volatility}, and \cite{bali2021long}.
Similarly, future research can evaluate the return prediction ability of different dynamic term structure models (DTSMs) for Treasury bonds using the analytical solutions of the expected returns, extending the empirical investigations of \cite{dai2000specification, Dai2002Expectation}, \cite{Ahn2002Quadratic}, \cite{duffee2002term}, \cite{bansal2002term}, %
{\cite{hong2005nonparametric}}, \cite{collin2008identification}, and \cite{joslin2011new}. To the best of our knowledge, \cite{adrian2013pricing} is the only existing study that directly incorporates the bond return information for model estimation using regressions of one-period excess returns. The availability of the analytical solutions of the expected bond returns over different horizons can also provide new insights on the \emph{term structure} of bond returns. While it may not be possible to study the term structure of corporate bond returns due to paucity of data and liquidity issues over short horizons, the DTSMs can be adapted for Treasury bond return prediction by forcing them to match the term structure of returns instead of the term structure of zero-coupon yields.

In the options area, the availability of the analytical solutions of the expected option returns for different horizons may not only advance our knowledge about the term structure of option returns, but also provide additional insight into the debate on whether index options are mispriced. For example, based on very high negative index put option expected returns over a \emph{monthly} horizon, \cite{Coval2001Expected} and \cite{chambers2014index} argue for mispricing, while \cite{broadie2009understanding} and \cite{baele2019cumulative} argue for no mispricing. It is an open question whether using the analytical solutions of the expected option returns over different return horizons of 1 day, 2 day, 1 week, 2 weeks, up to one month, would add more insight on this ongoing empirical debate.

Last but not least, as an empirical application in the risk management area, the framework of this paper can also be extended for the computation of variance, covariance, and other higher order moments, and various tail risk measures used routinely for the risk and return analysis of fixed income portfolios and equity portfolios embedded with options.

This paper is organized as follows. Section \ref{sec:model} develops the theory of two general EEMs, $\Ree$ and  $\Reeo$, using the dynamic change of measure method. These two measures lead to three numeraire-specific EEMs given as:
i) the $\Ro$ measure from the $\Ree$ class using the money market account as the numeraire;
ii) the $\RTo$ measure from the $\Reeo$ class using the $T$-maturity pure discount bond as the numeraire; and
iii) the $\RSo$ measure from the $\Reeo$ class using any general asset $S$ other than the money market account and the $T$-maturity pure discount bond as the numeraire.
These three numeraire-specific EEMs allow the derivation of analytical solutions to the expected future prices of virtually all contingent claims in finance that are priced using the change of measure method.

Section \ref{sec:Rmeasure} illustrates the use of $\Ro$ measure with the examples of the \cite{black1973pricing} option pricing model and the affine term structure models of \cite{dai2000specification} and \cite{collin2008identification}. This section also presents new $R$-transforms based upon the $\Ro$ measure, which can be used for obtaining the analytical solutions of the expected future prices of contingent claims that are priced using the transform-based models of \cite{bakshi2000spanning}, \cite{Duffie2000Transform}, and \cite{Chacko2002Pricing}, and the CGMY  L\'{e}vy model of \cite{carr2002fine}.

Section \ref{sec:R1Tmeasure} illustrates the use of $\RTo$ measure with the examples of the \cite{Merton1973Theory} option pricing model and the structural debt pricing model of \cite{collin2001credit}. This section also extends the theoretical framework of \cite{breeden1978prices} to obtain the $\RTo$ density and the physical density of the underlying asset using non-parametric and semi-parametric empirical models of expected option returns.

Internet Appendix Sections \ref{app:someATSM}, \ref{sec:qtsm}, \ref{app:fsoexpected}, \ref{app:MertonVasicek}, and our ongoing/future research
\cite[see][]{nawalkha2022sharpe} derive the expected future prices of different contingent claims under various classes of models in the equity, interest rate, and credit risk areas using the $\Ro$ and the $\RTo$ measures.\footnote{In these Internet Appendix sections, we illustrate additional applications of the EEMs $\RR$ and $\RTo$ for obtaining the analytical solutions of the expected future prices of contingent claims under the specific models of \cite{Vasicek1977An}, \cite{Cox1985A}, the ${A_{1r}(3)}$ model of \cite{balduzzi1996simple} and \cite{dai2000specification},  the general quadratic model of \cite{Ahn2002Quadratic} with the specific example of the QTSM3 model, the \citeauthor{Merton1973Theory}'s (\citeyear{Merton1973Theory}) model with the \cite{Vasicek1977An} short rate process embedded in the \citeauthor{merton1973intertemporal}'s (\citeyear{merton1973intertemporal}) ICAPM framework, and the forward-start option model using the affine jump-diffusion processes of \cite{Duffie2000Transform}. Using the $R$-transforms presented in Section \ref{sec:Rmeasure}.\ref{sec:Rtransforms}, \cite{nawalkha2022sharpe} derive the analytical solutions of the expected future prices of European options under the stochastic-volatility-double-jump model of \cite{Duffie2000Transform}, the CGMY  L\'{e}vy model of \cite{carr2002fine}, and interest rate derivatives under the affine jump-diffusion-based models of \cite{Chacko2002Pricing}.} Internet Appendix Section \ref{iapp:margrabe} illustrates the use of $\RSo$ measure with an extension of the \cite{margrabe1978value} exchange option pricing model with stochastic interest rates. We show that while the valuation of exchange option does not depend upon interest rates under the \cite{margrabe1978value} model, the expected future price of the exchange option does depend upon the parameters of the interest rate process.

The final section considers some theoretical extensions of the EEM theory and provides concluding remarks. Some of the proofs of the theorems and the propositions are provided in the Appendix, and some of the models are presented in the Internet Appendix.

\begin{center}
\section{{The Equivalent Expectation Measure Theory}}\label{sec:model} \vspace{-1em}
\end{center}

This section builds the theory of two general EEMs $\Ree$ and $\Reeo$. Identifying these general EEMs with specific numeraires, such as the money market account, the $T$-maturity pure discount bond, and a general traded asset, leads to different numeraire-specific EEMs, which can be used to obtain expected future prices of contingent claims under virtually all models that admit an analytical solution to the claim's current price.

\subsection{The $\mathbb{R}^*$ Measure}
Let the price of a contingent claim be denoted by $\X$, with a maturity/expiration date $T$ on which the claim pays out a single contingent cash flow. The cash flow can depend upon the entire history of the underlying state variables in arbitrary ways  \citep[as in the papers of][]{harrison1979martingales, bakshi2000spanning, Duffie2000Transform, Chacko2002Pricing}, allowing a variety of regular and exotic contingent claims, such as European options, Asian options, barrier options, and other claims, which make a single payment at a fixed terminal date $T$.

We fix a probability space $\left(\Omega, \mathcal{F}, \mathbb{P}\right)$ and a filtration $\mathcal{F}_s$, $0\leq s\leq T$, satisfying the usual conditions,\footnote{$\mathcal{F}_s$ is right-continuous and $\mathcal{F}_0$ contains all the $P$-negligible events in $\mathcal{F}$.} where $\mathbb{P}$ is the physical probability measure. We assume that the contingent claim is being valued by an EMM $\Qem$ defined on $\left(\Omega, \mathcal{F} \right)$, which uses a traded asset $\Y$ as the numeraire for valuing the claim.\footnote{The numeraire $\Y$ could be the money market account, zero-coupon bond, or any other asset, for instance, stock, commodity, etc.} According to the martingale valuation results \citep[see, e.g.,][Theorem 26.2]{geman1995changes,bjork2009arbitrage}, the process $\X_s/\Y_s$ is a martingale under $\Qem$ for all $0\leq s\leq T$, and the current time $t$ price of the claim is given as
\begin{align}\label{eq:xtprice}
	\X_t = \Y_t\E_t^\Qem\left[\frac{\X_T}{\Y_T}\right].
\end{align}

Due to the equivalence of $\Qem$ and $\pp$, there exists an almost surely positive random variable $\LT$ such that
\begin{align*}
	\Qem(A) = \int_A \LT(\omega) \D \pp(\omega) \ \   \text{for all}  \ \ A \in \mathcal{F}_T,
\end{align*}
where $\LT$ is the Radon-Nikod\'{y}m derivative of $\Qem$ with respect to $\pp$, and we can write it as
\begin{align*}\LT &=\left.\frac{\D \Qem}{\D \pp}\Big|_{\mathcal{F}_T}, \right.
\end{align*}
with $\E[\LT]=1$. The Radon-Nikod\'{y}m derivative process is defined as
\begin{align*}%
	\Ls = \E\left[\LT \middle| \mathcal{F}_s\right], \ \ \text{or} \ \ \Ls =\frac{\D \Qem}{\D \pp}\Big|_{\mathcal{F}_s}, \ \ \text{for} \ \ 0\leq s\leq T.
\end{align*}

The \emph{expected future price} of the claim, $\E_t \left[ \X_H\right]$, can be solved using equation \eqref{eq:xtprice} at a future time $H$, as follows:
\begin{align}\label{eq:priceQ}
	\E_t \left[ \X_H\right] &=  \E_t^\pp \left[ \Y_H\E_H^\Qem\left[\frac{\X_T}{\Y_T}\right]\right].
\end{align}

Similar to equation \eqref{eq:callpriceat_t} in the introduction, the main reason for {the intractability of equation} \eqref{eq:priceQ} is that the law of iterated expectations cannot be used since the inner and the outer expectations are not under the same probability measure. In order to make the inner and the outer expectations under the same probability measure, we define an equivalent probability measure $\Ree$ in {the following Lemma \ref{lemma:Ree}}, under which the following two conditions are met for all $0 \leq t \leq H \leq T$:
\begin{enumerate}
	\item[\textbf{C1.}] The $\pp$ transition probabilities at time $t$ of all events until time $H$ are equal to the corresponding $\Ree$ transition probabilities of those events.
	\item[\textbf{C2.}] The $\Qem$ transition probabilities at time $H$ of all events until time $T$ are equal to the corresponding $\Ree$ transition probabilities of those events.
\end{enumerate}

If both the above conditions are satisfied, then equation \eqref{eq:priceQ} can be simplified using iterated expectations under the $\Ree$ measure as follows:
\begin{align}\label{eq:priceQIT}
	\E_t \left[ \X_H\right] &=  \E_t^\Ree \left[\Y_H\E_H^\Ree\left[\frac{\X_T}{\Y_T}\right]\right] \nonumber \\
	&=  \E_t^\Ree \left[\E_H^\Ree\left[\Y_H\frac{\X_T}{\Y_T}\right]\right] \nonumber \\
	& = \E^\Ree_t \left[\Y_H \frac{\X_T}{\Y_T} \right].
\end{align}

To this end, we construct the $\Ree$ measure using the following lemma:
\begin{lemma}\label{lemma:Ree}
	For a fixed $H$ with $0\leq H\leq T$, define a process $\LLHs$ as
	\begin{align}\label{eq:LR}
		\LLHs  &=\left\{\begin{array}{ll}
			\displaystyle \frac{\Ls}{\LH}, &  \text{if} \ \  {H\leq s\leq T}, \\
			\displaystyle  1, &  \text{if } \ \ {0\leq s < H}.
		\end{array}\right.
	\end{align}
	Let
	\begin{align*}%
		\Ree(A) &= \int_A \mathcal{L}_T^*(H; \omega) \D \pp(\omega) \ \   \text{for all}  \ \ A \in \mathcal{F}_T,
	\end{align*}
	then $\Ree$ is a probability measure equivalent to $\pp$,  and $\LLHs$ is the Radon-Nikod\'{y}m derivative process of $\Ree$ with respect to $\pp$.
\end{lemma}

\begin{proof}
	See Appendix \ref{app:thm1} (Section \ref{app:thm1}).
\end{proof}

We now state the main theorem of this paper.
\begin{theorem}\label{thm:main}
	For all $0\leq t\leq H\leq T$, the time $t$ expectation of the claim's future price $F_H$ can be computed by using the equivalent expectation measure $\Ree$ defined by Lemma \ref{lemma:Ree}, as shown in equation \eqref{eq:priceQIT}.
\end{theorem}
\begin{proof}
	Since equation \eqref{eq:priceQIT} follows from equation \eqref{eq:priceQ} if conditions C1 and C2 are satisfied, the proof requires us to show that the $\Ree$ measure defined in Lemma \ref{lemma:Ree} satisfies both these conditions. For the first condition C1, note that the Radon-Nikod\'{y}m derivative process of $\Ree$ w.r.t. $\p$ equals 1 before time $H$, which implies that $\Ree$ = $\p$, before $H$. For the second condition C2, note that for any event, say $A_T \in \F_T$, its conditional $\Qem$ probability at time $H$ is the same as the corresponding $\Ree$ probability as follows:
	\begin{align*}
		\E_H^\Qem\left[\I_{A_T} \right]&=\E_H^\p\left[\frac{\LT}{\LH}\I_{A_T}\right]  =\E_H^\p\left[\frac{\LT/\LH}{\LH/\LH}\I_{A_T}\right]	=\E_H^\p\left[\frac{\LHT}{\LHH}\I_{A_T}\right] =\E^\Ree_H\left[\I_{A_T}\right].
	\end{align*}
	
	As an alternative proof we can derive equation \eqref{eq:priceQIT} directly from equation \eqref{eq:priceQ} by using Lemma \ref{lemma:Ree} (and without explicitly invoking conditions C1 and C2, which are implied by Lemma \ref{lemma:Ree}) as follows:
	\begin{align}\label{eq:thmprf}
		\begin{aligned}
			\E_t \left[ \X_H\right] &= \E^\pp_t\left[\E_H^\Qem\left[ \Y_H \frac{\X_T}{\Y_T} \right]\right] = \E^\pp_t\left[\E_H^\pp\left[\frac{\LT}{\LH}\cdot \Y_H \frac{\X_T}{\Y_T} \right]\right] \\
			& = \E^\pp_t\left[\frac{\LT}{\LH}\cdot \Y_H \frac{\X_T}{\Y_T} \right]  = \E^\pp_t\left[\frac{\LT/\LH}{1}\cdot \Y_H \frac{\X_T}{\Y_T} \right] \\
			& = \E^\pp_t\left[\frac{\LHT}{\LHt}\cdot \Y_H \frac{\X_T}{\Y_T}\right] = \E^\Ree_t \left[ \Y_H \frac{\X_T}{\Y_T} \right].
		\end{aligned}
	\end{align}
\end{proof}

We now give some general properties of the $\Ree$ measure that provide insights into its relationship with the physical measure $\pp$ and the equivalent martingale measure $\Qem$. We begin with Proposition \ref{prop:Ree1}, which gives the relationship of the $\Ree$ measure with the $\pp$ and $\Qem$ measures.

\begin{proposition}\label{prop:Ree1}
	For all $0\leq t\leq H\leq s \leq T$, the $\Ree$ measure has the following properties:
	\begin{enumerate}[(i)]
		\item $\Ree(A\,|\,\F_{t}) = \pp(A\,|\,\F_{t})$ for all $A\in \mathcal{F}_{H}$. \label{Ree:num1}
		\item  $\Ree(A\,|\,\F_{s}) = \Qem(A\,|\,\F_{s}) $ for all $A\in \mathcal{F}_{T}$. \label{Ree:num2}
		\item When $H=T$, $\Ree(A\,|\,\F_{t}) = \pp(A\,|\,\F_{t})$  for all $A\in \mathcal{F}_T$. \label{Ree:num3}
		\item When $H=t$, $\Ree(A\,|\,\F_{t})  = \Qem(A\,|\,\F_{t}) $ for all $A\in \mathcal{F}_T$. \label{Ree:num4}
		\item  $\E_t^\Ree\left[Z_T\right]  = \E_t^\p\left[\E_H^\Qem\left[Z_T\right]\right]$  for any random variable $Z_T$ at time $T$. \label{Ree:num5}
	\end{enumerate}
\end{proposition}
\begin{proof}
	See Appendix \ref{app:thm1} (Section \ref{app:Reempro}).
\end{proof}

According to Proposition \ref{prop:Ree1}\eqref{Ree:num1}, the conditional $\Ree$ probabilities at any time $t$ until time $H$, of the events at time $H$, are the same as the corresponding $\pp$ probabilities of those events, satisfying condition C1 given earlier. According to  Proposition \ref{prop:Ree1}\eqref{Ree:num2}, the conditional $\Ree$ probabilities at any time $s$ on or after time $H$, of the events at time $T$, are the same as the corresponding conditional $\Qem$ probabilities of those events, satisfying condition C2 given earlier. Propositions \ref{prop:Ree1}\eqref{Ree:num3} and \ref{prop:Ree1}\eqref{Ree:num4} show that the $\Ree$ measure nests both the physical measure $\pp$ and the equivalent martingale measure $\Qem$. When $H = t$, the $\Ree$ measure becomes the $\Qem$ measure, and the expected value computing problem becomes the traditional claim pricing problem. When $H=T$, the $\Ree$ measure becomes the $\pp$ measure. For all other values of $H$ that are strictly between the current time $t$ and the maturity/expiration date $T$, the $\Ree$ measure remains a hybrid measure. Proposition \ref{prop:Ree1}\eqref{Ree:num5} shows that the expectation of any random variable $Z_T$ under the $\Ree$ measure is equivalent to the iterated expectation of that variable computed using both $\pp$ and $\Qem$ measures.

\begin{proposition}\label{prop:Ree2}
	The expectation of the price ratio process $\X_s/\Y_s$ has the following properties:
	\begin{enumerate}[(i)]
		\item For all $0\leq H\leq s_1\leq s_2\leq T$, $ \E_{s_1}^\Ree\left[{\X_{s_2}}/{\Y_{s_2}}\right]  =   \E_{s_1}^\Qem\left[{\X_{s_2}}/{\Y_{s_2}}\right] ={\X_{s_1}}/{\Y_{s_1}}$. \label{Ree2:num1}
		\item For all $0\leq s_1\leq s_2\leq H\leq T$, $ \E_{s_1}^\Ree\left[{\X_{s_2}}/{\Y_{s_2}}\right]  =   \E_{s_1}^\pp\left[{\X_{s_2}}/{\Y_{s_2}}\right] \neq     \E_{s_1}^\Qem\left[{\X_{s_2}}/{\Y_{s_2}}\right] ={\X_{s_1}}/{\Y_{s_1}}$. \label{Ree2:num2}
	\end{enumerate}
\end{proposition}

\begin{proof}
	See Appendix \ref{app:thm1} (Section \ref{app:prop2}).
\end{proof}

According to Proposition \ref{prop:Ree2}\eqref{Ree2:num1}, the price ratio process $\X_s/\Y_s$ is a martingale under $\Ree$ from time $H$ to time $T$. However, according to Proposition \ref{prop:Ree2}\eqref{Ree2:num2}, unless $\pp(A)$ = $\Qem(A)$ for all $A\in \mathcal{F}_H$, the price ratio process $\X_s/\Y_s$ is not a martingale under $\Ree$ from time 0 to time $H$. This shows the essential difference between the equivalent $martingale$ measure $\Qem$ under which the price ratio $\X_s/\Y_s$ is a martingale from time 0 until time $T$, and the corresponding equivalent $expectation$ measure $\Ree$ under which $\X_s/\Y_s$ is a martingale only from time $H$ until time $T$.

Generalizing the results in equations \eqref{eq:callpriceattunderP2} and \eqref{eq:callpriceattunderP3}, the expected future price of the claim in equation \eqref{eq:priceQIT} can also be expressed as follows:
\begin{align}
	\label{eq:callpriceattunderP2G}
	\E_t[\X_H] &=\E_t^\pp\left[\frac{M_T}{M_H} \X_T\right] \\
	\label{eq:callpriceattunderP3G}
	&=\E_t^\pp\left[\frac{N_H}{N_T} \X_T\right], \tag{\theequation \textit{a}}
\end{align}
where
the future SDF $M_T/M_H$ equals $N_H/N_T$, the inverse of the future gross return on the numeraire portfolio from time $H$ to time $T$ \cite[see][]{long1990numeraire}, for all $0\leq t\leq H \leq T$.
The equivalence relationship for the expected future price of a contingent claim given in equations \eqref{eq:priceQIT}, \eqref{eq:callpriceattunderP2G}, and \eqref{eq:callpriceattunderP3G} connects the new dynamic change of measure method to two asset pricing literatures in finance on the SDF-based models and the numeraire change-based models. If the numeraire $G$ is selected as the numeraire portfolio $N$ of \cite{long1990numeraire}, then both $\Qem$ and $\Ree$ become the physical measure $\pp$, and equation \eqref{eq:callpriceattunderP3G} obtains as a special case of equation \eqref{eq:priceQIT}. If the numeraire $G$ is selected as the money market account, then $\Qem$ becomes the risk-neutral measure $\q$, and $\Ree$ becomes the $\Ro$ measure illustrated in the introduction.

To evaluate the usefulness of the dynamic change of measure method relative to the future SDF method given by equation \eqref{eq:callpriceattunderP2G} and the dynamic change of numeraire method given by equation \eqref{eq:callpriceattunderP3G}, first consider those contingent claim models under which $G_H/G_T$ is distributed independently of $\X_T$, such that the expectation in equation \eqref{eq:priceQIT} can be given as a product of two expectations that are easier to derive and compute. For example, consider all models which use the money market account as the numeraire, and assume a constant or deterministic short rate.
The expectation in equation \eqref{eq:priceQIT} under these models depends upon only the distribution of $\X_T$ under the $\R$ measure, since $G_H/G_T$ is a constant.\footnote{By definition, $\Ree=\R$, when the numeraire is the money market account.}  In contrast, the expectations using equations \eqref{eq:callpriceattunderP2G} and \eqref{eq:callpriceattunderP3G} depend upon the interaction of $\X_T$ with either $M_T/M_H$ or with $N_H/N_T$ under the $\pp$ measure. Since $\X_T$ is almost never independent of either $M_T/M_H$ or $N_H/N_T$, the analytical solutions are easier to derive using the dynamic change of measure method given by equation \eqref{eq:priceQIT}. This advantage of the dynamic change of measure method is especially relevant when the underlying state variable processes follow complex stochastic processes, such as under the constant-elasticity-of-variance models of \cite{cox1976valuation}, the stochastic volatility models of \cite{hull1987pricing} and \cite{heston1993closed}, and the stochastic-volatility-jump models of \cite{Duffie2000Transform} and \cite{pan2002jump}. Intuitively, the same reasoning explains the widespread adoption of the risk-neutral measure $\q$ for contingent claims valuation, even though the SDF-based approach and the numeraire change approach can also be used.

Next, consider those contingent claim models under which $G_H/G_T$ is stochastic and $\X_T$ is constant in equation \eqref{eq:priceQIT}. These models comprise the dynamic term structure models for pricing default-free bonds with the numeraire $G$ as the money market account. Since a large empirical literature exists on these term structure models which estimates the state variables processes under both $\pp$ and $\q$ measures, the state variable processes under the $\Ro$ measure are readily available. Moreover, the techniques to derive the analytical solutions of the expected future price using the dynamic change of measure method given by equation \eqref{eq:priceQIT} follow as straightforward extensions of the existing techniques given by \cite{dai2000specification}, \cite{Ahn2002Quadratic}, and \cite{collin2008identification}, as shown in Section \ref{sec:Rmeasure}.\ref{sec:atsm}.

Finally, consider the contingent claim models under which $G_H/G_T$ and $\X_T$ are not independently distributed. For these models we propose the $\Reeo$ measure in the next section, which breaks the expectation of the future price of the claim in equation \eqref{eq:priceQIT} into a product of two expectations. Doing this simplifies the derivation of the analytical solution of the expected future price of the claim considerably under models that use numeraires other than the money market account for pricing the claim.

\subsection{The ${\mathbb{R}}^{*}_1$ Measure}
Consider separating the single expectation of the expression $\Y_H \X_T/\Y_T$ in equation \eqref{eq:priceQIT} into a product of \emph{two} expectations, and define another equivalent expectation measure $\Reeo$, as follows:
\begin{align}\label{eq:priceQIT1}
	\E_t\left[ \X_H\right]&= \E^\Ree_t\left[\Y_H \right] \E_t^\Reeo\left[\frac{\X_T}{\Y_T} \right] \nonumber \\
	&= \E^\p_t \left[\Y_H   \right] \E_t^\Reeo\left[\frac{\X_T}{\Y_T}  \right],
\end{align}
where for any $\mathcal{F}_T$-measurable variable $Z_T$, $\Reeo$ is defined with respect to $\Ree$ as follows:
\begin{align}\label{eq:QIT12def}
	\E_t^{\Reeo}[Z_T] = \E_t^\Ree\left[\frac{\Y_H}{\E_t^\Ree\left[\Y_H\right]} \cdot Z_T\right].
\end{align}

The existence of the $\Ree$ measure guarantees the existence of the $\Reeo$ measure as defined in equation \eqref{eq:QIT12def}. We construct the $\Reeo$ measure with respect to $\p$, using the following lemma.

\begin{lemma}\label{lemma:reeo}
	For a fixed $H$ with $0\leq H\leq T$, define a process $\LHso$ as
	\begin{align*}
		\begin{aligned}
			\LHso &  =\left\{\begin{array}{ll}
				\displaystyle \frac{\Y_H}{\E_0^\p\left[\Y_H\right]}\cdot \frac{\Ls}{\LH}, &  \text{if} \ \  {H\leq s\leq T}, \\
				\vspace{-0.8em}& \\
				\displaystyle  \frac{\E_s^\p\left[\Y_H\right]}{\E_0^\p\left[\Y_H\right]}, &    \text{if } \ \ {0\leq s<H}.
			\end{array}\right.
		\end{aligned}
	\end{align*}
	Let
	\begin{align*}%
		\Reeo(A) \triangleq \int_A \mathcal{L}_{1T}^*(H; \omega)\D \pp(\omega) \ \ \text{for all $A \in \mathcal{F}_T$},
	\end{align*}
	then $\Reeo$ is a probability measure equivalent to $\pp$,  and $\LHso$ is the Radon-Nikod\'{y}m derivative process of $\Reeo$ with respect to $\pp$.
\end{lemma}
\begin{proof}
	See Appendix \ref{app:thm1} (Section \ref{app:thm2}).
\end{proof}

The next theorem proves that the definition of the $\Reeo$ measure given in Lemma \ref{lemma:reeo} is consistent with the expectation result given by equation \eqref{eq:priceQIT1}.

\begin{theorem}\label{thm:main2}
	For all $0\leq t\leq H\leq T$, the time $t$ expectation of the claim's future price $F_H$ can be computed by using the equivalent expectation measure $\Reeo$ defined by Lemma \ref{lemma:reeo}, as shown in equation \eqref{eq:priceQIT1}.
\end{theorem}

\begin{proof}
	The proof of equation \eqref{eq:priceQIT1} follows from the third equality in equation \eqref{eq:thmprf} and applying Lemma \ref{lemma:reeo}, as follows:
	\begin{align*}%
		\begin{aligned}
			\E_t \left[ \X_H\right] &= \E^\pp_t\left[\frac{\LT}{\LH}\cdot \Y_H \frac{\X_T}{\Y_T} \right]  = \E_t^\p\left[\Y_H\right]\E^\pp_t\left[\frac{\Y_H}{\E_t^\p\left[\Y_H\right]}\cdot\frac{\LT}{\LH}\cdot  \frac{\X_T}{\Y_T} \right] \\
			& = \E_t^\p\left[\Y_H\right]\E^\pp_t\left[\frac{\Y_H/\E_0\left[\Y_H\right]\cdot \LT/\LH}{\E_t^\p\left[\Y_H\right]/\E_0\left[\Y_H\right]} \frac{\X_T}{\Y_T}\right] \\
			& =  \E_t^\p\left[\Y_H\right] \E^\pp_t\left[\frac{\LHTo}{\LHto}\cdot \frac{\X_T}{\Y_T} \right] \\
			& =  \E_t^\p\left[\Y_H\right]\E^\Reeo_t \left[ \frac{\X_T}{\Y_T} \right].
		\end{aligned}
	\end{align*}
	
	In addition, it is easy to verify that the mathematical construction of the $\Reeo$ measure in Lemma \ref{lemma:reeo} is consistent with its proposed definition in equation \eqref{eq:QIT12def} as follows.
	
	For any $\mathcal{F}_T$-measurable variable $Z_T$, we have
	\begin{align}\label{eq:RandR1}
		\begin{aligned}
			\E_t^{\Reeo}\left[Z_T\right] &= \E_t^\p\left[\frac{\LHTo}{\LHto}Z_T\right] =  \E_t^\p\left[\frac{\Y_H/\E_0\left[\Y_H\right]\cdot \LT/\LH}{\E_t^\p\left[\Y_H\right]/\E_0\left[\Y_H\right]}Z_T\right]\\
&=\E_t^\p\left[\frac{\LT}{\LH}\cdot\frac{\Y_H}{\E_t^\p\left[\Y_H\right]}Z_T\right] =\E_t^\p\left[\frac{\LHT}{\LHt}\cdot\frac{\Y_H}{\E_t^\p\left[\Y_H\right]}Z_T\right]\\
  &= \E_t^\Ree\left[\frac{\Y_H}{\E_t^\p\left[\Y_H\right]}  Z_T\right]= \E_t^\Ree\left[\frac{\Y_H}{\E_t^\Ree\left[\Y_H\right]}  Z_T\right],
		\end{aligned}
	\end{align}
	where the last equality follows since $\E_t^\Ree\left[\Y_H\right]=\E_t^\p\left[\LHH/\LHt\cdot\Y_H\right]=\E_t^\p\left[\Y_H\right]$.
\end{proof}

Internet Appendix Section \ref{app:DiffEEM}.\ref{app:propofReeo} gives some general properties of the $\Reeo$ measure that provide additional insights into its relationship with the physical measure $\pp$ and the equivalent martingale measure $\Qem$. According to these properties, the conditional $\Reeo$ probabilities at any time $s$ on or after time $H$, of the events at time $T$, are the same as the corresponding conditional $\Ree$ probabilities, as well as the conditional $\Qem$ probabilities of those events. However, the conditional $\Reeo$ probabilities at any time $t$ until time $H$, of any events at time $H$, are not the same as the corresponding $\pp$ probabilities of those events, due to an adjustment term equal to $\Y_H/{\E_t^\p\left[\Y_H\right]}$, related to the numeraire $\Y$.

The separation of the expectation of the expression $\Y_H \X_T/\Y_T$ in equation \eqref{eq:priceQIT} into a product of two expectations in equation \eqref{eq:priceQIT1} is useful for obtaining the analytical solution of the expected future price of a claim when $\Y_H/\Y_T$ and $\X_T$ are not distributed independently of each other.\footnote{Of course, equation \eqref{eq:priceQIT1} is not the only way to separate the expectation of the expression $\Y_H \X_T/\Y_T$ in equation \eqref{eq:priceQIT} into a product of two or more expectations. Another way may be to separate this expectation into a product of the expectations of $\Y_H/\Y_T$ and $\X_T$ using another EEM measure obtained as yet another transformation of the $\Ree$ measure, similar to the transformation shown in equation \eqref{eq:QIT12def}. This can be done as follows: Let
	$\E_t \left[ \X_H\right]=\E^\Ree_t \left[\frac{\Y_H}{\Y_T} \right] \E^{\mathbb{R}^{*}_2}_t\left[\X_T \right],$ where for any $\mathcal{F}_T$-measurable variable $Z_T$, $\Reet$ is defined with respect to $\Ree$ as follows: $$\E_t^{\Reet}[Z_T]=\E_t^\Ree\left[\frac{\Y_H/\Y_T}{\E_t^\Ree\left[\Y_H/\Y_T\right]}\cdot Z_T\right].$$
	We show in Internet Appendix Section \ref{app:DiffEEM}.\ref{sec:Reet} that if the numeraires are restricted to be either the money market account or the pure discount bond, then the $\Reeo$ measure subsumes the $\Reet$ measure. A final way may be to separate this expectation into a product of three expectations of $\Y_H$, $\X_T$, and $1/\Y_T$ using even more EEMs. However, we find that the expectations given under the two EEMs $\Ree$ and $\Reeo$ by equations \eqref{eq:priceQIT} and \eqref{eq:priceQIT1}, respectively, are \emph{sufficient} for obtaining the analytical solutions of the expected future prices of contingent claims under all models in finance that admit analytical solutions to the current prices of these claims using the change of measure method.} We demonstrate applications of equation \eqref{eq:priceQIT1} using the \citeauthor{Merton1973Theory}'s (\citeyear{Merton1973Theory}) stochastic interest rate-based option pricing model and \citeauthor{collin2001credit}'s (\citeyear{collin2001credit}) corporate bond pricing model in Section \ref{sec:R1Tmeasure}, and \citeauthor{margrabe1978value}'s (\citeyear{margrabe1978value}) asset-exchange option model in Internet Appendix Section \ref{iapp:margrabe}.

Under many contingent claim models, such as \cite{Merton1973Theory},  \cite{margrabe1978value}, and \cite{Grabbe1983The}, the solution of $\E_t^\Reeo\left[{\X_T}/{\Y_T}  \right]$ in equation \eqref{eq:priceQIT1}
requires the distribution of the normalized asset price ${S_T}/{\Y_T}$, under the $\Reeo$ measure, where $S_T$ is the time-$T$ price of the asset underlying the claim. The expectation of ${S_T}/{\Y_T}$ under the $\Reeo$ measure can be solved by substituting $\X_T$ = $S_T$ in equation \eqref{eq:priceQIT1}, as follows:
\begin{align}\label{eq:priceQIT1R}
	\E_t^\Reeo\left[\frac{S_T}{\Y_T} \right] = \frac{\E^\pp_t\left[S_H\right]}{\E^\pp_t\left[\Y_H \right]}.
\end{align}
The derivations of the analytical solutions to the expected future prices of contingent claims
are simplified considerably using the result in equation \eqref{eq:priceQIT1R}, since the physical expectations of $S_H$ and $G_H$ are typically available in closed-form.

The $\Ree$ measure given by Lemma \ref{lemma:Ree} defines three numeraire-specific EEMs: $\Ro$, $\RT$, and $\RS$, corresponding to the three numeraires: the money market account, the $T$-maturity pure discount bond, and any tradable asset $S$ other than the money market account and the $T$-maturity pure discount bond, respectively.
Similarly, the $\Reeo$ measure given by Lemma \ref{lemma:reeo} defines three numeraire-specific EEMs: $\Ro_1$, $\RTo$, and $\RSo$, corresponding to the above three numeraires, respectively.
Out of these six EEMs, we find that only three EEMs are useful for obtaining the analytical solutions of the expected future prices of contingent claims, given as follows:

\begin{enumerate}[i)]
	\item  the $\Ro$ measure from the $\Ree$ class, obtained as a generalization of the risk-neutral measure $\q$ given by \cite{black1973pricing}, \cite{cox1976valuation}, and \cite{harrison1979martingales};
	\item the $\RTo$ measure from the $\Reeo$ class, obtained as a generalization of the forward measure $\QT$ given by \cite{Merton1973Theory},\footnote{The call price expectation under the $\QT$ measure is consistent with the \citeauthor{Merton1973Theory}'s (\citeyear{Merton1973Theory}) PDE for the call price using the Feynman-Kac theorem.} \cite{geman1989importance}, \cite{jamshidian1989exact}, and \cite{geman1995changes};
	\item the $\RSo$ measure from the $\Reeo$ class, obtained as a generalization of the equivalent martingale measure $\QS$, where the numeraire that defines $\QS$ is any tradable asset $S$ other than the money market account and the $T$-maturity pure discount bond (e.g., see \cite{margrabe1978value}, \cite{geman1995changes}, \cite{rady1997option},
	and the examples in \cite{vecer2011stochastic}, among others).
\end{enumerate}

We demonstrate various applications of the $\Ro$ measure and the $\RTo$ measure in Sections \ref{sec:Rmeasure} and \ref{sec:R1Tmeasure}, respectively. Due to space constraints, we demonstrate the application of the $\RSo$ measure in Internet Appendix Section \ref{iapp:margrabe} using the example of the \cite{margrabe1978value} exchange option pricing model.

\begin{center}
\section{The $\mathbb{R}$ Measure} \label{sec:Rmeasure}\vspace{-1em}
\end{center}

Assuming similar form for the physical and risk-neutral processes, the $\Ro$ measure can be used for the derivation of the analytical solutions of expected future prices of contingent claims under all models that admit an analytical solution to the claim's price using the risk-neutral measure $\q$. The following corollary to Theorem \ref{thm:main} can be used for all applications of the $\Ro$ measure.

\begin{corollary}\label{coro:three}
	When the numeraire asset is the money market account, i.e., $\Y = B = \ue^{\int_0^{\cdot} r_u\D u}$,  then $\Ree$ = $\Ro$, and the expected future price of the contingent claim $\Y$, can be given using equation \eqref{eq:priceQIT} as
	\begin{align}\label{eq:rep1}
		\begin{aligned}
			\E_t \left[ \X_H\right]  &= \E^\Ro_t \left[\ue^{-\int_H^T r_u\D u} \X_T\right].
		\end{aligned}
	\end{align}
\end{corollary}

\begin{proof}
	Equation \eqref{eq:rep1} is a direct application of equation \eqref{eq:priceQIT}.
\end{proof}

The $\Ro$ measure can be used to derive the expected future prices of contingent claims under the following models:\footnote{We apply the $\Ro$ measure to some of the models in the main part of the paper, and provide this measure under a multidimensional  Brownian  motion in Internet Appendix Section \ref{app:DiffEEM}.\ref{app:multiBM}, which can be applied to other models not explicitly considered in the main part of the paper. Some of these applications are given in the Appendix and the Internet Appendix. In our ongoing/future research \cite[see][]{nawalkha2022sharpe}, we also derive the $\Ro$ measure for L\'{e}vy processes.} the equity option pricing models of \cite{black1973pricing}, \cite{cox1976valuation}, \cite{merton1976option}, \cite{hull1987pricing}, and \cite{rubinstein1991pay}; the corporate debt pricing models of \cite{Merton1974ON}, \cite{Black1976VALUING}, and \cite{Leland1996Optimal}; the term structure models of \citet{dai2000specification, Dai2002Expectation}, \citet{Ahn2002Quadratic}, \cite{leippold2003design}, and \cite{collin2008identification} for pricing default-free bonds; the credit default swap pricing model of \cite{longstaff2005corporate}; the VIX futures and the variance swaps models of \cite{dew2017price}, \cite{eraker2017explaining}, \cite{johnson2017risk} and \cite{cheng2019vix}\footnote{\cite{eraker2017explaining} and \cite{cheng2019vix} consider only VIX futures and not variance swaps. \cite{dew2017price} and \cite{johnson2017risk} consider both VIX futures and variance swaps.}; the currency option model of \cite{garman1983foreign}; and the various Fourier transform-based contingent claim models of \cite{heston1993closed}, \cite{Duffie2000Transform}, \cite{carr2002fine}, and others, summarized in Section \ref{sec:Rmeasure}.\ref{sec:Rtransforms}, among others.

We present the $\Ro$ measure using the examples of the \cite{black1973pricing} option pricing model, the affine term structure models of \cite{dai2000specification} and \cite{collin2008identification}, and the Fourier transform-based model of \cite{Duffie2000Transform} in this section.

\subsection{The Black-Scholes Model}\label{sec:BSR}

The geometric Brownian motion for the asset price process under the \cite{black1973pricing} model is given as follows:
\begin{align}\label{eq:GBM}
	\frac{\D S_s}{S_s} &=\mu\D s + \sigma \D \Wp,
\end{align}
where $\mu$ is the drift, $\sigma$ is the volatility, and $\Wpo$ is the Brownian motion under the physical measure $\pp$. The Girsanov theorem allows the transformation of Brownian motions under two equivalent probability measures. Defining $\Wq = \Wp + \int_0^s{\gamma} \D u$, where $\gamma = \left(\mu-r\right)/\sigma$ is the market price of risk (MPR), the asset price process under the risk neutral measure $\q$ is
\begin{align}\label{eq:GBMQ}
	\frac{\D S_s}{S_s} &=r\D s + \sigma \D \Wq.
\end{align}

Using Lemma \ref{lemma:Ree}, the Brownian motion under the $\Ro$ measure is derived in Appendix \ref{prof:BMR} as follows:
\begin{align*}%
	\Wr = \Wp + \int_0^s\I_{\left\{u\geq H\right\}}{\gamma} \D u.
\end{align*}

Substituting $\Wp$ from the above equation into equation \eqref{eq:GBM} gives the asset price process under the $\R$ measure as follows:
\begin{align}\label{eq:GBMR}
	\frac{\D S_s}{S_s} &=\left(r+\sigma\gamma\I_{\{s< H\}}\right)\D s + \sigma \D \Wr\\
	\label{eq:GBMRa}
	&=\left(\mu\I_{\{s< H\}} + r\I_{\{s\geq H\}}\right)\D s + \sigma \D \Wr, \tag{\theequation \textit{a}}
\end{align}
where $\I_{\left\{\cdot\right\}}$ is an indicator function which equals 1 if the condition is satisfied (and 0 otherwise).

A highly intuitive property of the $\R$ measure is that any stochastic process under this measure can be obtained by a simple inspection of that process under the $\pp$ measure and the $\q$ measure as follows: The stochastic process under the $\R$ measure is the physical stochastic process \emph{until before the horizon $H$}, and it becomes the risk-neutral process \emph{on or after the horizon $H$}. This property follows from Propositions \ref{prop:Ree1}\eqref{Ree:num1} and \ref{prop:Ree1}\eqref{Ree:num2} and is also consistent with conditions C1 and C2 given earlier. Thus, the asset price process in equation \eqref{eq:GBMRa} (which equals equation \eqref{eq:GBMR} since $\mu =r+\gamma\sigma$) can be obtained by a simple inspection of the asset price process under the physical measure and the risk-neutral measure, given by equations \eqref{eq:GBM} and \eqref{eq:GBMQ}, respectively. This insight is not limited to only the Black-Scholes model and holds more generally under any contingent claim model that uses the $\q$ measure for valuation. This insight is also consistent with the triangular relationship between the Brownian motions $\Wro$, $\Wqo$, and $\Wpo$, given as follows:
\begin{align*}
	\Wr = \left\{\begin{array}{ll}
		\displaystyle \int_0^H \D \Wpu + \int_H^s \D \Wqu, &\text{if} \ \  H \leq s\leq T, \\
		\displaystyle \int_0^s \D \Wpu, &\text{if} \ \  0\leq s < H. \\
	\end{array}\right.
\end{align*}
The Brownian motion under the $\Ro$ measure is the same as the Brownian motion under the $\pp$ measure from time 0 until time $H$, and then from time $H$ onwards, it evolves as the Brownian motion under the $\q$ measure. %

To demonstrate the usefulness of the $\Ro$ measure, consider a European call option $C$ on the asset, maturing at time $T$ with strike price $K$. Substituting $\X_T = C_T = \max (S_T - K, 0) = (S_T - K)^+$ as the terminal payoff from the European call option, equation \eqref{eq:rep1} in Corollary \ref{coro:three} gives
\begin{align}\label{eq:xtpricer}
	\E_t[C_H] = \E_t^\Ro\left[\ue^{-r(T-H)}(S_T - K)^+\right],
\end{align}
where $C_H$ is the price of the call option at time $H$. Since equation \eqref{eq:GBMRa} implies that $S_T$ is lognormally distributed under the $\Ro$ measure, the expected call price can be solved easily and is given in equation \eqref{eq:BSexp} in Appendix {\ref{app:BSsolution}}. While \cite{rubinstein1984simple} first derived the expected price formula under the Black-Scholes model by using the property of double integrals of normally distributed variables, our derivation is much simpler.\footnote{\cite{rubinstein1984simple} uses the following property of normal distribution:
	$$\int_{-\infty}^{\infty}\mathcal{N}(A+Bz)\frac{1}{\sqrt{2\pi}}\ue^{-z^2/2}\D z = \mathcal{N}\left(\frac{A}{\sqrt{1+B^2}}\right).$$
	
    This double integral property is hard to extend to other distributions.}

The Black-Scholes call price formula has also been used as a part of the solution of the option valuation models with more complex stochastic processes and/or terminal boundary conditions. For example, \cite{merton1976option} solves the call option price under a jump diffusion model using a weighted sum of an infinite number of Black-Scholes call prices, and \cite{rubinstein1991pay} solves the price of a forward-start call option as a constant share of the underlying asset's price, where the constant is determined using the Black-Scholes call price. These frameworks can be extended for computing the analytical solution of the expected future call price under these complex option valuation models by using the Black-Scholes expected future call price given by equations \eqref{eq:xtpricer} and \eqref{eq:BSexp}, \emph{as a part of their solutions}. We demonstrate this insight in Appendix {\ref{app:fsoexpectedBS}} by extending Rubinstein's (\citeyear{rubinstein1991pay}) framework to obtain the expected future price of a forward-start call option using the $\Ro$ measure.

\subsection{The Affine Term Structure Models}\label{sec:atsm}

Surprisingly, the analytical solution of the expected future price of a $T$-maturity pure discount bond at time $H$ (for all $t \leq  H \leq T$), has not been derived under any of the dynamic term structure models, such as the affine models of \cite{dai2000specification} and \cite{collin2008identification}, the quadratic model of \cite{Ahn2002Quadratic}, and the forward-rate models of \cite{Heath1992Bond}. Using the maximal affine term structure models of \cite{dai2000specification} as an expositional example, this section shows how such an expectation can be solved under the dynamic term structure models using the $\Ro$ measure. Substituting $F_T = P(T,T) = 1$ in equation {\eqref{eq:rep1}} of Corollary {\ref{coro:three}}, the time $t$ expectation of the future price of the pure discount bond at time $H$ can be given under the $\Ro$ measure as follows:
\begin{align}\label{eq:atsmprice0}
	\E_t\left[ P(H,T)\right] &=\E_t^{\R}\left[\exp\left(- \int_H^T r_u\D u \right) \right].
\end{align}

The above expectation can be generally solved under all dynamic term structure models that admit an analytical solution to the bond's current price,\footnote{The only condition required is that the functional forms of the market prices of risks should allow maintaining the similar admissible forms under both the $\pp$ and $\q$ measures, so that the admissible affine or quadratic form will also hold under the $\Ro$ measure (which is a hybrid of $\pp$ and $\q$ measures as shown by Proposition \ref{prop:Ree1}).} and is mathematically equivalent to the iterated expectation in equation \eqref{eq:priceQ}, with the money market account as the numeraire, given as follows:
\begin{align}\label{eq:atsmprice01}
	\E_t\left[ P(H,T)\right] &=\E_t^{\p}\left[\E_H^{\q}\left[\exp\left(- \int_H^T r_u\D u \right) \right]\right].
\end{align}

However, instead of using the above equation to obtain expected (simple) return, empirical researchers in finance have generally used expected \emph{log} returns based upon the following equation:\footnote{Expected log return is defined as $\E_t\left[ \ln \left(P(H,T)/P(t,T) \right) \right]$, and expected simple return is defined as $\E_t\left[P(H,T)-P(t,T)\right]/P(t,T)$, where $P(t,T)$ is the current price.}
\begin{align}\label{eq:atsmlnprice}
	\E_t\left[ \ln P(H,T)\right] &=\E_t^{\p} \left[ \ln \left(\E_H^{\q}\left[\exp\left(- \int_H^T r_u\D u \right) \right]\right)\right].
\end{align}

Analytical solutions of the \emph{exponentially} affine and \emph{exponentially} quadratic form exist for the inside expectation term $\E_H^{\q}\left[\exp\left(- \int_H^T r_u\D u \right) \right]$ in the above equation for the affine and the quadratic term structure models, respectively. Taking the \emph{log} of the exponential-form solution of this inside expectation term simplifies equation \eqref{eq:atsmlnprice} considerably since log(exp(x)) = x. %
Moreover, for different econometric reasons, researchers have used expected log returns based upon equation \eqref{eq:atsmlnprice}, instead of using expected simple returns based upon equation \eqref{eq:atsmprice01} for studying default-free bond returns. For example, \cite{Dai2002Expectation}, \cite{bansal2002term}, \cite{cochrane2005bond}, and \cite{eraker2015durable}, all use $\E_t^{\p}\left[\ln\left(P(H,T)/P(t,T) \right) \right]$ for the models studied in their papers.

We would like to make two important observations in the context of the above discussion. First, an analytical solution of $\E_t\left[ P(H,T)\right]$ given by equation \eqref{eq:atsmprice0}, can be generally derived using the $\Ro$ measure for any dynamic term structure model which admits an analytical solution for the current price of the zero-coupon bond using the $\q$ measure. Second, it is generally not possible to know the analytical solution of $\E_t\left[ P(H,T)\right]$ by knowing the analytical solution of $\E_t\left[ \ln \left(P(H,T)\right) \right]$, or vice versa, for most term structure models except over the infinitesimal horizon $H = t + dt$.\footnote{Ito's lemma can be used to obtain $\E_t\left[ P(H,T)\right]$ from $\E_t\left[ \ln \left(P(H,T)\right) \right]$, or vice-versa, when $H$ = $t$ + $dt$, where $dt$ is an infinitesimally small time interval.} Thus, even though the estimates of the expected log returns over a finite horizon have been obtained by many researchers (usually using one month returns), such estimates cannot be used by fixed income investors to compute the \emph{expected returns} of bonds and bond portfolios over a finite horizon, without making additional return distribution assumptions.

In order to derive the expected returns of default-free bonds, the following derives the expected future price of a pure discount bond at any finite horizon $H$ (for all $t \leq H \leq T$) using the affine term structure models (ATSMs) of \cite{dai2000specification} and \cite{collin2008identification}. \cite{collin2008identification} present a new method to obtain maximal ATSMs that uses infinitesimal maturity yields and their quadratic covariations as the globally identifiable state variables. The new method allows model-independent estimates for the state vector that can be estimated directly from yield curve data. When using three or less state variables---which is typical in empirical research on term structure models---the new approach of \cite{collin2008identification} is consistent with the general framework of \cite{dai2000specification} given below, but with some clear advantages in the estimation and interpretation of the economic state variables.\footnote{When using more than three state variables, the maximal affine models of \cite{collin2008identification} can have more identifiable parameters than the ``maximal'' model of \cite{dai2000specification}. For these cases, the framework presented in this paper can be generalized so that the generalized $A_m(N)$ class of maximal affine models can have $m$ square-root processes and $N-m$ Gaussian state variables, but with $M$ number of Brownian motions, such that $N \leq M$. Extending the \cite{dai2000specification} model so that the number of Brownian motions may exceed the number of state variables ensures that the extended model nests the more general maximal affine models of \cite{collin2008identification} when the number of state variables is more than three. For example, \cite{collin2008identification} present a maximal $A_2(4)$ model which requires 5 Brownian motions in order to be consistent with the \cite{dai2000specification} framework. Also, see the discussion in \citeauthor{collin2008identification} (\citeyear{collin2008identification}, pp. 764-765, and footnote 19).} This framework is also consistent with the unspanned stochastic volatility-based affine term structure models of \cite{collin2009can} with appropriate parameter restriction on the stochastic processes.

Assume that the instantaneous short rate $r$ is an affine function of a vector of $N$ state variables $ \mymat{Y} = (Y_{1}, Y_{2}, ... , Y_{N} )$ (which may be observed as in \cite{collin2008identification} or latent as in \cite{dai2000specification}), given as
\begin{align*}
	r_s = \delta_0 + \sum_{i=1}^{N} \delta_i Y_{is} \triangleq \delta_0 + \mymat{\delta_y}' \mymat{Y_s},
\end{align*}
where $\mymat{Y}$ follows an affine diffusion under the physical measure $\p$, as follows:
\begin{align}\label{eq:StatevarprocessP}
	\D \mymat{Y_s} = \mat{K} \left(\mymat{\Theta} - \mymat{Y_s} \right)\D s + \mymat{\Sigma} \sqrt{\mymat{V_s}} \D \mymat{\Wp}.
\end{align}
$\mymat{W^{\p}}$ is a vector of $N$ independent standard Brownian motions under $\p$, $\mymat{\Theta}$ is $N\times 1$ vector, $\mat{K}$ and $\mymat{\Sigma}$ are $N\times N$ matrices, which may be nondiagonal and asymmetric, and $\mymat{V}$ is a diagonal matrix with the $i$th diagonal element given by
\begin{align*}
	\begin{bmatrix} \mymat{V_s} \end{bmatrix}_{ii} = \alpha_i + \mymat{\beta_i}' \mymat{Y_s}.
\end{align*}
Assume that the market prices of risks, $\mymat{\gamma_s}$, are given by
\begin{align*}
	\mymat{\gamma_s} = \sqrt{\mymat{V_s}}\mymat{\gamma},
\end{align*}
where $\mymat{\gamma}$ is an $N \times 1$ vector of constants. Using the change of measure, the risk-neutral process for $\mymat{Y}$ is given as,
\begin{align}\label{eq:StatevarprocessQ}
	\D \mymat{Y_s} &= \left[\mat{K} \left(\mymat{\Theta} - \mymat{Y_s} \right)-\mymat{\Sigma} \mymat{V_s} \mymat{\gamma}\right]\D s + \mymat{\Sigma} \sqrt{\mymat{V_s}} \D \mymat{\Wq} \\ \label{eq:StatevarprocessQ2}
	& = \mat{K^*}\left(\mymat{\Theta^*} - \mymat{Y_s} \right)\D s + \mymat{\Sigma} \sqrt{\mymat{V_s}} \D \mymat{\Wq}, \tag{\theequation \textit{a}}
\end{align}
where $\mat{K^*} = \mat{K} +\mymat{\Sigma}\mymat{\Phi}$, $\mymat{\Theta^*} =\inv{\mat{K^*}}\left( \mat{K}\mymat{\Theta} - \mymat{\Sigma}\mymat{\psi}\right)$, the $i$th row of $\mymat{\Phi}$ is given by $\gamma_i \mymat{\beta_i}'$, and $\mymat{\psi}$ is an $N\times 1$ vector whose $i$th element is given by $\gamma_i\alpha_i$.

Using Lemma \ref{lemma:Ree} and Internet Appendix Section \ref{app:DiffEEM}.\ref{app:multiBM}, the stochastic process for $ \mymat{Y}$ under the $\RR$ measure is given as follows:
\begin{align}\label{eq:Statevarprocess}
	\D \mymat{Y_s} &= \left[\mat{K} \left(\mymat{\Theta} - \mymat{Y_s} \right)-\I_{\left\{s\geq H\right\}} \mymat{\Sigma} \mymat{V_s} \mymat{\gamma}\right]\D s + \mymat{\Sigma} \sqrt{\mymat{V_s}} \D \mymat{\Wr} \\ \label{eq:Statevarprocess2}
	&=\left[ \mat{K} \left(\mymat{\Theta} - \mymat{Y_s} \right) \I_{\{s< H\}} + \mat{K^*} \left(\mymat{\Theta^*} - \mymat{Y_s} \right)\I_{\{s\geq H\}}\right]\D s + \mymat{\Sigma} \sqrt{\mymat{V_s}} \D \mymat{\Wr}. \tag{\theequation \textit{a}}
\end{align}

Recall from the previous section that any stochastic process under the $\R$ measure is the physical stochastic process until before the horizon $H$, and it becomes the risk-neutral process on or after the horizon $H$. Thus, the state variable process under the $\RR$ measure in equation  \eqref{eq:Statevarprocess2} can also be obtained by a simple inspection of the state variable process under the $\pp$ measure given by equation \eqref{eq:StatevarprocessP}, and under the $\q$ measure given by equation \eqref{eq:StatevarprocessQ2}.

As shown by \cite{dai2000specification}, all $N$-factor ATSMs can be uniquely classified into $N+1$ non-nested subfamilies $\A_m(N)$, where $m$ = 0, 1,...,$N$, indexes the degree of dependence of the conditional variances on the number of state variables $\mymat{Y}$. The $\A_m(N)$ model must satisfy several parametric restrictions for admissibility as outlined in \cite{dai2000specification}. Under these restrictions, the expected future price of a $T$-maturity pure discount bond under the $\A_m(N)$ model can be given as
\begin{align}
	\E_t\left[ P(H,T)\right] &=\E_t^{\R}\left[\exp\left(- \int_H^T r_u\D u \right) \right] = \E_t^{\R}\left[\exp\left( - \int_H^T \left(\delta_0 + \mymat{\delta_y}' \mymat{Y_u}\right)\D u \right) \right]  \nonumber \\
	&=\exp\left(- \delta_0 (T-H) \right)\cdot\E_t^{\R}\left[\E_H^{\R}\left[\exp\left( - \int_H^T  \mymat{\delta_y}' \mymat{Y_u}\D u \right) \right]\right]. \label{eq:atsmiterated}
\end{align}

In order to compute the above expectation, we first consider a more general expectation quantity starting at time $H$ defined as
\begin{align*}
 \Psi^*(Y_H; \tau)& \triangleq\E_H^\R\left[\exp\left(-\mymat{b^*}'\mymat{Y_{H+\tau}}-\int_H^{H+\tau}\mymat{c^*}'\mymat{Y_u} \D u\right)\right] \\
& = \E_H^\q\left[\exp\left(-\mymat{b^*}'\mymat{Y_{H+\tau}}-\int_H^{H+\tau}\mymat{c^*}'\mymat{Y_u} \D u\right)\right].
\end{align*}
For any well-behaved $N\times 1$ vectors $\mymat{b^*}$, $\mymat{c^*}$, $\mymat{\Theta^*}$ and $N\times N$ matrix $\mat{K^*}$, according to the Feynman-Kac theorem, the expected value $\Psi^*(Y_H; \tau)$ should fulfill the following PDE:
\begin{align*}
	-\frac{\partial \Psi^*}{\partial \tau}   +  \frac{ \partial \Psi^*}{\partial Y} \mat{K^*} \left(\mymat{\Theta^*} - \mymat{Y} \right) + \frac{1}{2}\text{tr}\left[\frac{\partial^2\Psi^*}{\partial Y^2}\mymat{\Sigma}\mymat{V}\mymat{\Sigma}' \right] = \mymat{c^*}'\mymat{Y}\Psi^*,
\end{align*}
with a boundary condition $\Psi^*(Y_H; 0)=\exp\left(-\mymat{b^*}'\mymat{Y_H}\right)$. Then, we can easily show that the solution to this PDE is
\begin{align}\label{eq:atsmABQ}
	\Psi^*(Y_H; \tau)= \exp\left(- A^{(\mymat{b^*}, \mymat{c^*})}_{\mat{K^*},\mymat{\Theta^*}}(\tau)- \mymat{B^{(\mymat{b^*}, \mymat{c^*})}_{\mat{K^*},\mymat{\Theta^*}}(\tau)}'\mymat{Y_H}\right),
\end{align}
where $ A^{(\mymat{b^*}, \mymat{c^*})}_{\mat{K^*},\mymat{\Theta^*}}(\tau)$ and $ \mymat{B^{(\mymat{b^*}, \mymat{c^*})}_{\mat{K^*},\mymat{\Theta^*}}(\tau)}$ are obtained by solving the following ordinary differential equations (ODEs):
\begin{align}\label{eq:AmnABQ}
	\begin{aligned}
		\frac{\D A^{(\mymat{b^*}, \mymat{c^*})}_{\mat{K^*},\mymat{\Theta^*}}(\tau)      }{\D \tau} &= \mymat{\Theta^*}'\mat{K^*}'\mymat{B^{(\mymat{b^*}, \mymat{c^*})}_{\mat{K^*},\mymat{\Theta^*}}(\tau)} - \frac{1}{2}\sum_{i=1}^N \begin{bmatrix} \mymat{\Sigma}'\mymat{B^{(\mymat{b^*}, \mymat{c^*})}_{\mat{K^*},\mymat{\Theta^*}}(\tau)} \end{bmatrix}_{i}^2\alpha_i,  \\
		\frac{\D \mymat{B^{(\mymat{b^*}, \mymat{c^*})}_{\mat{K^*},\mymat{\Theta^*}}(\tau)}  }{\D \tau} &= - \mat{K^*}'\mymat{B^{(\mymat{b^*}, \mymat{c^*})}_{\mat{K^*},\mymat{\Theta^*}}(\tau)}- \frac{1}{2}\sum_{i=1}^N \begin{bmatrix} \mymat{\Sigma}'\mymat{B^{(\mymat{b^*}, \mymat{c^*})}_{\mat{K^*},\mymat{\Theta^*}}(\tau)} \end{bmatrix}_{i}^2\mymat{\beta_i} + \mymat{c},
	\end{aligned}
\end{align}
with the terminal conditions $A^{(\mymat{b^*}, \mymat{c^*})}_{\mat{K^*},\mymat{\Theta^*}}(0) = 0$, $\mymat{B^{(\mymat{b^*}, \mymat{c^*})}_{\mat{K^*},\mymat{\Theta^*}}(0)} =\mymat{b^*}$. $ A^{(\mymat{b^*}, \mymat{c^*})}_{\mat{K^*},\mymat{\Theta^*}}(\tau)$ and $ \mymat{B^{(\mymat{b^*}, \mymat{c^*})}_{\mat{K^*},\mymat{\Theta^*}}(\tau)}$ can be solved through numerical procedures such as Runge-Kutta. Furthermore, another similar expectation quantity $\Psi(Y_t; \tau)$ starting at time $t$ with $\tau<H-t$ satisfies
\begin{align}
	\Psi(Y_t; \tau)&\triangleq\E_t^\R\left[\exp\left(-\mymat{b}'\mymat{Y_{t+\tau}}-\int_t^{t+\tau}\mymat{c}'\mymat{Y_u} \D u\right)\right] = \E_t^\p\left[\exp\left(-\mymat{b}'\mymat{Y_{t+\tau}}-\int_t^{t+\tau}\mymat{c}'\mymat{Y_u} \D u\right)\right] \nonumber \\
	&= \exp\left(-A^{(\mymat{b}, \mymat{c})}_{\mat{K},\mymat{\Theta}}(\tau)- \mymat{B^{(\mymat{b}, \mymat{c})}_{\mat{K},\mymat{\Theta}}(\tau)}'\mymat{Y_t}\right), \label{eq:atsmABP}
\end{align}
where $A^{(\mymat{b}, \mymat{c})}_{\mat{K},\mymat{\Theta}}(\tau)$ and $\mymat{B^{(\mymat{b}, \mymat{c})}_{\mat{K},\mymat{\Theta}}(\tau)}$ are obtained similarly as in equation \eqref{eq:AmnABQ} with the parameters $(b^*, c^*, \mat{K^*},\mymat{\Theta^*})$ replaced with parameters $(b, c, \mat{K},\mymat{\Theta})$.

Now we are ready to express the solution to the expected future price of a $T$-maturity pure discount bond at time $H$ under the $\A_m(N)$ model as follows: From equation \eqref{eq:atsmiterated}, for all  $t\leq H\leq T$,
\begin{align}
	\E_t\left[ P(H,T)\right]	&=\exp\left(- \delta_0 (T-H) \right)\cdot\E_t^{\R}\left[\E_H^{\R}\left[\exp\left( - \int_H^T  \mymat{\delta_y}' \mymat{Y_u}\D u \right) \right]\right] \nonumber \\
	&=\exp\left(- \delta_0 (T-H) \right)\cdot\E_t^{\R}\left[\exp\left( - A^{(\mymat{b^*}, \mymat{c^*})}_{\mat{K^*},\mymat{\Theta^*}}(T-H) -  B^{(\mymat{b^*}, \mymat{c^*})}_{\mat{K^*},\mymat{\Theta^*}}(T-H)'\mymat{Y_H} \right) \right] \nonumber \\
	&=\exp\left(- \delta_0 (T-H) - A^{(\mymat{b^*}, \mymat{c^*})}_{\mat{K^*},\mymat{\Theta^*}}(T-H) - A^{(\mymat{{b}}, \mymat{c})}_{\mat{K},\mymat{\Theta}}(H-t)-\mymat{B^{(\mymat{{b}}, \mymat{c})}_{\mat{K},\mymat{\Theta}}(H-t)'\mymat{Y_t}} \right), \label{eq:atsmsolu}
\end{align}
with
$$
b^* = 0, \ \ c^* = \mymat{\delta_y}, \ \ b = B^{(\mymat{b^*}, \mymat{c^*})}_{\mat{K^*},\mymat{\Theta^*}}(T-H), \ \ c = 0.
$$
The second equality applies the result in equation \eqref{eq:atsmABQ}, while the last equality applies the result in equation \eqref{eq:atsmABP}.

Equation \eqref{eq:atsmsolu} can be used to obtain the expected future price of a pure discount bond under any specific ATSMs of \cite{dai2000specification} and \cite{collin2008identification}. {We demonstrate three examples in Internet Appendix Section \ref{app:someATSM}. Specifically,} Internet Appendix Sections \ref{app:someATSM}.\ref{app:Vas} and \ref{app:someATSM}.\ref{app:CIR} give the analytical solutions of the expected future price of a pure discount bond, under the Vasicek model and the CIR model, respectively. To our knowledge, these equations provide the first analytical solutions of the expected future price of a pure discount bond under the Vasicek and CIR models in the finance literature. Using a tedious derivation, Internet Appendix Section \ref{app:someATSM}.\ref{sec:A1rmodel} gives the expected future price of the pure discount bond under the $A_{1r}(3)$ model of \cite{dai2000specification}. Furthermore, Internet Appendix Section \ref{sec:qtsm} derives a general solution to the expected future price of a pure discount bond under the maximal quadratic term structure model (QTSM), as well as a specific solution under the QTSM3 model \cite[see][]{Ahn2002Quadratic}.

The $\RR$ measure can also be used to compute the expected yield \cite[e.g., see ][]{duffee2002term} at a future horizon $H$, under the $N$-factor ATSMs, as follows:
\begin{align*}
	\text{Expected Yield} =\frac{-\E_t\left[\ln P(H,T)\right]}{T-H}.
\end{align*}

Hence, to derive the expected yield, it suffices for us to obtain the key term $\E_t\left[\ln P(H,T)\right]$. By simplifying equation \eqref{eq:atsmlnprice} using the $\RR$ measure, we obtain $\E_t\left[\ln P(H,T)\right]$ under the $\A_m(N)$ model as follows:\footnote{By applying Proposition \ref{prop:Ree1}, $\E_t^\R\left[f\left(\E_H^\R\left[\Y_H \frac{\X_T}{\Y_T}\right]\right)\right]  = \E_t^\p \left[f\left(\E_H^\Qem\left[\Y_H \frac{\X_T}{\Y_T}\right]\right)\right]$ also holds for a function $f(x)$ of $x$. Here, we employ the log function $f(x) = \ln(x)$, use $\Qem = \q$, and assume the numeraire $G$ is the money market account.}
\begin{align*}
	\begin{aligned}
		\E_t\left[\ln P(H,T)\right] & = \E_t^\R\left[ \ln \E_H^{\R}\left[\exp\left(- \int_H^T r_u\D u \right) \right]\right] \\
		& =\E_t^\R\left[\ln \E_H^{\R}\left[\exp\left( - \int_H^T \left(\delta_0 + \mymat{\delta_y}' \mymat{Y_u}\right)\D u \right) \right]\right] \\
		&=  - \delta_0 (T-H) - A^{(\mymat{b^*}, \mymat{c^*})}_{\mat{K^*},\mymat{\Theta^*}}(T-H) -\mymat{B^{(\mymat{b^*}, \mymat{c^*})}_{\mat{K^*},\mymat{\Theta^*}}(T-H)'\E_t^\R\left[\mymat{Y_H} \right]},
	\end{aligned}
\end{align*}
with $b^* = \mymat{0}$, $c^* =\mymat{\delta_y} $, $\E_t^\R\left[\mymat{Y_H} \right] = \mymat{\Theta} + \ue^{-\mat{K}\left(H-t\right)
}\left(\mymat{Y_t} -\mymat{\Theta} \right)$, and $A^{(\mymat{b^*}, \mymat{c^*})}_{\mat{K^*},\mymat{\Theta^*}}(\tau)$ and $ \mymat{B^{(\mymat{b^*}, \mymat{c^*})}_{\mat{K^*},\mymat{\Theta^*}}(\tau)}$ are defined as in equation {\eqref{eq:AmnABQ}}. %

\subsection{The $R$-Transforms}\label{sec:Rtransforms}
In two path-breaking papers, \cite{bakshi2000spanning} and \cite{Duffie2000Transform} derive new transforms that extend the Fourier transform-based method of \cite{heston1993closed} for the valuation of a wide variety of contingent claims. Though \cite{Chacko2002Pricing} also independently develop a generalized Fourier transform-based method for the valuation of a variety of interest rate derivatives, their approach can be obtained as a special case of the more general transform-based method of \cite{Duffie2000Transform}. %
Since the specific transforms in these papers use the risk-neutral measure $\q$ to do valuation, %
we refer to these transforms as $Q$-transforms.

This section derives new $R$-transforms using the $\Ro$ measure to obtain the expected future prices of a wide variety of contingent claims that are priced using the $Q$-transforms
of \cite{bakshi2000spanning},  \cite{Duffie2000Transform}, and \cite{Chacko2002Pricing}.
The $R$-transforms nest the corresponding $Q$-transforms for the special case when $H=t$. The $R$-transforms can be used to derive analytical solutions of the expected future prices of contingent claims that are priced under the following models: the affine option pricing models of \cite{heston1993closed},  \cite{bates1996jumps, bates2000post}, \cite{bakshi1997empirical}, \cite{pan2002jump}, \cite{bakshi2000spanning}, \cite{Duffie2000Transform},  and \cite{Chacko2002Pricing}; the L\'{e}vy option pricing models of \citet{carr2002fine} and \cite{Carr2003The}; the forward-start option model of \cite{kruse2005pricing};
the bond option pricing models in the maximal affine classes by \cite{dai2000specification, Dai2002Expectation},  \cite{collin2008identification}, \cite{collin2009can}, and in the maximal quadratic class by \cite{Ahn2002Quadratic}, among others. \cite{nawalkha2022sharpe} obtain the analytical solutions of the expected future prices of equity options and a variety of interest rate derivatives using some of the above models.

As in \cite{Duffie2000Transform}, we define an affine jump diffusion state process on a probability space $\left(\Omega, \mathcal{F}, \mathbb{P}\right)$, and assume that $\mymat{Y}$ is a Markov process in some state space $D \in \mathcal{R}^N$, solving the stochastic differential equation
\begin{align*}%
	\D \mymat{Y_s} = \mymat{\mu(Y_s)}\D s + \mymat{\sigma(Y_s)} \D \mymat{\Wp} + \D \mymat{J_s},
\end{align*}
where $\mymat{\Wpo}$ is a standard Brownian motion in $\mathcal{R}^N$; $\mymat{\mu}:D\rightarrow\mathcal{R}^N$, $\mymat{\sigma}:D\rightarrow \mathcal{R}^{N\times N}$, and $\mymat{J}$ is a pure jump process whose jumps have a fixed probability distribution $\mymat{\nu}$ on $\mathcal{R}^N$ and arrive with intensity $\left\{\lambda(\mymat{Y_s}):s\geq 0\right\}$, for some $\lambda: D\rightarrow [0,\infty)$. The transition semi-group of the process $\mymat{Y}$ has an infinitesimal generator $\mathscr{D}$ of the L\'{e}vy type, defined at a bounded $\mathcal{C}^2$ function $f:D\rightarrow\mathcal{R}$, with bounded first and second derivatives, by
\begin{align*}
	\mathscr{D}f(\mymat{y})=\mymat{f_y(y)\mu(y)} + \frac{1}{2}\text{tr}\left[\mymat{f_{yy}(y)\sigma(y)\sigma(y)'}\right] +{\lambda(\mymat{y})}\int_{\mathcal{R}^N}\left[f(\mymat{y+j})-f(\mymat{y})\right]\D \mymat{\nu(j)}.
\end{align*}

As in \cite{Duffie2000Transform}, \cite{Duffie1996A}, and \cite{dai2000specification}, we assume that $\left(D, \mymat{\mu}, \mymat{\sigma}, \mymat{\lambda}, \mymat{\nu}\right)$ satisfies the joint restrictions, such that $\mymat{Y}$ is well defined, and the affine dependence of $\mymat{\mu}$, $\mymat{\sigma\sigma'}$, ${\lambda}$, and $r$ are determined by coefficients $(\mymat{k},\mymat{h},\mymat{l}, \mymat{\rho})$ defined by
\begin{itemize}
	\item $\mymat{\mu(y)} = \mymat{k_0} + \mymat{k_1y}$, for $\mymat{k}=(\mymat{k_0}, \mymat{k_1})\in \mathcal{R}^N \times \mathcal{R}^{N\times N}$.
	\item $\left(\mymat{\sigma(y)\sigma(y)'}\right)_{ij} =(\mymat{h_0})_{ij} + (\mymat{h_1})_{ij}\cdot \mymat{y}$, for $\mymat{h}=(\mymat{h_0}, \mymat{h_1})\in \mathcal{R}^{N\times N} \times \mathcal{R}^{N\times N\times N}$.
	\item  $\lambda(\mymat{y}) = l_0 + \mymat{l_1'y}$, for $\mymat{l}=(l_0, \mymat{l_1})\in \mathcal{R} \times \mathcal{R}^{N}$.
	\item  $r(\mymat{y}) = \rho_0 + \mymat{\rho_1'y}$, for $\mymat{\rho}=(\rho_0, \mymat{\rho_1})\in \mathcal{R} \times \mathcal{R}^{N}$.
\end{itemize}

For $\mymat{x}\in \mathcal{C}^N$, the set of $N$-tuples of complex numbers, we define $\theta(\mymat{x})=\int_{\mathcal{R}^N}\exp\left(\mymat{x'j}\right)\D \mymat{\nu(j)}$ whenever the integral is well defined. Hence, the ``coefficients'' $(\mymat{k}, \mymat{h}, \mymat{l}, \theta)$ capture the distribution of $\mymat{Y}$, and a ``characteristic'' $\mymat{\chi} = (\mymat{k}, \mymat{h}, \mymat{l}, \theta, \mymat{\rho})$ captures both the distribution of $\mymat{Y}$ as well as the effects of any discounting. Furthermore, we assume the state vector $\mymat{Y}$ to be an affine jump diffusion with coefficients $(\mymat{k^*}, \mymat{h^*}, \mymat{l^*}, \theta^*)$ under the risk-neutral measure $\q$, and the relevant characteristic for risk-neutral pricing is then  $\mymat{\chi^*} = (\mymat{k^*}, \mymat{h^*}, \mymat{l^*}, \theta^*, \mymat{\rho^*})$.

\cite{Duffie2000Transform} define a transform $\phi: \mathcal{C}^N\times \mathcal{R}_+\times \mathcal{R}_+  \rightarrow \mathcal{C}$ of $\mymat{Y_T}$ conditional on $\mathcal{F}_t$, when well defined for all $t\leq T$, as\footnote{Notably, when $z= i u$ (where $i = \sqrt{-1}$), this $\phi(z;t,T)$ becomes the $Q$-transform of \cite{bakshi2000spanning} (see equation (6)), defined as the characteristic function of the state-price density.}
\begin{align}\label{eq:Qtransform}
	\begin{aligned}
		\phi(\mymat{z}; t, T) &\triangleq \E_t^\q\left[\exp{\left(-\int_t^T r(\mymat{Y_u})\D u\right)}\exp\left( \mymat{z'Y_T}\right)\right].
	\end{aligned}
\end{align}
The current prices of various types of contingent claims can be computed conveniently using the above transform. Since this transform is derived to value the current price of contingent claims, we refer to it as the $Q$-transform. To extend the $Q$-transform and compute the expected future prices of contingent claims, we define an $R$-transform $\phi^R: \mathcal{C}^N \times \mathcal{R}_+\times \mathcal{R}_+ \times \mathcal{R}_+ \rightarrow \mathcal{C}$ of $\mymat{Y_T}$ conditional on $\mathcal{F}_t$, when well defined for all $t\leq H\leq T$, as%
\begin{align}\label{eq:Rtransform}
	\phi^R(\mymat{z}; t, T, H) &\triangleq \E_t^\RR\left[\exp{\left(-\int_H^T r(\mymat{Y_u})\D u\right)}\exp\left( \mymat{z'Y_T}\right)\right]. %
\end{align}

Proposition \ref{prop:Rtrans} in Appendix \ref{app:soluRandExtT} obtains the solution of the $R$-transform $\phi^R$.  For the special case of $H=t$, the $R$-transform and its solution given by equations \eqref{eq:Rtransform} and \eqref{eq:Rtransformsol}, reduce to the $Q$-transform and its solution, respectively in  \citeauthor{Duffie2000Transform} (\citeyear{Duffie2000Transform}, equations 2.3 and 2.4). \cite{nawalkha2022sharpe} use the $R$-transform to derive the analytical solutions of the expected future prices of equity options under the stochastic-volatility-double-jump (SVJJ) model of \cite{Duffie2000Transform} and the CGMY model of \cite{carr2002fine}, and interest rate derivatives under the multifactor affine jump diffusion model of \cite{Chacko2002Pricing}. As an extension of the $R$-transform given in equation \eqref{eq:Rtransform}, Internet Appendix Section \ref{app:fsoexpected} proposes a new \emph{forward-start $R$-transform}, which can be used to obtain the expected future price of a forward-start option under affine-jump diffusion processes.

While the $Q$-transform given in equation \eqref{eq:Qtransform} is sufficient for the valuation of many contingent claims---for example, standard equity options, discount bond options, caps and floors, exchange rate options, chooser options, digital options, etc.---certain pricing problems like Asian option valuation or default-time distributions are easier to solve using the extended transform given by equation (2.13) in \cite{Duffie2000Transform}. The following proposes an \emph{extended} $R$-transform corresponding to their extended transform, which we refer to as the \emph{extended} $Q$-transform. We define the extended $R$-transform $\varphi^R:  \mathcal{R}^N \times \mathcal{C}^N  \times \mathcal{R}_+\times \mathcal{R}_+ \times \mathcal{R}_+ \rightarrow \mathcal{C}$ of $\mymat{Y_T}$ conditional on $\mathcal{F}_t$, when well defined for $t\leq H\leq T$, as
\begin{align}\label{eq:Rtransform2}
	\varphi^R(\mymat{v}, \mymat{z}; t, T, H) &\triangleq \E_t^\RR\left[\exp{\left(-\int_H^T r(\mymat{Y_u})\D u\right)}\left(v'Y_T\right)\exp\left( \mymat{z'Y_T}\right)\right]. %
\end{align}

Proposition \ref{prop:Rtransext} in Appendix \ref{app:soluRandExtT} gives the solution to the extended $R$-transform $\varphi^R$. For the special case of $H=t$, the extended $R$-transform and its solution given by equations \eqref{eq:Rtransform2} and \eqref{eq:Rtransformsol2} reduce to the extended $Q$-transform and its solution, respectively in \citeauthor{Duffie2000Transform} (\citeyear{Duffie2000Transform}, equations 2.13 and 2.14).\footnote{An application of the extended $R$-transform for the derivation of the analytical solution of the expected future price of an Asian option on interest rates is given by \cite{nawalkha2022sharpe}.}

\begin{center}
\section{The ${\mathbb{R}}^{T}_1$ Measure} \label{sec:R1Tmeasure} \vspace{-1em}
\end{center}

This section considers different applications of the $\RTo$ measure, which is perhaps the next most important EEM after the $\Ro$ measure, and is derived as a generalization of the forward measure $\QT$. Assuming similar forms for the stochastic processes under the physical measure $\pp$ and the forward measure $\QT$, the $\RTo$ measure leads to analytical solutions of expected future prices of contingent claims under all models that admit an analytical solution to the claim's price using the forward measure. The following corollary to Theorem \ref{thm:main2} can be used for all applications of the $\RTo$ measure.

\begin{corollary}\label{coro:three2}
	When the numeraire asset is a $T$-maturity pure discount bond, i.e., $\Y = P(\cdot,T)$, then $\Reeo$ = $\RTo$, and the expected future price of the contingent claim $\Y$, can be given using equation \eqref{eq:priceQIT1} as
	\begin{align}\label{eq:rep2}
		\begin{aligned}
			\E_t \left[ \X_H\right]   & =   \E_t^\p\left[P(H,T) \right]\E_t^{\RTo}\left[\X_T \right].
		\end{aligned}
	\end{align}
	Let $S_T$ be the terminal value of the asset underlying the claim, and $V$ = $S/P(\cdot,T)$ define the normalized asset price. Then the expectation of $V_T$ under the $\RTo$ measure can be computed as follows:
	\begin{align}\label{eq:priceQIT1RM}
		\E_t^\RTo\left[V_T\right] = \frac{\E^\pp_t\left[S_H\right]}{\E_t^\p\left[P(H,T) \right]}.
	\end{align}
\end{corollary}
\begin{proof}
	Equations \eqref{eq:rep2} and \eqref{eq:priceQIT1RM} are direct applications of equations \eqref{eq:priceQIT1} and \eqref{eq:priceQIT1R}, respectively, with $\Y = P(\cdot,T)$ and $P(T,T) = 1$.
\end{proof}

The $\RTo$ measure can be used to derive the expected future prices of contingent claims under the following models:\footnote{We apply the $\RTo$ measures to some of the models in the main part of the paper, and provide this measure under a multidimensional Brownian  motion in Internet Appendix Section \ref{app:DiffEEM}.\ref{app:Rforward}, which can also be applied to other models not explicitly considered in the main part of the paper.} the stochastic interest rate-based equity  option pricing model of \cite{Merton1973Theory} and its extensions; the stochastic interest rate-based corporate debt pricing models of \cite{longstaff1995simple}, \cite{Jarrow1997A} and \cite{collin2001credit}; various term structure models for pricing bond options and caps, such as \citet{dai2000specification, Dai2002Expectation}, \cite{collin2008identification}; \citet{Ahn2002Quadratic},  \cite{leippold2003design}; \cite{Heath1992Bond}, \cite{miltersen1997closed}, \cite{brace1997market}, and \cite{jamshidian1997libor}; and the currency option pricing models of \cite{Grabbe1983The}, \cite{AMIN1991Pricing}, and \cite{hilliard1991currency}, among others.

We present the $\RTo$ measure using the examples of the  \cite{Merton1973Theory} option pricing model and the \cite{collin2001credit} corporate debt pricing model in Sections \ref{sec:R1Tmeasure}.\ref{sec:R1TmeasureMerton} and \ref{sec:R1Tmeasure}.\ref{sec:CDG1}, respectively.  Section \ref{sec:R1Tmeasure}.\ref{sec:recov} extends the theoretical framework of \cite{breeden1978prices} to obtain the $\RTo$ density and the physical density of the underlying asset using non-parametric and semi-parametric option return models.

\subsection{The Merton Model} \label{sec:R1TmeasureMerton}
As in \cite{Merton1973Theory}, we assume the asset price $S$ is described by
\begin{align}\label{eq:MertonS}
	\frac{\D S_s}{S_s} &= \mu(s)\D s + \sigma(s) \D W_{1s}^{\p},
\end{align}
where $W_{1s}^{\p}$ is the Brownian motion under $\p$, $\mu(s)$ is the instantaneous expected return which may be stochastic, and $\sigma(s)$ is the instantaneous volatility which is assumed to be deterministic.

For a $T$-maturity pure discount bond $P(\cdot, T)$, we assume its price follows
\begin{align}\label{eq:MertonP}
	\frac{\D P(s,T)}{P(s,T)} &= \mu_P(s,T)\D s + \sigma_P(s,T) \D W_{2s}^{\p},
\end{align}
where $W_{2s}^{\p}$ is the Brownian motion associated with the $T$-maturity pure disocunt bond, which has a deterministic correlation with $W_{1s}^{\p}$, given by the coefficient $\rho(s)$;\footnote{While \citeauthor{Merton1973Theory} (\citeyear{Merton1973Theory}, equation (25)) assumes a constant correlation, we allow the asset return's correlation with the bond return to change \emph{deterministically} as the bond maturity declines over time. Doing this also allows \emph{bond option models} based on the Gaussian multifactor term structure models to be nested in the framework presented in this section by assuming that the asset underlying the option contract is a pure discount bond maturing at any time after $T$.\label{footnote:Footref}} $\mu_P(s,T)$ is the instantaneous expected return, which may be stochastic; and $\sigma_P(s,T)$ is the instantaneous volatility which is assumed to be deterministic.

Consider a European call option $C$ written on the asset $S$ with a strike price of $K$, and an option expiration date equal to $T$. Using the $T$-maturity pure discount bond as the numeraire, the expected future price of this option is given by equation \eqref{eq:rep2} in Corollary \ref{coro:three2}, with $\X_T = C_T = \left(S_T - K\right)^+$ as the terminal payoff given as follows:
\begin{align}\label{eq:priceQIT1Merton}
	\begin{aligned}
		\E_t \left[ C_H\right] &= \E^\p_t \left[P(H,T)\right] \E_t^\RTo\left[\left(S_T - K\right)^+\right] \\
		&= \E^\p_t \left[P(H,T)\right] \E_t^\RTo\left[\frac{\left(S_T - K\right)^+}{P(T,T)}\right] \\
		&= \E^\p_t \left[P(H,T)\right] \E_t^\RTo\left[\left(V_T - K\right)^+\right],
	\end{aligned}
\end{align}
where $V = S/P(\cdot, T)$ is the asset price normalized by the numeraire.

The derivation of an analytical solution of the expected future call price \citep[of a similar form as in][]{Merton1973Theory} using equation \eqref{eq:priceQIT1Merton} requires that the asset price normalized by the bond price be distributed lognormally under the $\RTo$ measure. Based upon Theorem \ref{thm:main2}, and equations \eqref{eq:MertonS} and \eqref{eq:MertonP}, \emph{sufficient} assumptions for this to occur are:
\begin{enumerate}[i)]
	\item the physical drift of the asset price process is of the form $\mu$(s) = $r_s$ + $\gamma$(s), where the risk premium $\gamma$(s) is deterministic, and
	\item the short rate process $r_s$, and the bond price process in equation \eqref{eq:MertonP} are consistent with the various multifactor Gaussian term structure models \cite[see][] {dai2000specification,Heath1992Bond}, such that the physical drift of the bond price process is of the form $\mu_P(s,T)$ = $r_s$ + $\gamma_P(s,T)$, where the risk premium $\gamma_P(s,T)$ is deterministic.
\end{enumerate}

Without specifying a specific Gaussian term structure model for the short rate process it is not possible to express the asset price process and the bond price process under the $\RTo$ measure. However, even without knowing the term structure model, the two assumptions given above allow the derivation of the analytical solution of the expected future call price, which is very similar in its form to the \citeauthor{Merton1973Theory}'s (\citeyear{Merton1973Theory}) original formula. Our solution also generalizes Merton's formula to hold under multifactor Gaussian term structure models for the special case of $H = t$.\footnote{\citeauthor{Merton1973Theory}'s (\citeyear{Merton1973Theory}, footnote 43) original model assumes a constant (and not deterministic) correlation between the asset price process and the bond price process, which is consistent with only single-factor Gaussian term structure models.}

Assume that based upon equations \eqref{eq:MertonS} and \eqref{eq:MertonP}, $\ln V_T$ is normally distributed under both the $\pp$ measure and the $\QT$ measure with the same variance equal to $v_p^2$.\footnote{It is well-known that $V=S/P(\cdot, T)$ is a martingale under the $\QT$ measure. This implicitly follows from \citeauthor{Merton1973Theory}'s (\citeyear{Merton1973Theory}) PDE using Feynman Kac theorem, and explicitly by applying the $\QT$ measure as in \cite{geman1995changes}. Since the change of measure from $\pp$ to $\QT$ changes only the deterministic drift, the variance of $\ln V_T$ remains the same under both measures.} Using the Radon-Nikod\'{y}m derivative process of $\RTo$ with respect to $\pp$ given in Lemma \ref{lemma:reeo}, and the two assumptions given above, $\ln V_T$ is also distributed normally under the $\RTo$ measure with the same variance
$v_p^2$. The change of measure from $\pp$ to either $\QT$ or to $\RTo$, changes the drift of the $\ln V$ process from a given deterministic drift to another deterministic drift, which does not change the variance of $\ln V_T$.\footnote{As an example, see the case of the Merton model with the single-factor Vasicek model in Internet Appendix Section \ref{app:MertonVasicek}.} Thus, we can express the expectation of $\ln V_T$ under the $\RTo$ measure
as follows:
\begin{align}\label{eq:ETRolnT}
	\E_t^\RTo\left[\ln V_T\right] =\ln \E_t^\RTo\left[V_T\right] - \frac{1}{2}\mathrm{Var}_t^\RTo\left[\ln V_T\right] =\ln \frac{\E_t^\p\left[S_H\right]}{\E_t^\p\left[P(H,T)\right]} - \frac{1}{2}v_p^2,
\end{align}
where the second equality follows from equation \eqref{eq:priceQIT1RM} of Corollary \ref{coro:three2}.

Simplifying equation \eqref{eq:priceQIT1Merton}, gives
\begin{align}%
	\begin{aligned}
		\E_t \left[ C_H\right] & = \E^\p_t\left[P(H,T)\right] \E_t^\RTo\left[V_T \I_{\left\{\ln V_T>\ln K\right\}} \right] - K \E^\p_t\left[P(H,T)\right] \E_t^\RTo\left[\I_{\left\{\ln V_T>\ln K\right\}} \right] \\ \nonumber
		& = \E^\p_t\left[P(H,T)\right]\E_t^\RTo\left[V_T\right] \E_t^\RTohat\left[\I_{\left\{\ln V_T>\ln K\right\}} \right] - K \E^\p_t\left[P(H,T)\right] \E_t^\RTo\left[\I_{\left\{\ln V_T>\ln K\right\}} \right] \\\nonumber
		& = \E_t^\p\left[S_H\right] \E_t^\RTohat\left[\I_{\left\{\ln V_T>\ln K\right\}} \right] - K \E^\p_t\left[P(H,T)\right] \E_t^\RTo\left[\I_{\left\{\ln V_T>\ln K\right\}} \right],
	\end{aligned}
\end{align}
where $\RTohat$ is an adjusted probability measure defined by $  \E_t^{\RTohat}[Z_T] = \E_t^\RTo\left[\left({V_T}/\E_t^\RTo\left[V_T\right]\right) \cdot Z_T\right]$ for a given $\mathcal{F}_T$-measurable variable $Z_T$, and the last equality follows since $ \E^\p_t\left[P(H,T)\right]\E_t^\RTo\left[V_T\right]=\E_t^\p\left[S_H\right]$ using equation \eqref{eq:priceQIT1RM}. The solution to the expected future call price follows
from the above equation, and is given as  %
\begin{align}\label{eq:Mertonoptexp}
	\E_t[C_H] =\E_t\left[S_H\right]  \mathcal{N}\left(\hat{d}_1\right) - K\E_t\left[P(H,T)\right] \mathcal{N}\left( \hat{d}_2\right),
\end{align}
where
\begin{align*}
	\hat{d}_1 &= \frac{1}{v_p}\ln\frac{\E_t\left[S_H\right]}{\E_t\left[P(H,T)\right]K } + \frac{v_p}{2}, \ \ \quad
	\hat{d}_2 = \frac{1}{v_p}\ln\frac{\E_t\left[S_H\right]}{\E_t\left[P(H,T)\right]K } - \frac{v_p}{2}, \\
	v_p &= \sqrt{\int_t^T \left[\sigma(u)^2-2\rho(u)\sigma(u)\sigma_P(u,T) +\sigma_P(u,T)^2 \right]\D u}.
\end{align*}
The above analytical solution depends on four variables: $\E_t\left[S_H\right]$, $\E_t\left[P(H,T)\right]$, $v_p$, and $K$ (the option expiration, $T - t$, is implicitly included in $\E_t\left[P(H,T)\right]$ and $v_p$), and reduces to the \citeauthor{rubinstein1984simple}'s (\citeyear{rubinstein1984simple}) formula given in Appendix \ref{app:BSsolution} (see equation \eqref{eq:BSexp}), assuming a constant asset return volatility, a constant asset drift, and a constant short rate. The analytical solution also nests the \citeauthor{Merton1973Theory}'s (\citeyear{Merton1973Theory}) solution for the current call price for the special case of $H=t$, under which the $\RTo$ measure becomes the $\QT$ forward measure  \cite[see][]{geman1989importance}, as well as the solution of the expected future price of a call option written on a pure discount bond under all multifactor Gaussian term structure models.%
\footnote{See footnote \ref{footnote:Footref}.}

Under the two assumptions given above, the definition of $v_p$ remains the same as under the \citeauthor{Merton1973Theory}'s (\citeyear{Merton1973Theory}) model, and the specific solutions of $\E_t\left[S_H\right]$ and $\E_t\left[P(H,T)\right]$ can be derived with the appropriate specifications of the asset price process and the bond price process in equations \eqref{eq:MertonS} and \eqref{eq:MertonP}, respectively. These solutions can be shown to be consistent with various multifactor asset pricing models, such as the \citeauthor{merton1973intertemporal}'s (\citeyear{merton1973intertemporal}) ICAPM or
the APT  \cite[see][]{ross1976arbitrage, connor1989intertemporal}, which allow stochastic interest rates but with the additional restrictions of non-stochastic volatilities and non-stochastic risk premiums.\footnote{While the theoretical frameworks of \citeauthor{merton1973intertemporal} (\citeyear{merton1973intertemporal}, equation 7) and \citeauthor{connor1989intertemporal} (\citeyear{connor1989intertemporal}, equation 1) allow stochastic volatilities and stochastic risk premiums in an intertemporal setting, the principal components-based or the macrovariables-based empirical estimations of the unconditional ICAPM and the unconditional intertemporal APT assume constant factor volatilities and constant risk premiums.} For example, Internet Appendix Section \ref{app:MertonVasicek} derives a full solution of the expected call price given by equation \eqref{eq:Mertonoptexp}, which is consistent with the Merton's (\citeyear{merton1973intertemporal}) 2-factor ICAPM with the Vasicek's (\citeyear{Vasicek1977An}) short rate process.

The expected option price solution in equation \eqref{eq:Mertonoptexp} can be obtained by simply replacing the current prices $S_t$ and $P(t,T)$ with the expected future prices $\E_t\left[S_H\right]$ and $\E_t\left[P(H,T)\right]$, respectively, everywhere in the \citeauthor{Merton1973Theory}'s (\citeyear{Merton1973Theory}) original formula.\footnote{Of course, we also allow a \emph{deterministic} correlation term in the Merton's formula to allow for multiple Gaussian factors.} Interestingly, this handy feature of the new analytical solution applies to even other option pricing models with deterministic volatilities and deterministic risk premiums for the underlying asset return processes, including:
i) the currency option price models of \cite{Grabbe1983The}, \cite{AMIN1991Pricing}, and \cite{hilliard1991currency}, in which the expected future price of the option written on the currency exchange rate $S$ (the domestic price of 1 unit of foreign currency) can be obtained by replacing the domestic currency-based current prices of the foreign pure discount bond and the domestic pure discount bond, that is, $S_tP_f(t,T)$ and $P_d(t,T)$, with the expected future prices  $\E_t\left[S_HP_f(H,T)\right]$ and $\E_t\left[P_d(H,T)\right]$, respectively, everywhere in the option price solution under these models; and ii) the \citeauthor{margrabe1978value}'s (\citeyear{margrabe1978value}) exchange option model, in which the expected future price of the option to exchange asset $S_{1}$ for asset $S_{2}$, can be obtained by replacing the current prices $S_{1t}$ and $S_{2t}$ with the expected future prices $\E_t\left[S_{1H}\right]$ and $\E_t\left[S_{2H}\right]$, respectively, everywhere in the \citeauthor{margrabe1978value}'s (\citeyear{margrabe1978value}) option price solution (see Internet Appendix Section \ref{iapp:margrabe}).

\subsection{The Collin-Dufresne and Goldstein Model}\label{sec:CDG1}

This section uses the EEM $\RTo$ to obtain the expected future price of a corporate bond under the \cite{collin2001credit} (CDG hereafter) structural model. The CDG model allows the issuing firm to continuously adjust its capital structure to maintain a stationary mean-reverting leverage ratio. The firm's asset return process and the short rate process under the CDG model are given as follows:
\begin{align*}
	\begin{aligned}
		\frac{\D S_s}{S_s}&= (r_s + \gamma^S \sigma)\D s + \sigma \left(\rho \D W_{1s}^{\p}+ \sqrt{1-\rho^2}\D W_{2s}^{\p}\right), \\
		\D r_s &= \alpha_r \left(m_r - r_s\right)\D s + \sigma_r \D W_{1s}^{\p},
	\end{aligned}
\end{align*}
where $\alpha_r, m_r, \sigma_r, \sigma>0$, $-1\leq\rho\leq1$, $\gamma^S \in \mathbb{R}$, $Z_{1}^{\p}$ and $Z_{2}^{\p}$ are independent Brownian motions under the physical measure $\p$. The market prices of risks associated with the two Brownian motions are given as
\begin{align*}
	\begin{aligned}
		\gamma_{1s} &= \gamma_r , \\
		\gamma_{2s} &= \frac{1}{\sqrt{1-\rho^2}}\left(-\rho \gamma_r + \gamma^S\right).
	\end{aligned}
\end{align*}

CDG assume that the face value of firm's total debt $F^D$ follows a stationary mean reverting process given as\footnote{The \cite{longstaff1995simple} model obtains as a special case of the CDG model when the total debt of the firm is constant.}
\begin{align*}
	\frac{\D F^D_s}{F^D_s} &= \lambda \left( \ln S_s - \nu - \phi r_s -\ln F^D_s \right) \D s.
\end{align*}

By defining the log-leverage ratio as $l \triangleq \ln (F^D/S) $, and using It\^{o}'s lemma, we obtain the physical process for the log-leverage ratio as
\begin{align*}
	\D l_s = \lambda \left( \bar{l} - \frac{\sigma\gamma^S}{\lambda}- \left( \frac{1}{\lambda} + \phi\right) r_s - l_s \right) \D s - \sigma\left(\rho \D W_{1s}^{\p}+ \sqrt{1-\rho^2} \D W_{2s}^{\p}\right),
\end{align*}
where $\bar{l}\triangleq \frac{\sigma^2}{2\lambda} - \nu$.

Then using Lemma \ref{lemma:reeo} and Internet Appendix Section \ref{app:DiffEEM}.\ref{app:Rforward}, the processes of the state variables $\left\{ l, r \right\}$ under the $\RTo$ measure are given as
\begin{align}\label{eq:CDGlandr}
	\begin{aligned}
		\D l_s &=  \lambda\left(\bar{l} - \left( \frac{1}{\lambda} + \phi\right) r_s -  l_s + \frac{\rho\sigma\sigma_r}{\lambda} B_{\alpha_r}{(T-s)} -\left( \frac{\rho\sigma\sigma_r}{\lambda} B_{\alpha_r}{(H-s)} + \frac{\sigma\gamma^S}{\lambda} \right)\I_{\{s< H\}} \right) \D s \\
	&\quad	 - \sigma\left(\rho \D \ZTone + \sqrt{1-\rho^2} \D \ZTtwo\right),\\
		\D r_s &= \alpha_r\left(m_r -\frac{ \sigma_r \gamma_r}{\alpha_r}  - r_s - \frac{\sigma_r^2}{\alpha_r}B_{\alpha_r}{(T-s)}+ \left(\frac{\sigma_r^2}{\alpha_r}B_{\alpha_r}{(H-s)}+ \frac{ \sigma_r \gamma_r}{\alpha_r} \right) \I_{\{s< H\}}\right)\D s  \\
& \quad + \sigma_r \D \ZTone,
	\end{aligned}
\end{align}
where $B_{\alpha}{(\tau)} = (1/\alpha)(1-\ue^{-\alpha \tau})$.

Consider a risky discount bond issued by the firm with a promised payment of $1$ at the bond maturity date $T$. Under the CDG framework, all bonds issued by the firm default simultaneously when the log-leverage ratio $l$ hits a pre-specified constant $\ln K$ from below.  In the event of default, the bondholder receives a fraction $1-\omega$ times the present value of a default-free zero coupon bond that pays 1 at the maturity date $T$. Under these payoff assumptions, the current price of a risky discount bond $D(t, T)$ maturing at time $T$ is given by \cite{collin2001credit} as
\begin{align*}%
	D(t,T) = P(t,T)\left(1 -  \omega\E_t^{\QT} \left[\I_{\left\{\tau_{[t,T]}<T\right\}}\right]\right),
\end{align*}
where $\QT$ is the forward measure, the variable $\tau_{[t,T]}$ is the first passage time to default during the time period $t$ to $T$, and $\E_t^{\QT} \left[\I_{\left\{\tau_{[t,T]}<T\right\}}\right]$ is the $\QT$ probability of the log-leverage ratio hitting the boundary $\ln K$ from below during this time period.

We derive the expected future bond price by generalizing the above equation using the $\RTo$ measure. Substituting the terminal payoff $\X_T = 1 - \omega \I_{\left\{\tau_{[t,T]}<T\right\}}$ in equation \eqref{eq:rep2} of Corollary \ref{coro:three2}, the expected future bond price is given as
\begin{align}\label{eq:gvth3}
	\E_t\left[ D(H,T)\right] & =\E_t^{\p}\left[ P(H, T)\right] \left(1 - \omega \E_t^{\RTo} \left[\I_{\left\{\tau_{[t,T]}<T\right\}} \right] \right), %
\end{align}
where $\E_t^{\RTo} \left[\I_{\left\{\tau_{[t,T]}<T\right\}} \right]$ is the $\RTo$ probability of the log-leverage ratio hitting the boundary $\ln K$ during the time period $t$ to $T$. This probability can be approximated in semi-closed form using a single numerical integral by generalizing \citeauthor{mueller2000simple}'s (\citeyear{mueller2000simple}) extension of the CDG formula for the risky bond price. The solution of the expected future bond price in equation \eqref{eq:gvth3} is given in Appendix \ref{app:CDG}.

\subsection{Non-Parametric and Semi-Parametric Applications}\label{sec:recov}

Until now we have developed the EEM theory in the context of parametric models assuming the \emph{same form} for the state variable processes under both the physical measure and the given EMM. This assumption does not always hold empirically. For example, the physical process and the risk-neutral process for the asset return cannot have the same form under the cumulative prospect theory model of option returns given by \cite{baele2019cumulative}. This section gives additional theoretical results related to the EEM theory which do not require strong parametric assumptions.

To derive these results we assume that the expected future prices of standard European options with a sufficiently large range of strike prices have been estimated using non-parametric or semi-parametric methods \citep[e.g., see][]{Coval2001Expected, bondarenko2003statistical, bondarenko2014put, jones2006nonlinear, driessen2007empirical,  bakshi2010returns, constantinides2013puzzle, israelov2017forecasting, baele2019cumulative, buchner2019latent}. Using these expected future option prices, we show how one can obtain \emph{both} the $\RTo$ density and the physical density of the underlying asset return using an extension of the theoretical framework of \cite{breeden1978prices}. Furthermore, we show that the $\RTo$ density can be used to obtain the expected future prices (and expected returns) of a wider range of complex European contingent claims with terminal payoffs that are arbitrary functions of the underlying asset's future price.

To fix ideas assume that the Arrow-Debreu price density or the state price density (SPD) of the underlying asset at current time $t$ is given by the function $f_t(S_T)$, where $S_T$ is the future price of this asset at time $T$. Also, assume that the future state price density (FSPD) of this asset at a future time $H$ is given by the function $f_H(S_T|Y_H)$, where $Y_H$ is a vector of state variables (such as future asset price $S_H$, future volatility, etc.) at time $H$, for all $t\leq H\leq T$. %
Let the physical expectation $g_t(H, S_T) \triangleq \E_t[f_H(S_T|Y_H)]$ define the expected FSPD.

Consider a European call option with a payoff equal to $C_T = \left(S_T - K \right)^+$ and a European put option with a payoff equal to $P_T = \left(K - S_T \right)^+$ at time $T$, and let $C_H$ and $P_H$ be the future prices of the call option and the put option, respectively, at a future time $H \leq T$. Theorem \ref{thm:Breedenext} presents an extension of the \cite{breeden1978prices} result, showing the relationship of the expected FSPD with the probability densities of $S_T$ under the $\pp$, $\Ro$ and $\RTo$ measures.

\begin{theorem}\label{thm:Breedenext}
	The expected FSPD, $g_t(H, S_T) \triangleq \E_t[f_H(S_T|Y_H)]$, is given by the second derivative of either the expected future call price $\E_t[C_H]$ or the expected future put price $\E_t[P_H]$ with respect to the strike price $K$, and is related to the probability densities of $S_T$ under the $\RTo$, $\Ro$, and $\pp$ measures, as follows:
	
	\begin{align}\label{eq:AD-RTo}
		g_t(H, S_T)& = \left[\frac{\partial^2 \E_t\left[C_H\right]}{\partial K^2}\right]_{K=S_T}= \left[\frac{\partial^2 \E_t\left[P_H\right]}{\partial K^2}\right]_{K=S_T} \ \ \\
		&=  \E_t^\p\left[P(H,T)\right]p_t^\RTo(S_T), \nonumber
		\ \ \text{for all} \ \ t \leq H \leq T, \\
		&=  p_t^\pp(S_T), \nonumber
		\ \ \text{for all} \ \ t \leq H = T,
	\end{align}
	which under the special case of constant or deterministic short rate becomes,
	\begin{align}\label{eq:AD-Ro}
		g_t(H, S_T)&=  \left[\ue^{-\int_H^T r_u\D u}\right]p_t^\Ro(S_T), \ \ \text{for all} \ \ t \leq H \leq T, \\
		&=  p_t^\pp(S_T), \nonumber
		\ \ \text{for all} \ \ t \leq H = T,
	\end{align}
	where $p_t^\RTo(S_T)$, $p_t^\Ro(S_T)$, and $p_t^\pp(S_T)$ are time $t$ probability densities of $S_T$ under $\RTo$, $\Ro$, and $\pp$  measures, respectively. Moreover, the expected future price of a contingent claim $\X$ with an arbitrary terminal payoff function $\X_T = h(S_T)$, can be numerically computed using the expected FSPD, $g_t(H, S_T)$ as follows:
	\begin{align}\label{eq:ArbitClaimExpectation}
		\begin{aligned}
			\E_t \left[ \X_H\right]   & = \int_0^{\infty}g_t(H,S)h(S)dS.
		\end{aligned}
	\end{align}
\end{theorem}

\begin{proof}
	See Appendix \ref{app:thm1} (Section \ref{app:espd}) for the proofs of equations \eqref{eq:AD-RTo} and \eqref{eq:AD-Ro}. The proof of equation \eqref{eq:ArbitClaimExpectation} follows by substituting $g_t(H,S_T)$ from equation \eqref{eq:AD-RTo} into equation \eqref{eq:rep2}, and simplifying.
\end{proof}

Equation \eqref{eq:AD-RTo} shows that the second derivative of the expected future call price function with respect to the strike price gives the expected FSPD, which equals the discounted $\RTo$ probability density, similar to how the SPD equals the discounted forward density $\QT$ in \cite{breeden1978prices}. If the expected FSPD can be econometrically estimated using the expected future prices of standard European call and put options,\footnote{We demonstrate how the expected FSPD can be extracted using the expected future prices of standard European call and put options in Internet Appendix Section \ref{sec:extractFSPD} by extending the methodology of \cite{jackwerth2004option}.} then it can be used to obtain the expected future prices (and expected returns) of a wider range of complex European contingent claims with terminal payoffs that are arbitrary functions of the underlying asset's future price as shown in equation \eqref{eq:ArbitClaimExpectation}. This can be useful if a researcher desires consistency between the expected returns of such complex claims with those of the standard European options without making strong parametric assumptions about the underlying state variable processes.

A surprisingly simple and overlooked result in the option pricing literature is given by the last equality in both equations \eqref{eq:AD-RTo} and \eqref{eq:AD-Ro}, according to which the expected \emph{terminal} FSPD (i.e., $g_t(T, S_T)$)---defined as the second derivative of the expected terminal call (or put) price with respect to the strike price---equals the physical density of the underlying asset's future price. This arbitrage-based, tautologically-exact relationship between the expected terminal FSPD and the physical density is especially relevant in the light of the recent theoretical and empirical research that investigates the relationship between the SPD and the physical density \cite[see][]{ross2015recovery, martin2019expected, jensen2019generalized, jackwerth2020does}.\footnote{%
	While the motivation of the Ross recovery paradigm is that only today's option price data contains useful forward-looking information about the distributon of the future asset prices, the motivation of using Theorem \ref{thm:Breedenext} is that expected future option prices estimated using ``recent" option price data (e.g., daily option prices from the past 3 to 6  months) contain useful forward-looking information about the distribution of future asset prices.} The empirical implications of Theorem \ref{thm:Breedenext} can be explored by future research on non-parametric and semi-parametric option return models.

\begin{center}
\section{Extensions and Conclusion} \label{sec:conclusion}\vspace{-1em}
\end{center}

This section considers some extensions and applications of the EEM theory and then concludes the paper. Until now the paper has focused only on those contingent claim models that admit an analytical solution to the claim's current price. As the first extension, consider a general claim (e.g., European, Barrier, Bermudan, quasi-American, etc.) that remains alive until time $H$, but the analytical solution of its current price does not exist. Numerical methods can be used to compute the expected future price of such a claim at time $H$ using the EEMs, such as the $\Ro$ measure. For example, consider a quasi-American put option which is exercisable on or after time $T_e \geq H$, until the option expiration date $T$ (i.e., $T_e \leq T$). The expected future price of this option at time $H$ is given as
\begin{align*}
	\E_t \left[P_H^A\right]  = \sup_{\tau \in \mathcal{T}_{T_e, T}} \E^\Ro_t \left[\ue^{-\int_H^\tau r_u\D u} P_\tau^A\right]
	\ \   \text{for}  \ \ 0\leq t\leq H \leq T_e \leq T,
\end{align*}
where $P_s^A$ is the price of the quasi-American put option at time $s$, for all $H \leq s \leq T$, $P_T^A = P_T = \max (K - S_T, 0)$, $\mathcal{T}$ is the set of all stopping times with respect to the filtration $\F$, and $\mathcal{T}_{T_e, T} \triangleq \left\{\tau\in \mathcal{T}|\p\left(\tau\in\left[T_e, T\right]\right)=1\right\}$ is the subset of $\mathcal{T}$. The above equation follows by applying Propositions \ref{prop:Ree1}\eqref{Ree:num1} and \ref{prop:Ree1}\eqref{Ree:num2}, and valuing the American option as an optimal stopping problem \cite[see][]{jacka1991optimal} from time $T_e$ until time $T$. The above example shows how the EEM theory can be applied for obtaining the expected future price of any type of claim at time $H$, assuming the claim remains alive until time $H$. Numerical methods such as trees, finite difference methods, and Monte Carlo methods can be used for pricing such claims under the $\Ro$ measure.

Second, consider those claims that may not be alive until time $H$ in some states of nature. For example, consider a standard American option that can be exercised at any time until the expiration date of the option. While the expected future price of an American option cannot be computed, since the option is not alive at time $H$ in some states of nature, the expected future price of a \emph{reinvestment strategy} that invests the value of the American option upon exercise in another security---for example, an $H$-maturity pure discount bond or a money market account---until time $H$, can be computed numerically using the $\Ro$ measure.%

Third, consider those claims that pay out a stream of cash flows from time $t$ until time $T$---for example, a bond that makes periodic coupon payments. The \emph{expected returns} of such claims over the horizon $H$ cannot be computed without knowing the \emph{reinvestment strategy} that invests the cash flows maturing \emph{before} time $H$. Like the case in the second point above,
all of the cash flows maturing before time $H$ can be reinvested in another security---for example, an $H$-maturity pure discount bond or a money market account---for the remaining duration until time $H$, and the expected future value of these reinvested cash flows can be computed numerically using the $\Ro$ measure.

Fourth, though this paper has focused on the derivation of the analytical solutions of only the expected returns of contingent claims, future research can use the EEM theory to derive analytical solutions of the variance, covariance, and other higher-order moments of contingent claim returns.\footnote{For example, see \cite{nawalkha2022sharpe}.} Also, various numerical methods, such as trees, finite difference methods, and Monte Carlo methods can be used to obtain the higher order moments of the return on a portfolio of claims under the $\Ro$ measure. Furthermore, various finite-horizon tail risk measures (such as Value-at-Risk, conditional VaR, expected shortfall, and others) for contingent claim portfolio returns can be obtained with  Monte Carlo simulations under the $\Ro$ measure. Thus, the EEMs such as the $\Ro$ measure can significantly advance the risk and return analysis of fixed income and financial derivative portfolios.

Finally, perhaps the most important empirical application of the analytical solutions of the expected future prices of contingent claims obtained using the EEM theory is that various methods, such as Markov Chain Monte Carlo, generalized method of moments \cite[see][]{hansen1982large}, and others become feasible for studying both the \emph{cross-section} and the \emph{term structure} of returns of a variety of contingent claims, such as Treasury bonds, corporate bonds, interest rate derivatives, credit derivatives, equity derivatives, and others. It is our hope that the EEM theory can shift the focus of research from mostly valuation to \emph{both} valuation and expected returns, for a variety of finite-maturity securities and claims used in the fixed income markets and the derivatives markets.

Before concluding we would like to underscore an important point regarding the derivations of the expected future prices of different contingent claims under various classes of models in the equity, interest rate, and credit risk areas, given in Sections \ref{sec:Rmeasure} and \ref{sec:R1Tmeasure}, and in Internet Appendix Sections \ref{app:someATSM}, \ref{sec:qtsm}, \ref{app:fsoexpected}, \ref{app:MertonVasicek}, and \ref{iapp:margrabe}. %
Since the EEMs and the $R$-transforms provide a \emph{single analytical solution} for the expected future price (which includes both the current price for $H = t$, and the expected future price for all $H > t$), and since a given EEM and a given $R$-transform always nest the associated EMM and the associated $Q$-transform, respectively, for the special case when $H=t$, all future work on the derivations of the analytical solutions of the expected future prices of contingent claims can use either the EEMs or the $R$-transforms instead of using the associated EMMs or the associated $Q$-transforms, respectively. Since the expected future price is the most relevant input, and in many cases the only input required for computing the expected return of a contingent claim, the use of the EEMs or the $R$-transforms kills two birds with one stone---it allows the derivation of a single analytical solution that does both \emph{valuation} and computes the \emph{expected return} of the claim.

\renewcommand{\theequation}{\thesection\arabic{equation}}
\setcounter{equation}{0}
\theoremstyle{jfapp}
\begin{appendix}

 \setcounter{corollary}{0}
 \setcounter{proposition}{0}
 \setcounter{lemma}{0}
 \setcounter{definition}{0}

\renewcommand{\thelemma}{A\arabic{lemma}}
\renewcommand{\thedefinition}{A\arabic{definition}}
\renewcommand{\theproposition}{A\arabic{proposition}}
\renewcommand{\thecorollary}{A\arabic{corollary}}

\section{The Proofs}
\subsection{Proof of Lemma \ref{lemma:Ree} }\label{app:thm1}
\begin{proof}
	To prove that $\Ree$ and $\pp$ are equivalent probability measures, and $\LHs$ is the corresponding Radon-Nikod\'{y}m derivative process, we show that $\LHs$ is a martingale; $\LHT$ is almost surely positive; and $\E[\LHT]=1$ \citep[see, e.g.,][Theorem 1.6.1 and Definition 1.6.3]{shreve2004stochastic}.
	
	Consider the following three cases to show that $\LHs$ is a martingale \cite[see, e.g.,][Chapter I]{borodin2002handbook}, or
	$\E\left[\mathcal{L}_{s_2}^*(H) | \mathcal{F}_{s_1}\right] = \mathcal{L}_{s_1}^*(H)$, for all $0\leq s_1 \leq s_2\leq T$.
	First, for all $0\leq s_1 \leq s_2 \leq H \leq T$, $\mathcal{L}_{s_1}^*(H) = \mathcal{L}_{s_2}^*(H) = 1$, hence $\E\left[\mathcal{L}_{s_2}^*(H) | \mathcal{F}_{s_1}\right] = \mathcal{L}_{s_1}^*(H)$.
	Second, for all $0\leq H \leq s_1 \leq s_2\leq T$,
	\begin{align*}
		\E\left[\mathcal{L}_{s_2}^*(H) | \mathcal{F}_{s_1}\right]& =  \E\left[\frac{L_{s_2}^*}{\LH}\,\Big|\, \mathcal{F}_{s_1}\right] =   \E\left[\frac{\E\left[\LT|\mathcal{F}_{s_2}\right]}{\LH}\,\Big|\, \mathcal{F}_{s_1}\right]=  \frac{ \E\left[\E\left[\LT|\mathcal{F}_{s_2}\right]| \mathcal{F}_{s_1}\right]}{\LH}   \\
& =  \frac{ \E\left[\LT| \mathcal{F}_{s_1}\right]}{\LH} = \frac{L_{s_1}^*}{\LH} = \mathcal{L}_{s_1}^*(H).
	\end{align*}
	Third, for all $0\leq s_1  \leq H \leq s_2 \leq T$,
	\begin{align*}
		\E\left[\mathcal{L}_{s_2}^*(H) | \mathcal{F}_{s_1}\right]& =  \E\left[\frac{L_{s_2}^*}{\LH}\,\Big|\, \mathcal{F}_{s_1}\right] =   \E\left[\E\left[\frac{L_{s_2}^*}{\LH} \,\Big|\,\mathcal{F}_H\right]\mathcal{F}_{s_1}\right] \\
& =  \E\left[\E\left[\frac{\E\left[\LT|\mathcal{F}_{s_2}\right]}{\LH}\,\Big|\,\mathcal{F}_H \right]\,\Big|\,\mathcal{F}_{s_1}\right] =  \E\left[\frac{\E\left[\E\left[\LT|\mathcal{F}_{s_2}\right]|\mathcal{F}_H \right]}{\LH}\,\Big|\,\mathcal{F}_{s_1}\right] \\
& =  \E\left[\frac{\E\left[\LT|\mathcal{F}_H \right]}{\LH}\,\Big|\,\mathcal{F}_{s_1}\right] = \E\left[\frac{\LH}{\LH}\,\Big|\,\mathcal{F}_{s_1}\right] = \mathcal{L}_{s_1}^*(H).
	\end{align*}
	Therefore, $\LHs$ is a martingale.
	
	Next, we show that $\LHT$ is almost surely positive and $\E[\LHT]=1$. First, note that $\LT$ is almost surely positive since $\pp$ and $\Qem$ are equivalent probability measures. Since $\LH =  \E[\LT|\mathcal{F}_H]$, then $\LH$ is also almost surely positive. Finally, by definition in equation \eqref{eq:LR}, $\LHT = \LT/\LH$ is also almost surely positive. Additionally, $\E[\LHT]=\E[ \LT/\LH] = \E[\E[\LT|\mathcal{F}_H]/\LH] = \E[ \LH/\LH] = 1$.
\end{proof}

\subsection{Proof of Proposition \ref{prop:Ree1}}\label{app:Reempro}
\begin{proof}
	For the first property, for all $0\leq t\leq H$, given $A\in \F_H$, we have $\Ree(A|\F_t) =\E_t^\Ree\left[\I_A \right]= \E_t^\pp\left[\LHH/\LHt\cdot\I_A   \right]=\E_t^\pp\left[(\LH/\LH)/1\cdot\I_A   \right]=\E_t^\pp\left[\I_A\right]=\pp(A|\F_t)$.
	
	For the second property, for all $H\leq s\leq T$, given $A\in \F_T$, we have
	\begin{align*}
		\begin{aligned}
			\Ree(A\,|\,\F_s) &= \E^\Ree_s\left[\I_A\right]=\E_s^\p\left[\frac{\LHT}{\LHs}\I_A\right] =\E_s^\p\left[\frac{\LT/\LH}{\Ls/\LH}\I_A\right] =\E_s^\p\left[\frac{\LT}{\Ls}\I_A\right] \\
			&=\E_s^\Qem\left[\I_A\right] =\Qem(A\,|\,\F_s).
		\end{aligned}
	\end{align*}

	The third property is a special case of the first property when $H=T$, while the fourth property is a special case of the second property when $H=t$.
	
	For the last property, we have
	\begin{align*}
		\E_t^\Ree\left[Z_T\right]  =\E_t^\Ree\left[\E_H^\Ree\left[Z_T\right]\right] =\E_t^\Ree\left[\E_H^\Qem\left[Z_T\right]\right] =  \E_t^\p\left[\E_H^\Qem\left[Z_T\right]\right],
	\end{align*}
	where the second equality is an application of the second property, the last equality is an application of the first property, since $\E_H^\Qem\left[Z_T\right]$ is a random variable that is $\F_H$ measurable.
\end{proof}

\subsection{Proof of Proposition \ref{prop:Ree2}}\label{app:prop2}
\begin{proof}
	For the first property, for all $0\leq H\leq s_1 \leq s_2 \leq T$,
	\begin{align*}
		\E_{s_1}^\Ree\left[\frac{\X_{s_2}}{\Y_{s_2}}\right] & =   \E_{s_1}^\pp\left[\frac{\Lc_{s_2}^*(H)}{\Lc_{s_1}^*(H)}\cdot\frac{\X_{s_2}}{\Y_{s_2}}\right] =
		\E_{s_1}^\pp\left[\frac{L_{s_2}^*/\LH}{L_{s_1}^*/\LH}\cdot\frac{\X_{s_2}}{\Y_{s_2}}\right]
        =  \E_{s_1}^\pp\left[\frac{L_{s_2}^*}{L_{s_1}^*}\cdot\frac{\X_{s_2}}{\Y_{s_2}}\right] \\
        & = \E_{s_1}^\Qem\left[\frac{\X_{s_2}}{\Y_{s_2}}\right] =\frac{\X_{s_1}}{\Y_{s_1}}.
	\end{align*}
	
	For the second property, for all $0\leq s_1\leq s_2\leq H\leq T$,
	\begin{align*}
		\E_{s_1}^\Ree\left[\frac{\X_{s_2}}{\Y_{s_2}}\right] & =   \E_{s_1}^\pp\left[\frac{\Lc_{s_2}^*(H)}{\Lc_{s_1}^*(H)}\cdot\frac{\X_{s_2}}{\Y_{s_2}}\right] =
		\E_{s_1}^\pp\left[\frac{1}{1}\cdot\frac{\X_{s_2}}{\Y_{s_2}}\right] =  \E_{s_1}^\pp\left[\frac{\X_{s_2}}{\Y_{s_2}}\right],
	\end{align*}
	which is generally not equal to $ \E_{s_1}^\Qem\left[{\X_{s_2}}/{\Y_{s_2}}\right]$, except when $\p = \Qem$.
\end{proof}

\subsection{Proof of Lemma \ref{lemma:reeo} }\label{app:thm2}
\begin{proof}
	We show that $\Reeo$ and $\pp$ are equivalent probability measures, and $\LHso$ is the corresponding Radon-Nikod\'{y}m derivative process similar to the above process.
	First, we show that $\LHso$ is a martingale using the following three cases.
	For $0 \leq s_1 \leq s_2 \leq H\leq T$,
	\begin{align*}
		\E\left[\Lc_{1s_2}^*(H) | \mathcal{F}_{s_1}\right]& =  \E\left[\frac{\E_{s_2}\left[\Y_H\right]}{\E_0\left[\Y_H\right]}\,\Big|\, \mathcal{F}_{s_1}\right] =   \frac{\E_{s_1}\left[\E_{s_2}\left[\Y_H\right]\right]}{\E_0\left[\Y_H\right]}  =\frac{\E_{s_1}\left[\Y_H\right]}{\E_0\left[\Y_H\right]}   = \mathcal{L}_{1s_1}^*(H).
	\end{align*}
	For $0  \leq H \leq s_1 \leq s_2\leq T$,
	\begin{align*}
		\E\left[\Lc_{1s_2}^*(H) | \mathcal{F}_{s_1}\right]& =  \E\left[\frac{\Y_H}{\E_0\left[\Y_H\right]}\cdot\frac{L_{s_2}^*}{\LH}\,\Big|\, \mathcal{F}_{s_1}\right] =   \frac{\Y_H}{\E_0\left[\Y_H\right]}\cdot\frac{\E_{s_1}\left[L_{s_2}^*\right]}{\LH}=\frac{\Y_H}{\E_0\left[\Y_H\right]}\cdot\frac{L_{s_1}^*}{\LH}\\
& = \mathcal{L}_{1s_1}^*(H).
	\end{align*}
	For $0  \leq s_1\leq H  \leq s_2\leq T$,
	\begin{align*}
		\E\left[\Lc_{1s_2}^*(H) | \mathcal{F}_{s_1}\right]& =  \E\left[\frac{\Y_H}{\E_0\left[\Y_H\right]}\cdot\frac{L_{s_2}^*}{\LH}\,\Big|\, \mathcal{F}_{s_1}\right] =  \E_{s_1}\left[\E_H\left[\frac{\Y_H}{\E_0\left[\Y_H\right]}\cdot\frac{L_{s_2}^*}{\LH}\right]\right] \\
   & = \E_{s_1}\left[\frac{\Y_H}{\E_0\left[\Y_H\right]}\cdot\frac{\E_H\left[L_{s_2}^*\right]}{\LH}\right]
		 =\E_{s_1}\left[\frac{\Y_H}{\E_0\left[\Y_H\right]}\cdot\frac{\LH}{\LH}\right] =\frac{\E_{s_1}\left[\Y_H\right]}{\E_0\left[\Y_H\right]}\\
& = \mathcal{L}_{1s_1}^*(H).
	\end{align*}
	
	Next, we show that $\LHTo$ is almost surely positive and $\E[\LHTo]=1$. Since the numeraire asset $\Y$ is positive, using similar logic as in the proof of Lemma \ref{lemma:Ree}, $\LHTo$ is almost surely positive, and
	\begin{align*}
		\E\left[\LHTo\right]&=\E_0\left[\frac{\Y_H}{\E_0\left[\Y_H\right]}\cdot\frac{L_{T}^*}{\LH}\right]=\E_0\left[\E_H\left[\frac{\Y_H}{\E_0\left[\Y_H\right]}\cdot\frac{L_{T}^*}{\LH}\right]\right] \\
& =\E_0\left[\frac{\Y_H}{\E_0\left[\Y_H\right]}\cdot\frac{\E_H\left[L_{T}^*\right]}{\LH}\right] =\E_0\left[\frac{\Y_H}{\E_0\left[\Y_H\right]}\cdot\frac{\LH}{\LH}\right]=1.
	\end{align*}
\end{proof}

\subsection{Proof of Theorem \ref{thm:Breedenext}}\label{app:espd}
\begin{proof}
	Using the future state price density $f_H(S_T|Y_H)$ at time $H$, where $Y_H$ represents the state variable vector at time $H$, the future price of a European call option expiring at time $T$, can be obtained as
	\begin{align}\label{eq:CHspd}
		C_H =  \int_{K}^\infty \left(S-K\right) f_H(S |Y_H)\D S.
	\end{align}
	
	Let $p_t^\p(Y_H)$ be the probability density of the state variable vector $Y_H$ conditioned on time $t$ information under $\p$. For any $\F_H$-measurable variable $Z_H = Z_H(Y_H)$, which depends on the state variable vector $Y_H$, its expectation under $\p$ is
	\begin{align}\label{eq:EZH}
		\E_t[Z_H] &  = \int Z_H \ p_t^\p(Y_H)\D Y_H.
	\end{align}
	The expected future price of the call option can be simplified as
	\begin{align}\label{eq:ECHAD}
		\begin{aligned}
			\E_t[C_H] &  = \int C_H \ p_t^\p(Y_H)\D Y_H, \\
			&  = \int  \left(\int_{K}^\infty \left(S-K\right) f_H(S|Y_H)\D S\right) p_t^\p(Y_H)\D Y_H \\
			& = \int_{K}^\infty \left(S-K\right) \left( \int f_H(S|Y_H)p_t^\p(Y_H)\D Y_H\right) \D S \\
			&  = \int_{K}^\infty \left(S-K\right) \E_t\left[f_H(S|Y_H)\right]\D S,
		\end{aligned}
	\end{align}
	where the first equality uses equation \eqref{eq:EZH}, the second equality uses equation \eqref{eq:CHspd}, the third equality changes the integral order, and the last equality uses equation \eqref{eq:EZH} again.
	
	Taking the second derivative of $\E_t(C_H)$ with respect to $K$ gives
	\begin{align*}
		\begin{aligned}
			\E_t\left[f_H(K|Y_H)\right] =  \frac{\partial^2 \E_t\left[C_H\right] }{\partial K^2} \ \ \text{or,} \ \  \E_t\left[f_H(S_T|Y_H)\right] = \left[\frac{\partial^2 \E_t\left[C_H\right]}{\partial K^2}\right]_{K=S_T},
		\end{aligned}
	\end{align*}
	which proves the first part of equation \eqref{eq:AD-RTo} in Theorem \ref{thm:Breedenext}. The same logic applies to the case of the put option, so we omit the proof.
	
	Recalling equation \eqref{eq:rep2} in Corollary \ref{coro:three2}, we can also write the expected future price of the call option as
	\begin{align}\label{eq:ECHRTo}
		\begin{aligned}
			\E_t[C_H] & = \E_t^\p\left[P(H,T) \right]\E_t^\RTo\left[\left(S_T-K\right)^+\right] \\
			& = \E_t^\p\left[P(H,T) \right] \int_K^\infty \left(S-K\right) p_t^\RTo(S)\D S,
		\end{aligned}
	\end{align}
	where $p_t^\RTo(S_T)$ is time $t$ probability density of $S_T$ under $\RTo$.
	
	Due to the equivalence of equations \eqref{eq:ECHAD} and \eqref{eq:ECHRTo}, we have
	\begin{align*}%
		g_t(H, S_T) \triangleq \E_t\left[f_H(S_T|Y_H)\right]  =  \E_t^\p\left[P(H,T)\right]p_t^\RTo(S_T),
	\end{align*}
	which proves the second part of equation \eqref{eq:AD-RTo} in Theorem \ref{thm:Breedenext}. Equation \eqref{eq:AD-Ro} is a special case of equation \eqref{eq:AD-RTo}, and we omit the proof.
\end{proof}

\section{One-dimensional Brownian Motion under the $\mathbb{R}$ Measure}\label{prof:BMR}

Consider the Radon-Nikod\'{y}m derivative processes $\LHs$ and $\Ls$ in Lemma \ref{lemma:Ree} and denote them as $\LHsB$ and $\LsB$, respectively, for the special case when the numeraire is the money market account, or $\Y = B \triangleq \ue^{\int_0^{\cdot} r_u \D u}$. For this special case, $\Ree$ = $\Ro$, as defined in Corollary \ref{coro:three}, and we know that $\LsB$ under the \cite{black1973pricing} model is given as
\begin{align}\label{eq:LSB}
	\LsB &\triangleq\frac{\D \q}{\D \pp}\Big|_{\mathcal{F}_s}=\exp\left(-\int_0^s {\gamma}\D \Wpu - \frac{1}{2}\int_0^s \gamma^2\D u\right),
\end{align}
where $\gamma = \left(\mu-r\right)/\sigma$ is the market price of risk (MPR) and $r$ is the constant risk-free rate.

Using equations \eqref{eq:LR} and \eqref{eq:LSB}, $\LHsB$ can be given as follows:
\begin{align*}
	\begin{aligned}
		\LHsB \triangleq\frac{\D \R}{\D \pp}\Big|_{\mathcal{F}_s}&=\left\{\begin{array}{ll}
			\displaystyle \frac{\LsB}{\LHB}= \exp\left(-\int_H^s {\gamma}\D \Wpu - \frac{1}{2}\int_H^s {\gamma}^2\D u\right), &\text{if} \ \   H\leq s\leq T, \\
			\displaystyle 1, &\text{if} \ \ 0\leq s < H, \\
		\end{array}\right. \\
		&=\exp\left(-\int_0^s \GuH\D \Wpu -
		\frac{1}{2}\int_0^s \left(\GuH\right)^2\D u\right),
	\end{aligned}
\end{align*}
where
\begin{align*}
	\GsH &\triangleq \I_{\left\{s\geq H\right\}}\gamma=\left\{\begin{array}{ll}
		\displaystyle \gamma, &\text{if} \ \ H\leq s\leq T, \\
		\displaystyle 0, &\text{if} \ \  0\leq s < H. \\
	\end{array}\right.
\end{align*}

The Girsanov theorem allows the transformation of Brownian motions under two equivalent probability measures. Therefore, given $\q$ and $\pp$, we know that $\Wq = \Wp + \int_0^s{\gamma} \D u$. Similarly, the relationship of the Brownian motions under the $\Ro$ and $\pp$ measures is given as
\begin{align*}%
	\Wr = \Wp + \int_0^s \GuH \D u = \Wp + \int_0^s\I_{\left\{u\geq H\right\}}{\gamma} \D u.
\end{align*}

\section{Expected Future Option Price under the Black-Scholes Model}\label{app:BSsolution}
The expected European call price under the Black-Scholes model is given as
\begin{align}\label{eq:BSexp}
	\E_t \left[ C_H \right]
	&= S_t {\ue^{ \mu(H-t)}} \mathcal{N}\left(\hat{d}_1\right) -K\ue^{-r(T-{H})} \mathcal{N}\left(\hat{d}_2\right),
\end{align}
where $\mathcal{N}(\cdot)$ is the standard normal cumulative distribution function, and
\begin{align*}
	\displaystyle \hat{d}_1 & =  \frac{\ln( S_t/K) +{\mu(H-t)} +  r(T-{H}) + \frac{1}{2}\sigma^2(T-t)}{\sigma\sqrt{ T-t}}, \\
	\hat{d}_2 & = \frac{\ln( S_t/K) +{\mu(H-t)} +  r(T-{H}) - \frac{1}{2}\sigma^2(T-t)}{\sigma\sqrt{ T-t}}.
\end{align*}

\section{Expected Future Price of Forward-Start Option using the Black-Scholes Model}\label{app:fsoexpectedBS}

Consider a forward-start call option (FSO) with a strike price $kS_{T_0}$ determined at a future time $T_0$ (for any $T_0 \leq T$), and a terminal payoff $(S_T - kS_{T_0})$ at time $T$. Let $C(S_t, K, \sigma, r, t, T) \triangleq C_t = \E_t^\q\left[\ue^{-r(T-t)}(S_T - K)^+\right]$ be the Black-Scholes call price function. Using the homogeneity property of the call price, \cite{rubinstein1991pay} obtains the price of the FSO as follows:
\begin{align*}
	\text{FSO}_t & = \E_t^\q\left[\ue^{-r(T_0-t)}\E^\q_{T_0}\left[\ue^{-r(T-T_0)}(S_T - kS_{T_0})^+\right]\right]\\
	& = \E_t^\q\left[\ue^{-r(T_0-t)}C(S_{T_0}, kS_{T_0}, \sigma, r, T_0, T)\right] \\
	& = \E_t^\q\left[\ue^{-r(T_0-t)}S_{T_0}C(1, k, \sigma, r, T_0, T)\right] \\
	& = S_tC(1, k, \sigma, r, T_0, T).
\end{align*}

Since the FSO can be interpreted as a constant share $C(1, k, \sigma, r, T_0, T)$ of the asset $S$, until time $T_0$, the expected future price of the FSO when $t\leq H \leq T_0 \leq T$, obtains trivially by taking the expected future value of the asset at time $H$, as follows:\footnote{Alternatively, the result in equation \eqref{eq:RubTrivial} can also be derived by applying the $\R$ measure using Corollary \ref{coro:three}. }
\begin{align}\label{eq:RubTrivial}
	\E_t[\text{FSO}_H]  = S_t\ue^{\mu(H-t)}C(1, k, \sigma, r, T_0, T).
\end{align}

The non-trivial case occurs when $t\leq T_0\leq H  \leq T$. Under this case, the $\R$ measure can be used to compute the expected future price of the FSO by applying Corollary \ref{coro:three} as follows:
\begin{align*}
	\E_t[\text{FSO}_H] &= \E_t^\Ro\left[\ue^{-r(T-H)}(S_T - kS_{T_0})^+\right].
\end{align*}

Let $\bar{C}(S_t, K, \sigma, r, \mu, t, H, T) \triangleq \E_t[C_H] = \E_t^\Ro\left[\ue^{-r(T-H)}(S_T - K)^+\right]$ be the Black-Scholes expected future call price function given by equation \eqref{eq:BSexp}. Using the homogeneity property of the expected future call price, and using Proposition \ref{prop:Ree1}\eqref{Ree:num1}, we can solve the expected future price of the FSO given in the above equation as follows:
\begin{align}\label{eq:RubnonTrivial}
	\E_t[\text{FSO}_H] & = \E_t^\Ro\left[\E^\R_{T_0}\left[\ue^{-r(T-H)}(S_T - kS_{T_0})^+\right]\right]\nonumber \\
	& = \E_t^\Ro\left[\bar{C}(S_{T_0}, kS_{T_0}, \sigma, r, \mu, T_0, H, T) \right] \nonumber\\
	& = \E_t^\Ro\left[S_{T_0}\bar{C}(1, k, \sigma, r, \mu, T_0, H, T) \right] \nonumber\\
	& = \E_t^\pp\left[S_{T_0}\bar{C}(1, k, \sigma, r, \mu, T_0, H, T) \right] \nonumber \\
	& = S_t \ue^{\mu (T_0-t)} \bar{C}(1, k, \sigma, r, \mu, T_0, H, T).
\end{align}
Together, equations \eqref{eq:RubTrivial} and \eqref{eq:RubnonTrivial} provide the full solution to the expected future price of the forward-start call option using: i) the Black-Scholes call price function when $H \leq T_0$, and ii) the Black-Scholes expected future call price function when $H \geq T_0$, respectively.

\section{Solutions of the $R$-transform and the Extended $R$-transform}\label{app:soluRandExtT}
Here we provide the solutions of the $R$-transform and the extended $R$-transform by giving the following two propositions:

\begin{proposition}\label{prop:Rtrans}
	Under technical regularity conditions given in \cite{Duffie2000Transform}, the $R$-transform $\phi^R$ of $\mymat{Y}$ defined by equation \eqref{eq:Rtransform} is given by
	\begin{align}\label{eq:Rtransformsol}
		\begin{aligned}
			\phi^R(\mymat{z};  t, T, H)
			=\exp\left(-A_{\mymat{\chi}}^{(\mybmm{b}, \mybmm{c})}(H-t) -\mymat{B_{\mymat{\chi}}^{(\mybmm{b}, \mybmm{c})}(H-t)}'\mymat{Y_t} \right),
		\end{aligned}
	\end{align}
	where $ \mybmm{b} =\left({A_{\mymat{\chi^*}}^{(\mybmm{b^*}, \mybmm{c^*})}(T-H),  \mymat{B_{\mymat{\chi^*}}^{(\mybmm{b^*}, \mybmm{c^*})}}(T-H)} \right)$, $\mybmm{c} = \left(0, 0\right)$, $\mybmm{b^*}=\left(0, \mymat{-z}\right)$, and $\mybmm{c^*}=\left(\rho_0, \mymat{\rho_1}\right)$.
	
	For any well-defined $\mybmm{b}\triangleq\left(b_0, \mymat{b_1} \right)\in \mathcal{C}^{N+1}$, $\mybmm{c}\triangleq\left(c_0, \mymat{c_1}\right) \in \mathcal{R}^{N+1}$ and a coefficient characteristic $\mymat{\chi} = (\mymat{k}, \mymat{h}, \mymat{l}, \theta)$, $A_{\mymat{\chi}}^{(\mybmm{b}, \mybmm{c})}(\tau)$ and $ \mymat{B_{\mymat{\chi}}^{(\mybmm{b}, \mybmm{c})}(\tau)}$ are defined by the solution of the following expectation with $\tau<H-t$,
	\begin{align*}
		& \E_t^\R\left[\exp\left(-b_0-\mymat{b_1}'\mymat{Y_{t+\tau}}-\int_t^{t+\tau}\left(c_0 +\mymat{c_1}'\mymat{Y_u} \right)\D u\right)\right] \\
		= & \E_t^\p\left[\exp\left(-b_0-\mymat{b_1}'\mymat{Y_{t+\tau}}-\int_t^{t+\tau}\left(c_0 +\mymat{c_1}'\mymat{Y_u} \right)\D u\right)\right] \\
		= & \exp\left(-A^{(\mybmm{b}, \mybmm{c})}_{\mymat{\chi}}(\tau)- \mymat{B^{(\mybmm{b}, \mybmm{c})}_{\mymat{\chi}}(\tau)}'\mymat{Y_t}\right),
	\end{align*}
	and are obtained by solving the following complex-valued ODEs:
	\begin{align}\label{eq:ODEsAB}
		\begin{aligned}
			\frac{\D{A^{(\mybmm{b}, \mybmm{c})}_{\mymat{\chi}}(\tau)}}{\D \tau} &=   \mymat{k_0'B^{(\mybmm{b}, \mybmm{c})}_{\mymat{\chi}}(\tau)} - \frac{1}{2}\mymat{B^{(\mybmm{b}, \mybmm{c})}_{\mymat{\chi}}(\tau)'h_0 B^{(\mybmm{b}, \mybmm{c})}_{\mymat{\chi}}(\tau)} - l_0\left(\theta\left(\mymat{-B^{(\mybmm{b}, \mybmm{c})}_{\mymat{\chi}}(\tau)}\right)-1\right)+
			c_0, \\
			\frac{\D{\mymat{B^{(\mybmm{b}, \mybmm{c})}_{\mymat{\chi}}(\tau)}}}{\D \tau} &=   \mymat{k_1'B^{(\mybmm{b}, \mybmm{c})}_{\mymat{\chi}}(\tau)} - \frac{1}{2}\mymat{B^{(\mybmm{b}, \mybmm{c})}_{\mymat{\chi}}(\tau)'h_1 B^{(\mybmm{b}, \mybmm{c})}_{\mymat{\chi}}(\tau)} - \mymat{l_1}\left(\theta\left(\mymat{-B^{(\mybmm{b}, \mybmm{c})}_{\mymat{\chi}}(\tau)}\right)-1\right) +\mymat{c_1},
		\end{aligned}
	\end{align}
	with boundary conditions $A^{(\mybmm{b}, \mybmm{c})}_{\mymat{\chi}}(0)=b_0$, and $\mymat{B^{(\mybmm{b}, \mybmm{c})}_{\mymat{\chi}}(0)}=\mymat{b_1}$. The solutions can be found numerically, for example by Runge-Kutta. $A_{\mymat{\chi^*}}^{(\mybmm{b^*}, \mybmm{c^*})}(\tau)$ and $\mymat{B_{\mymat{\chi^*}}^{(\mybmm{b^*}, \mybmm{c^*})}(\tau)}$ are defined by the solution of the following expectation,
	\begin{align*}
		& 	\E_H^\R\left[\exp\left(-b_0-\mymat{b_1}'\mymat{Y_{H+\tau}}-\int_H^{H+\tau}\left(c_0 +\mymat{c_1}'\mymat{Y_u} \right)\D u\right)\right] \\
		= &\E_H^\q\left[\exp\left(-b_0-\mymat{b_1}'\mymat{Y_{H+\tau}}-\int_H^{H+\tau}\left(c_0 +\mymat{c_1}'\mymat{Y_u} \right)\D u\right)\right] \\
		= & \exp\left(-A^{(\mybmm{b^*}, \mybmm{c^*})}_{\mymat{\chi^*}}(\tau)- \mymat{B^{(\mybmm{b^*}, \mybmm{c^*})}_{\mymat{\chi^*}}(\tau)}'\mymat{Y_H}\right),
	\end{align*}
	and are obtained similarly as in equation \eqref{eq:ODEsAB} with the coefficient characteristic $\mymat{\chi} = (\mymat{k}, \mymat{h}, \mymat{l}, \theta)$ replaced with a coefficient characteristic $\mymat{\chi^*} = (\mymat{k^*}, \mymat{h^*}, \mymat{l^*}, \theta^*)$, and with the parameter $\mybmm{b}$ and $\mybmm{c}$ replaced with $\mybmm{b^*}\triangleq\left(b_0^*, \mymat{b_1^*}\right)$ and $\mybmm{c^*}\triangleq\left(c_0^*, \mymat{c_1^*}\right)$.
\end{proposition}
\begin{proof}
	Under the integration conditions in the definition of \citet[][page  1351]{Duffie2000Transform}, we can prove that the solution of the expectation $\E_t^\p\left[\exp\left(-b_0-\mymat{b_1}'\mymat{Y_{t+\tau}}-\int_t^{t+\tau}\left(c_0 +\mymat{c_1}'\mymat{Y_u} \right)\D u\right)\right] $ is  $\exp\left(-A^{(\mybmm{b}, \mybmm{c})}_{\mymat{\chi}}(\tau)- \mymat{B^{(\mybmm{b}, \mybmm{c})}_{\mymat{\chi}}(\tau)}'\mymat{Y_t}\right)$, and the solution of the expectation $\E_H^\q\Big[\exp\Big(-b_0-\mymat{b_1}'\mymat{Y_{H+\tau}}\Big.\Big.$
	$\left.\left. -\int_H^{H+\tau}\left(c_0 +\mymat{c_1}'\mymat{Y_u} \right)\D u\right)\right]$ is $\exp\left(-A^{(\mybmm{b^*}, \mybmm{c^*})}_{\mymat{\chi^*}}(\tau)- \mymat{B^{(\mybmm{b^*}, \mybmm{c^*})}_{\mymat{\chi^*}}(\tau)}'\mymat{Y_H}\right)$, by following the proof of Proposition 1 in \cite{Duffie2000Transform}. Note that under the $\R$ measure, the coefficient characteristic is ${\chi}$ from $t$ to $H$ and ${\chi^*}$ from $H$ to $T$. Therefore, for the $R$-transform
	\begin{align*}
		\phi^R(\mymat{z}; t, T, H) &\triangleq \E_t^\RR\left[\exp{\left(-\int_H^T r(\mymat{Y_u})\D u\right)}\exp\left( \mymat{z'Y_T}\right)\right]   \\
		& = \E_t^\RR\left[\E_H^\RR\left[\exp{\left(\mymat{z'Y_T} -\int_H^T \left(\rho_0 + \mymat{\rho_1'Y_u}\right) \D u \right)}\right]\right] \\
		& = \E_t^\p\left[\E_H^\q\left[\exp{\left(\mymat{z'Y_T} -\int_H^T \left(\rho_0 + \mymat{\rho_1'Y_u}\right) \D u \right)}\right]\right] \\
		& = \E_t^\p\left[\exp{\left( -{A_{\mymat{\chi^*}}^{(\mybmm{b^*}, \mybmm{c^*})}(T-H)-  \mymat{B_{\mymat{\chi^*}}^{(\mybmm{b^*}, \mybmm{c^*})}}(T-H)}'\mymat{Y_H} \right)}\right]\\
		& = \exp\left(-A_{\mymat{\chi}}^{(\mybmm{b}, \mybmm{c})}(H-t) -\mymat{B_{\mymat{\chi}}^{(\mybmm{b}, \mybmm{c})}(H-t)}'\mymat{Y_t} \right),
	\end{align*}
	with  $\mybmm{b^*}=\left(0, \mymat{-z}\right)$,  $\mybmm{c^*}=\left(\rho_0, \mymat{\rho_1}\right)$, $ \mybmm{b} =\left({A_{\mymat{\chi^*}}^{(\mybmm{b^*}, \mybmm{c^*})}(T-H), \mymat{B_{\mymat{\chi^*}}^{(\mybmm{b^*}, \mybmm{c^*})}}(T-H)} \right)$, and $\mybmm{c} = \left(0, 0\right)$.
\end{proof}

\begin{proposition}\label{prop:Rtransext}
	Under technical regularity conditions given in \cite{Duffie2000Transform}, the extended $R$-transform $\varphi^R$ of $\mymat{Y}$ defined by equation \eqref{eq:Rtransform2} is given by
	\begin{align}\label{eq:Rtransformsol2}
		\begin{aligned}
			\varphi^R(\mymat{v}, \mymat{z};  t, T, H) =\phi^R(\mymat{z};  t, T, H)\left(D_{\mymat{\chi}}^{(\mybmm{b}, \mybmm{c}, \mybmm{d})}(H-t) +\mymat{E_{\mymat{\chi}}^{(\mybmm{b}, \mybmm{c}, \mybmm{d})}(H-t)}'\mymat{Y_t} \right),
		\end{aligned}
	\end{align}
	where $ \mybmm{b} =\left({A_{\mymat{\chi^*}}^{(\mybmm{b^*}, \mybmm{c^*})}(T-H),  \mymat{B_{\mymat{\chi^*}}^{(\mybmm{b^*}, \mybmm{c^*})}}(T-H)} \right)$, $\mybmm{c} = \left(0, 0\right)$, $ \mybmm{d} =\left({D_{\mymat{\chi^*}}^{(\mybmm{b^*}, \mybmm{c^*}, \mybmm{d^*})}(T-H),  \mymat{E_{\mymat{\chi^*}}^{(\mybmm{b^*}, \mybmm{c^*}, \mybmm{d^*})}}(T-H)} \right)$, $\mybmm{b^*}=\left(0, \mymat{-z}\right)$, $\mybmm{c^*}=\left(\rho_0, \mymat{\rho_1}\right)$, and $\mybmm{d^*}=\left(0, \mymat{v}\right)$.
	
	For any well-defined $\mybmm{b}\triangleq\left(b_0, \mymat{b_1} \right)\in \mathcal{C}^{N+1}$, $\mybmm{c}\triangleq\left(c_0, \mymat{c_1}\right) \in \mathcal{R}^{N+1}$, $\mybmm{d}\triangleq\left(d_0, \mymat{d_1} \right)\in \mathcal{R}^{N+1}$, and a coefficient characteristic $\mymat{\chi} = (\mymat{k}, \mymat{h}, \mymat{l}, \theta)$, $A_{\mymat{\chi}}^{(\mybmm{b}, \mybmm{c})}(\tau)$, $\mymat{B_{\mymat{\chi}}^{(\mybmm{b}, \mybmm{c})}(\tau)}$, $ D_{\mymat{\chi}}^{(\mybmm{b}, \mybmm{c}, \mybmm{d})}(\tau)$ and $ \mymat{E_{\mymat{\chi}}^{(\mybmm{b}, \mybmm{c}, \mybmm{d})}(\tau)}$ are defined by the solution of the following expectation with $\tau<H-t$,
	\begin{align*}
		&\E_t^\R\left[\exp\left(-b_0-\mymat{b_1}'\mymat{Y_{t+\tau}}-\int_t^{t+\tau}\left(c_0 +\mymat{c_1}'\mymat{Y_u} \right)\D u\right)\left(d_0 + \mymat{d_1}'\mymat{Y_{t+\tau}}\right)\right] \\
		=&\E_t^\p\left[\exp\left(-b_0-\mymat{b_1}'\mymat{Y_{t+\tau}}-\int_t^{t+\tau}\left(c_0 +\mymat{c_1}'\mymat{Y_u} \right)\D u\right)\left(d_0 + \mymat{d_1}'\mymat{Y_{t+\tau}}\right)\right] \\
		=& \exp\left(-A^{(\mybmm{b}, \mybmm{c})}_{\mymat{\chi}}(\tau)- \mymat{B^{(\mybmm{b}, \mybmm{c})}_{\mymat{\chi}}(\tau)}'\mymat{Y_t}\right)\cdot\left(D^{(\mybmm{b}, \mybmm{c}, \mybmm{d})}_{\mymat{\chi}}(\tau)+ \mymat{E^{(\mybmm{b}, \mybmm{c}, \mybmm{d})}_{\mymat{\chi}}(\tau)}'\mymat{Y_t}\right),
	\end{align*}
	where $A^{(\mybmm{b}, \mybmm{c})}_{\mymat{\chi}}(\tau)$ and $\mymat{B^{(\mybmm{b}, \mybmm{c})}_{\mymat{\chi}}(\tau)}$ are obtained as in Proposition \ref{prop:Rtrans}, and $D^{(\mybmm{b}, \mybmm{c}, \mybmm{d})}_{\mymat{\chi}}(\tau)$ and $\mymat{E^{(\mybmm{b}, \mybmm{c}, \mybmm{d})}_{\mymat{\chi}}(\tau)}$ satisfy the complex-valued ODEs
	\begin{align}\label{ExRODE}
		\begin{aligned}
			\frac{\D{D^{(\mybmm{b}, \mybmm{c}, \mybmm{d})}_{\mymat{\chi}}(\tau)}}{\D \tau} &= \mymat{k_0'E^{(\mybmm{b}, \mybmm{c}, \mybmm{d})}_{\mymat{\chi}}(\tau)} - \mymat{B^{(\mybmm{b}, \mybmm{c})}_{\mymat{\chi}}(\tau)'h_0 E^{(\mybmm{b}, \mybmm{c}, \mybmm{d})}_{\mymat{\chi}}(\tau)} + l_0
			\nabla\theta\left(\mymat{-B^{(\mybmm{b}, \mybmm{c})}_{\mymat{\chi}}(\tau)}\right) \mymat{E^{(\mybmm{b}, \mybmm{c}, \mybmm{d})}_{\mymat{\chi}}(\tau)}, \\
			\frac{\D{\mymat{E^{(\mybmm{b}, \mybmm{c}, \mybmm{d})}_{\mymat{\chi}}(\tau)}}}{\D \tau} &=   \mymat{k_1'E^{(\mybmm{b}, \mybmm{c}, \mybmm{d})}_{\mymat{\chi}}(\tau)} -\mymat{B^{(\mybmm{b}, \mybmm{c})}_{\mymat{\chi}}(\tau)'h_1 E^{(\mybmm{b}, \mybmm{c}, \mybmm{d})}_{\mymat{\chi}}(\tau)} + \mymat{l_1}\nabla\theta \left(\mymat{-B^{(\mybmm{b}, \mybmm{c})}_{\mymat{\chi}}(\tau)}\right)\mymat{E^{(\mybmm{b}, \mybmm{c}, \mybmm{d})}_{\mymat{\chi}}(\tau)},
		\end{aligned}
	\end{align}
	with boundary conditions $D^{(\mybmm{b}, \mybmm{c}, \mybmm{d})}_{\mymat{\chi}}(0)=d_0$, and $\mymat{E^{(\mybmm{b}, \mybmm{c}, \mybmm{d})}_{\mymat{\chi}}(0)}=\mymat{d_1}$, and $\nabla\theta(\mymat{x})$ is gradient of $\theta(\mymat{x})$ with respect to $\mymat{x} \in \mathcal{C}^N$. $A_{\mymat{\chi^*}}^{(\mybmm{b^*}, \mybmm{c^*})}(\tau)$, $\mymat{B_{\mymat{\chi^*}}^{(\mybmm{b^*}, \mybmm{c^*})}(\tau)}$, $D_{\mymat{\chi^*}}^{(\mybmm{b^*}, \mybmm{c^*}, \mybmm{d^*})}(\tau)$ and $ \mymat{E_{\mymat{\chi^*}}^{(\mybmm{b^*}, \mybmm{c^*}, \mybmm{d^*})}(\tau)}$ are defined by the solution of the following expectation,
	\begin{align*}
		&\E_H^\R\left[\exp\left(-b_0-\mymat{b_1}'\mymat{Y_{H+\tau}}-\int_H^{H+\tau}\left(c_0 +\mymat{c_1}'\mymat{Y_u} \right)\D u\right)\left(d_0 + \mymat{d_1}'\mymat{Y_{H+\tau}}\right)\right] \\
		=&\E_H^\q\left[\exp\left(-b_0-\mymat{b_1}'\mymat{Y_{H+\tau}}-\int_H^{H+\tau}\left(c_0 +\mymat{c_1}'\mymat{Y_u} \right)\D u\right)\left(d_0 + \mymat{d_1}'\mymat{Y_{H+\tau}}\right)\right] \\
		=& \exp\left(-A^{(\mybmm{b^*}, \mybmm{c^*})}_{\mymat{\chi^*}}(\tau)- \mymat{B^{(\mybmm{b^*}, \mybmm{c^*})}_{\mymat{\chi^*}}(\tau)}'\mymat{Y_H}\right)\cdot\left(D^{(\mybmm{b^*}, \mybmm{c^*}, \mybmm{d^*})}_{\mymat{\chi^*}}(\tau)+ \mymat{E^{(\mybmm{b^*}, \mybmm{c^*}, \mybmm{d^*})}_{\mymat{\chi^*}}(\tau)}'\mymat{Y_H}\right),
	\end{align*}
	where $A^{(\mybmm{b^*}, \mybmm{c^*})}_{\mymat{\chi^*}}(\tau)$ and $\mymat{B^{(\mybmm{b^*}, \mybmm{c^*})}_{\mymat{\chi^*}}(\tau)}$ are obtained as in Proposition \ref{prop:Rtrans}, and $D^{(\mybmm{b^*}, \mybmm{c^*}, \mybmm{d^*})}_{\mymat{\chi^*}}(\tau)$ and $\mymat{E^{(\mybmm{b^*}, \mybmm{c^*}, \mybmm{d^*})}_{\mymat{\chi^*}}(\tau)}$  are obtained similarly as in equation \eqref{ExRODE} with the coefficient characteristic $\mymat{\chi} = (\mymat{k}, \mymat{h}, \mymat{l}, \theta)$ replaced with  a coefficient characteristic $\mymat{\chi^*} = (\mymat{k^*}, \mymat{h^*}, \mymat{l^*}, \theta^*)$, and with the parameter $\mybmm{b}$, $\mybmm{c}$, and $\mybmm{d}$ replaced with $\mybmm{b^*}\triangleq\left(b_0^*, \mymat{b_1^*}\right)$, $\mybmm{c^*}\triangleq\left(c_0^*, \mymat{c_1^*}\right)$, and $\mybmm{d^*}\triangleq\left(d_0^*, \mymat{d_1^*} \right)$.
\end{proposition}

\begin{proof}
	The proof follows similarly from the proof of Proposition \ref{prop:Rtrans}, we omit the specific steps.
\end{proof}

\section{Expected Future Corporate Bond Price under the CDG Model}\label{app:CDG}

To compute the default probability term $\E_t^{\RTo} \left[\I_{\left\{\tau_{[t,T]}<T\right\}} \right]$ in equation \eqref{eq:gvth3}, we follow the method of \cite{mueller2000simple} which requires the conditional mean and conditional variance of the log-leverage ratio $l$. As can be seen from equation \eqref{eq:CDGlandr}, the variable $l$ follows a Gaussian distribution under the $\RTo$ measure. Therefore, we have
\begin{align*}
	\E_t^{\RTo} \left[\I_{\left\{\tau_{[t,T]}<T\right\}} \right]=\Pr(\tau_{[t,T]}<T)=\sum_{i=1}^n q_i\left(\tau_{\ln K} = t+t_i | l_t, r_t, t\right) = \sum_{i=1}^n q_i,
\end{align*}
where
\begin{align*}
	q_1 &=\frac{\mathcal{N}(a_1)}{\mathcal{N}(b_{1,1})}, \quad   q_i = \frac{\mathcal{N}(a_i) - \sum\limits_{j=1}^{i-1} \mathcal{N}(b_{i,j})\times q_j }{\mathcal{N}(b_{i,i})},
\end{align*}
with
\begin{align*}
	\mathcal{N}(a_i) &\triangleq  \Pr \left(l_{t+i \Delta}\geq \ln K | l_t, r_t, t\right) \\
	&= \mathcal{N}\left(\frac{ M(t, t+i\Delta, H, T) - \ln K}{\sqrt{S(t, t+i\Delta)}} \right), \\
	\mathcal{N}(b_{i,j})&\triangleq \Pr \left(l_{t+i \Delta} \geq \ln K | l_{t+t_j}=\ln K, l_t, r_t, t\right) \\
	&= \mathcal{N}\left(\frac{ \tilde{M}(t, t+i\Delta, t+t_j, \ln K, H, T) -\ln K}{\sqrt{\tilde{S}(t, t+i\Delta, t_j)}} \right).
\end{align*}
The log-leverage ratio $(l_u | l_t, r_t, t) \sim \mathcal{N}\left(M(t,u,H,T)\right.$, $\left.S(t, u)\right)$, for $t<u<T$. The conditional mean $M(t, u, H, T)$ and the conditional variance $S(t, u)$ can be computed as
\begin{align*}
	M(t,u,H,T) &= l_t\ue^{-\lambda(u-t)} - \left(1 + \lambda\phi \right)\left( r_t - m_r \right) \ue^{-\alpha_r\left(u-t\right)}B_{\left(\lambda - \alpha_r\right)}{\left(u-t\right)} \\
	&- \left( \frac{\rho\sigma\sigma_r}{\alpha_r} + \left(1 + \lambda\phi \right)\frac{\sigma_r^2}{2{\alpha_r}^2} \right)\left( \ue^{-\alpha_r \left(T-u\right)} - \ue^{-\alpha_r\left(H-u\right)} \right) B_{\left(\lambda + \alpha_r\right)}{(u-t)} \\
	&+ \left(1 + \lambda\phi \right)\frac{\sigma_r^2}{2{\alpha_r}^2} \left(  \ue^{-\alpha_r \left(T-u\right)} - \ue^{-\alpha_r\left(H-u\right)}  \right) \ue^{-2\alpha_r\left(u-t\right)}B_{\left(\lambda - \alpha_r\right)}{(u-t)} \\
	& + \left(  \lambda \bar{l} - \sigma\gamma^S - \left(1+\lambda\phi\right)m_r \right)B_{\lambda}{(u-t)}  \\
	&+ \I_{\{u\geq H\}} \left[ -\left(1 + \lambda\phi \right)\left(\frac{\sigma_r^2}{2{\alpha_r}^2} + \frac{\sigma_r\gamma_r}{\alpha_r} \right)\ue^{-\alpha_r\left(u-H\right)}B_{\left(\lambda - \alpha_r\right)}{\left(u-H\right)} \right. \\
	& - \left( \frac{\rho\sigma\sigma_r}{\alpha_r}  +   \left(1 + \lambda\phi \right)\frac{\sigma_r^2}{2{\alpha_r}^2} \right)\ue^{-\alpha_r\left(H-u\right)} B_{\left(\lambda + \alpha_r\right)}{(u-H)}  \\
	& + \left. \left( \frac{\rho\sigma\sigma_r}{\alpha_r} + \sigma\gamma^S + \left(1 + \lambda\phi \right)\left(\frac{\sigma_r^2}{{\alpha_r}^2} + \frac{\sigma_r\gamma_r}{\alpha_r}   \right)  \right)B_{\lambda}{\left(u-H\right)}  \right], \\
	S(t,u) &= \left( \frac{\left(1 + \lambda\phi \right)\sigma_r}{\lambda - \alpha_r}  \right)^2 B_{2\alpha_r}{\left(u-t\right)}  \\
	& +\left( \sigma^2 +  \left( \frac{\left(1 + \lambda\phi \right)\sigma_r}{\lambda - \alpha_r}  \right)^2 - \frac{2\rho\sigma\sigma_r \left(1 + \lambda\phi \right)}{\lambda - \alpha_r}  \right)B_{2\lambda}{\left(u-t\right)} \\
	& + 2\left(\frac{\rho\sigma\sigma_r \left(1 + \lambda\phi \right)}{\lambda - \alpha_r} - \left( \frac{\left(1 + \lambda\phi \right)\sigma_r}{\lambda - \alpha_r}  \right)^2   \right)B_{\left(\lambda+\alpha_r\right)}{\left(u-t\right)}.
\end{align*}

In addition, $(l_u | l_s, l_t, r_t, t) \sim \mathcal{N}\left(\tilde{M}(t,u,s, l_s, H,T), \tilde{S}(t, u, s)\right)$, for $u>s$, which can be obtained as
\begin{align*}
	\tilde{M}(t,u,s, l_s, H,T) &= M(t,u,H,T) + \frac{V(t,u,s)}{S(t,s)}\times \left(l_s - M(t, s, H, T)\right), \\
	\tilde{S}(t, u, s) &= S(t,u)\times\left(1-\frac{V(t,u,s)^2}{S(t,u)S(t,s)}\right),
\end{align*}
where $V(t,u,s)$ is the covariance of the random variables $(l_u | l_t, r_t, t )$ and $(l_s | l_t, r_t, t )$, given as
\begin{align*}
	V(t,u,s) &= \left( \frac{\left(1 + \lambda\phi \right)\sigma_r}{\lambda - \alpha_r}  \right)^2 \frac{\ue^{-\alpha_r\left(u-s\right)}-\ue^{-\alpha_r\left(u+s-2t\right)} }{2\alpha_r}\\
	& +\left( \sigma^2 +  \left( \frac{\left(1 + \lambda\phi \right)\sigma_r}{\lambda - \alpha_r}  \right)^2 - \frac{2\rho\sigma\sigma_r \left(1 + \lambda\phi \right)}{\lambda - \alpha_r}  \right)\frac{\ue^{-\lambda\left(u-s\right)}-\ue^{-\lambda\left(u+s-2t\right)} }{2\lambda} \\
	& + \left(\frac{\rho\sigma\sigma_r \left(1 + \lambda\phi \right)}{\lambda - \alpha_r} - \left( \frac{\left(1 + \lambda\phi \right)\sigma_r}{\lambda - \alpha_r}  \right)^2   \right) \\
	&\times\frac{\ue^{-\alpha_r\left(u-s\right)}-\ue^{-\alpha_r\left(u-t\right)-\lambda\left(s-t\right)} -\ue^{-\alpha_r\left(s-t\right)-\lambda\left(u-t\right)} + \ue^{-\lambda\left(u-s\right)} }{\lambda + \alpha_r}.
\end{align*}

Furthermore, $\E_t^{\p}\left[ P(H, T)\right]$ in equation \eqref{eq:gvth3} is the expected future bond price under the Vasicek model, which can be solved using the results derived in Section \ref{sec:Rmeasure}.\ref{sec:atsm}, and the solution is given in Internet Appendix Section \ref{app:someATSM}.\ref{app:Vas} (see equation \eqref{eq:vasbondprice}).

\end{appendix}


\clearpage

\putbib

%
%

\end{bibunit}

\clearpage

\phantomsection
 \addcontentsline{toc}{section}{Internet Appendix}

 \setcounter{equation}{0}
\setcounter{figure}{0}
\setcounter{table}{0}
\setcounter{lemma}{0}
\setcounter{proposition}{0}
\setcounter{corollary}{0}
\setcounter{footnote}{0}
\setcounter{page}{0}

 \begin{appendices}
 \begin{bibunit}
 \begin{center}
 {\large \bf Internet Appendix for \vspace{-0.05cm}\\
 ``A Theory of Equivalent Expectation Measures for Contingent Claim Returns'' \vspace{-0.05cm}}\\
 {SANJAY K. NAWALKHA and XIAOYANG ZHUO\vspace{-0.05cm}}\footnote{Citation format: Nawalkha, Sanjay K, and Xiaoyang Zhuo, Internet Appendix for ``A Theory of Equivalent Expectation Measures for Contingent Claim Returns,'' \emph{Journal of Finance} [Forthcoming]. Please note: Wiley-Blackwell is not responsible for the content or functionality of any additional information provided by the authors. Any queries (other than missing material) should be directed to the authors of the article.}
 \end{center}

 \begin{abstract}
  \noindent This Internet Appendix provides the following supplemental sections to the main article:
\begin{itemize}
  \item Section \ref{app:DiffEEM}. Additional Theoretical Results Related to Different EEMs
   \item Section \ref{app:someATSM}. Expected Future Bond Prices under Specific ATSMs
  \item Section \ref{sec:qtsm}. Expected Future Bond Price under the Quadratic Term Structure Model
  \item Section \ref{app:fsoexpected}. Expected Future Price of Forward-Start Option using the Forward-Start $R$-Transform
  \item Section \ref{app:MertonVasicek}. Expected Future Option Price under the Merton Model with Vasicek Short Rate Process
  \item Section \ref{iapp:margrabe}. Expected Future Price of an Asset Exchange Option
  \item Section \ref{sec:extractFSPD}. A Procedure to Extract the Expected FSPD
\end{itemize}
 \end{abstract}

 \setcounter{page}{1}

 \setcounter{section}{0}
 \setcounter{subsection}{0}

 \renewcommand{\thesection}{\Roman{section}}
 \renewcommand{\thesubsection}{\Alph{subsection}}

 \setcounter{equation}{0}
 \renewcommand{\theequation}{\thesection.\arabic{equation}}



\renewcommand{\theequation}{IA\arabic{equation}}%
\renewcommand{\thefigure}{IA.\arabic{figure}} \setcounter{figure}{0}
\renewcommand{\thetable}{IA.\Roman{table}} \setcounter{table}{0}

\renewcommand{\thelemma}{IA\arabic{lemma}}
\renewcommand{\thetheorem}{IA\arabic{theorem}}
\renewcommand{\thedefinition}{IA\arabic{definition}}
\renewcommand{\theproposition}{IA\arabic{proposition}}
\renewcommand{\thecorollary}{IA\arabic{corollary}}

\clearpage

\begin{center}
\section{Additional Theoretical Results Related to Different EEMs}\label{app:DiffEEM} \vspace{-1em}
\end{center}

\subsection{Properties of the ${\mathbb{R}}^{*}_1$ Measure}\label{app:propofReeo}
The properties of the $\Reeo$ measure is given by the following proposition:
\begin{proposition}\label{prop:Ree1star}
	For all $0\leq t\leq H\leq s \leq T$, the $\Reeo$ measure has the following properties:
	\begin{enumerate}[(i)]
		\item $\Reeo(A\,|\,\F_{t})=\E_t^\p\left[\frac{\Y_H}{\E_t^\p\left[\Y_H\right]}\I_A\right]$ for all $A\in \mathcal{F}_H$. \label{Ree1:num1}
		\item $\Reeo(A\,|\,\F_{s})=\Ree(A\,|\,\F_{s})=\Qem(A\,|\,\F_{s})$ for all $A\in \mathcal{F}_{T}$. \label{Ree1:num2}
		\item When $H=T$, $\Reeo(A\,|\,\F_{t})=\E_t^\p\left[\frac{\Y_T}{\E_t^\p\left[\Y_T\right]}\I_A\right]$  for all $A\in \mathcal{F}_T$.  \label{Ree1:num3}
		\item When $H=t$, $\Reeo(A\,|\,\F_{t})=\Ree(A\,|\,\F_{t})=\Qem(A\,|\,\F_{t})$ for all $A\in \mathcal{F}_T$. \label{Ree1:num4}
		\item  $\E_t^\Reeo\left[Z_T\right] = \E_t^\Ree\left[\frac{\Y_H}{\E_t^\Ree\left[\Y_H\right]} \cdot Z_T\right]= \E_t^\p\left[\frac{\Y_H}{\E_t^\p\left[\Y_H\right]}\E_H^\Qem\left[Z_T\right]\right]$ for any random variable $Z_T$ at time $T$. \label{Ree1:num5}
	\end{enumerate}
\end{proposition}

\begin{proof}
For the first property, for all $0\leq t\leq H$, given $A\in \F_H$, we have
\begin{align*}
\Reeo(A|\F_t) &=\E_t^\Reeo\left[\I_A \right]= \E_t^\pp\left[\frac{\LHHo}{\LHto}\cdot\I_A   \right]
= \E_t^\pp\left[ \frac{\Y_H/\E_0^\p\left[\Y_H\right]\cdot \LH/\LH}{ {\E_t^\p\left[\Y_H\right]}/{\E_0^\p\left[\Y_H\right]}}\cdot\I_A   \right]\\
&=\E_t^\p\left[\frac{\Y_H}{\E_t^\p\left[\Y_H\right]}\I_A\right].
\end{align*}

For the second property, for all $H\leq s\leq T$, given $A\in \F_T$, we have
\begin{align*}
\begin{aligned}
 \Reeo(A\,|\,\F_s) &= \E^\Reeo_s\left[\I_A\right]=\E_s^\p\left[\frac{\LHTo}{\LHso}\I_A\right]
 =\E_s^\p\left[\frac{\Y_H/\E_0^\p\left[\Y_H\right]\cdot\LT/\LH}{\Y_H/\E_0^\p\left[\Y_H\right]\cdot\Ls/\LH}\I_A\right] \\
 & =\E_s^\p\left[\frac{\LT}{\Ls}\I_A\right] = \E_s^\Qem\left[\I_A\right] =\Qem(A\,|\,\F_s). \\
\end{aligned}
\end{align*}

Furthermore, according to Proposition \ref{prop:Ree1}\eqref{Ree:num2}, $\Reeo(A\,|\,\F_s) = \Ree(A\,|\,\F_s)$ also holds for all $A\in \F_T$.

The third property is a special case of the first property when $H=T$, while the fourth property is a special case of the second property when $H=t$.

The first part of the last property is a direct result of equation \eqref{eq:QIT12def}, which was proven in equation \eqref{eq:RandR1}. For the second part of the last property, we have
\begin{align*}
\E_t^\Reeo\left[Z_T\right] = \E_t^\Reeo\left[\E_H^\Reeo\left[Z_T\right]\right] = \E_t^\Reeo\left[\E_H^\Qem\left[Z_T\right]\right] = \E_t^\p\left[\frac{\Y_H}{\E_t^\p\left[\Y_H\right]}\E_H^\Qem\left[Z_T\right]\right],
\end{align*}
where the second equality is an application of the second property, the last equality is an application of the first property, since $\E_H^\Qem\left[Z_T\right]$ is a random variable that is $\F_H$ measurable.
\end{proof}

Propositions \ref{prop:Ree1star}\eqref{Ree1:num2} and \ref{prop:Ree1star}\eqref{Ree1:num4} for the  $\Reeo$ measure are identical to the corresponding Propositions \ref{prop:Ree1}\eqref{Ree:num2} and \ref{prop:Ree1}\eqref{Ree:num4} for the $\Ree$ measure, respectively. According to Proposition \ref{prop:Ree1star}\eqref{Ree1:num2}, the conditional $\Reeo$ probabilities at any time $s$ on or after time $H$, of the events at time $T$, are the same as the corresponding conditional $\Ree$ probabilities, as well as the conditional $\Qem$ probabilities of those events. According to Proposition \ref{prop:Ree1star}\eqref{Ree1:num4}, when $H = t$, the $\Reeo$ measure becomes the same as the $\Ree$ measure and the $\Qem$ measure, and the expected value computing problem becomes the traditional claim pricing problem. Propositions \ref{prop:Ree1star}\eqref{Ree1:num1} and \ref{prop:Ree1star}\eqref{Ree1:num3} for the $\Reeo$ measure are different from the corresponding Propositions \ref{prop:Ree1}\eqref{Ree:num1} and \ref{prop:Ree1}\eqref{Ree:num3} for the $\Ree$ measure, respectively. According to Proposition \ref{prop:Ree1star}\eqref{Ree1:num1}, the conditional $\Reeo$ probabilities at any time $t$ until time $H$, of any events at time $H$, are not the same as the corresponding $\pp$ probabilities of those events, due to an adjustment term equal to $\Y_H/{\E_t^\p\left[\Y_H\right]}$, related to the numeraire $\Y$. According to Proposition \ref{prop:Ree1star}\eqref{Ree1:num3}, when $H=T$, the $\Reeo$ measure does not become the $\pp$ measure, due to an adjustment term equal to $\Y_T/{\E_t^\p\left[\Y_T\right]}$, related to the numeraire $\Y$. However, there is one exception to the last argument. When the numeraire $\Y$ is given as the $T$-maturity pure discount bond, $P(\cdot,T)$, then $\Y_T = P(T,T) = 1$.  Therefore, for the special case of the numeraire $\Y = P(\cdot,T)$, Proposition \ref{prop:Ree1star}\eqref{Ree1:num3} implies that when $H=T$, the $\Reeo$ measure becomes the $\pp$ measure, since $\Y_T/{\E_t^\p\left[\Y_T\right]} =P(T,T)/{\E_t^\p\left[P(T,T)\right]} = 1$. Proposition \ref{prop:Ree1star}\eqref{Ree1:num5} shows that it is easier to obtain the analytical solution of a contingent claim using the $\Reeo$ measure directly (whenever this measure leads to an analytical solution), since the term $\Y_H/{\E_t^\p\left[\Y_H\right]}$ complicates the solution using the $\Ree$ measure.

\subsection{The ${\mathbb{R}}^{*}_2$ Measure}\label{sec:Reet}
We first give the definition of the $\Reet$ measure in the following lemma:

\begin{lemma}\label{lemma:R2star}
	For a fixed $H$ with $0\leq H\leq T$, define a process $\LHst$ as
	\begin{align*}
		\begin{aligned}
			\LHst &  =\left\{\begin{array}{ll}
				\displaystyle  \frac{\Y_H\E_s^\Ree\left[\Y_s/\Y_T\right]}{\Y_s\E_0^\Ree\left[\Y_H/\Y_T\right]}\cdot \frac{\Ls}{\LH}, &  \text{if} \ \  {H\leq s\leq T}, \\
				\vspace{-0.8em}& \\
				\displaystyle  \frac{\E_s^\Ree\left[\Y_H/\Y_T\right]}{\E_0^\Ree\left[\Y_H/\Y_T\right]}, &    \text{if } \ \ {0\leq s<H}.
			\end{array}\right.
		\end{aligned}
	\end{align*}
	Let
	\begin{align*}%
		\Reet(A) \triangleq \int_A \mathcal{L}_{2T}^*(H; \omega)\D \pp(\omega) \ \ \text{for all $A \in \mathcal{F}_T$},
	\end{align*}
	then $\Reet$ is a probability measure equivalent to $\pp$, and $\LHst$ is the Radon-Nikod\'{y}m derivative process of $\Reet$ with respect to $\pp$.
\end{lemma}

We omit the proof of Lemma \ref{lemma:R2star} because it follows similarly from Appendix \ref{app:thm1} (Section \ref{app:thm1}). For the relation between $\Reet$ and $\Ree$, we have
\begin{align*}
	\begin{aligned}
		\E_t^{\Reet}\left[Z_T\right] &= \E_t^\p\left[\frac{\LHTt
		}{\LHtt}Z_T\right] =\E_t^\p\left[\frac{\LT}{\LH}\cdot\frac{\Y_H/\Y_T}{\E_t^\Ree\left[\Y_H/\Y_T\right]}Z_T\right]\\ &=\E_t^\p\left[\frac{\LHT}{\LHt}\cdot\frac{\Y_H/\Y_T}{\E_t^\Ree\left[\Y_H/\Y_T\right]}Z_T\right]  = \E_t^\Ree\left[\frac{\Y_H/\Y_T}{\E_t^\Ree\left[\Y_H/\Y_T\right]}  Z_T\right],
	\end{aligned}
\end{align*}
for any $\mathcal{F}_T$-measurable variable $Z_T$.

Now, we prove that the $\Reeo$ measure subsumes the $\Reet$ measure when the numeraires are restricted to be either the money market account or the pure discount bond.

According to the classification of EEMs in Section \ref{sec:model}, when $\Y=B=\ue^{\int_0^{\cdot}r_u\D u}$; in other words, the numeraire is the money market account, we have $\Ree = \Ro$, $\Reet = \Rtt$, and $\Ls = \LsB$. Then, for $\Rtt$'s Radon-Nikod\'{y}m derivative process
\begin{align*}
	\begin{aligned}
		\LHstB  =\left\{\begin{array}{ll}
			\displaystyle  \frac{B_H\E_s^\R\left[B_s/B_T\right]}{B_s\E_0^\R\left[B_H/B_T\right]}\cdot \frac{\LsB}{\LHB}=\frac{B_HP(s,T)}{B_s\E_0^\p\left[P(H,T)\right]}\cdot \frac{\LsB}{\LHB}, &  \text{if} \ \  {H\leq s\leq T}, \\
			\vspace{-0.8em}& \\
			\displaystyle  \frac{\E_s^\R\left[B_H/B_T\right]}{\E_0^\R\left[B_H/B_T\right]}=\frac{\E_s^\p\left[P(H,T)\right]}{\E_0^\p\left[P(H,T)\right]}, &    \text{if } \ \ {0\leq s<H},
		\end{array}\right.
	\end{aligned}
\end{align*}
where $\E_s^\R\left[P(H,T)\right]=\E_s^\p\left[P(H,T)\right]$ for $0\leq s\leq H$ due to Proposition \ref{prop:Ree1}.

When $\Y=P(\cdot,T)$, in other words, the numeraire is the pure discount bond, we have $\Ree = \RT$, $\Reet = \RTt$, and $\Ls = \LsP =  \frac{P(s,T)B_0}{P(0,T)B_s}\LsB$. Then, for $\RTt$'s Radon-Nikod\'{y}m derivative process
\begin{align*}
	\begin{aligned}
		\LHstP  =\left\{\begin{array}{ll}
			\displaystyle  \frac{P(H,T)\E_s^\RT\left[P(s,T)\right]}{P(s,T)\E_0^\RT\left[P(H,T)\right]}\cdot \frac{P(s,T)B_H}{P(H,T)B_s}\frac{\LsB}{\LHB}=\frac{B_HP(s,T)}{B_s\E_0^\p\left[P(H,T)\right]}\cdot \frac{\LsB}{\LHB}, &  \text{if} \ \  {H\leq s\leq T}, \\
			\vspace{-0.8em}& \\
			\displaystyle  \frac{\E_s^\RT\left[P(H,T)\right]}{\E_0^\RT\left[P(H,T)\right]}=\frac{\E_s^\p\left[P(H,T)\right]}{\E_0^\p\left[P(H,T)\right]}, &    \text{if } \ \ {0\leq s<H},
		\end{array}\right.
	\end{aligned}
\end{align*}
where $\E_s^\RT\left[P(H,T)\right]=\E_s^\p\left[P(H,T)\right]$ for $0\leq s\leq H$ due to Proposition \ref{prop:Ree1}.

When $\Y=P(\cdot,T)$, we have $\Ree = \RT$, $\Reeo = \RTo$, and $\Ls = \LsP =  \frac{P(s,T)B_0}{P(0,T)B_s}\LsB$. Then, for $\RTo$'s Radon-Nikod\'{y}m derivative process (also see Internet Appendix Section \ref{app:DiffEEM}.\ref{app:Rforward})
\begin{align*}
	\begin{aligned}
		\LHsoP & =\left\{\begin{array}{ll}
			\displaystyle \frac{P(H,T)}{\E_0^\p\left[P(H,T)\right]}\cdot \frac{P(s,T)B_H}{P(H,T)B_s}\frac{\LsB}{\LHB} =\frac{B_HP(s,T)}{B_s\E_0^\p\left[P(H,T)\right]}\cdot \frac{\LsB}{\LHB}, &  \text{if} \ \  {H\leq s\leq T}, \\
			\vspace{-0.8em}& \\
			\displaystyle  \frac{\E_s^\p\left[P(H,T)\right]}{\E_0^\p\left[P(H,T)\right]}, &    \text{if } \ \ {0\leq s<H}.
		\end{array}\right.
	\end{aligned}
\end{align*}

As can be seen, the Radon-Nikod\'{y}m derivative processes satisfy $\LHsoP=\LHstP=\LHstB$, which means $\RTo=\RTt=\Rtt$. Hence, $\RTo$ subsumes both $\RTt$ and $\Rtt$ measures when the numeraires are either the pure discount bond or the money market account.

\subsection{The $\mathbb{R}$ Measure for Multidimensional Brownian Motions}\label{app:multiBM}

Here we consider the $\R$ measure construction when there are multidimensional Brownian motions as the source of uncertainty. Assuming the existence of a risk-neutral measure $\q$, we let the Radon-Nikod\'{y}m derivative process $\LsB$ of $\q$ with respect to $\pp$ be
\begin{align}\label{eq:LsBmul}
	\LsB &\triangleq\frac{\D \q}{\D \pp}\Big|_{\F_s}=\exp\left(-\int_0^s {\gamma}_u\D \Wpu - \frac{1}{2}\int_0^s ||{\gamma}_u||^2\D u\right),
\end{align}
where $\gamma_s=\left(\gamma_{1s}, ..., \gamma_{Ns}\right)$, $0\leq s\leq T$, is an $N$-dimensional measurable, adapted process satisfying $ \int_0^T\gamma_{iu}^2\D u <\infty$ almost surely for $1\leq i\leq N$, and the Novikov condition, that is, $\E\left[\exp\left(\frac{1}{2}\int_0^T ||{\gamma}_u||^2\D u\right)\right]<\infty$. In addition, $||\gamma_s|| = \left(\sum_{i=1}^d \gamma_{is}^2\right)^{\frac{1}{2}}$ denotes the Euclidean norm. %

Then, similarly to the case of one-dimensional Brownian motion in Appendix \ref{prof:BMR}, we have the following lemma:
\begin{lemma}\label{prop:muBMchange}
	In the economy coupled with multidimensional Brownian motion as the source of uncertainty, for a fixed $H$ with $0\leq H\leq T$, define a process $\LHsB$ as
	\begin{align*}
		\LHsB &=\exp\left(-\int_0^s \GuH\D \Wpu - \frac{1}{2}\int_0^s ||\GuH||^2\D u\right),
	\end{align*}
	where
	\begin{align*}
		\GsH &\triangleq \I_{\left\{s\geq H\right\}}\gamma_s=\left\{\begin{array}{ll}
			\displaystyle \gamma_s, &\text{if} \ \  H\leq s \leq T, \\
			\displaystyle 0, &\text{if} \ \  0\leq s < H. \\
		\end{array}\right.
	\end{align*}
	Let
	\begin{align*}
		\Ro(A) \triangleq \int_A L_T(H; \omega)\D \mathbb{P}(\omega) \ \ \text{for all $A \in \mathcal{F}_T$}.
	\end{align*}
	Then $\R$ is a probability measure equivalent to $\pp$, and $\LHsB$ is the Radon-Nikod\'{y}m derivative process of $\R$ with respect to $\pp$.
	
	Moreover, define
	\begin{align}\label{eq:BMchangeRP}
		\Wr &\triangleq \Wp + \int_0^s \GuH \D u  \nonumber \\
		&= \Wp + \int_0^s\I_{\left\{u\geq H\right\}}{\gamma}_u \D u.
	\end{align}
	The process $\Wro$ is an $N$-dimensional Brownian motion under the $\Ro$ measure.
\end{lemma}

\begin{proof}
	We skip the steps that are consistent with the proof of Lemma \ref{lemma:Ree}. For multidimensional Brownian motions, clearly, we have
	\begin{align*}
		&\E\left[\exp\left(\frac{1}{2}\int_0^T ||
		\GuH||^2\D u\right)\right] \\
=&\E\left[\exp\left(\frac{1}{2}\int_H^T ||{{\gamma}}_u||^2\D u\right)\right]\leq\E\left[\exp\left(\frac{1}{2}\int_0^T ||{{\gamma}}_u||^2\D u\right)\right]<\infty.
	\end{align*}
	In other words, $\GsH$ also satisfies the Novikov condition, then the defined $\LHsB$ is a martingale and $\E\left[\LHTB\right]=1$. Since $\LHTB$ is defined as an exponential, it is almost surely strictly positive; hence, the probability measure $\Ro$ is well-defined and equivalent to $\mathbb{P}$.
	
	This Brownian motion transformation between $\Ro$ and $\pp$ in equation \eqref{eq:BMchangeRP} is an application of the multidimensional Girsanov theorem. The Girsanov theorem can apply to any $N$-dimensional process $\GsH$ that has the properties \romannumeral1). $\GsH$ is $\F_s$-adapted for each $s\in [0,T]$; \romannumeral2). $\int_0^T\left({\gamma}_{iu}(H)\right)^2\D u <\infty$ almost surely for $1\leq i\leq N$; \romannumeral3). $\GsH$ satisfies the Novikov condition, which was proven.
	
	For a given $0 \leq H \leq T$, by definition, when $s< H$, $\GsH=0$ is a constant. Since constant functions are measurable with respect to any $\sigma$-algebra, in particular with respect to $\mathcal{F}_s$, $s< H$, this proves that $\GsH$ is $\mathcal{F}_s$-adapted for each $s< H$. When $s \geq H$, $\GsH=\gamma_s$ is also $\mathcal{F}_s$-adapted by the definition of $\gamma_s$. Therefore, $\GsH$ is $\mathcal{F}_s$-adapted for each $s\in [0,T]$, and condition \romannumeral1) is proven. Condition \romannumeral2) is also satisfied due to that, for $1\leq i\leq N$,
	\begin{align*}
		\int_0^T\left({\gamma}_{iu}(H)\right)^2\D u=  \int_0^T\I_{\left\{u\geq H\right\}}{\gamma}_{iu}^2\D u =  \int_H^T {\gamma}_{iu}^2\D u \leq \int_0^T{\gamma}_{iu}^2\D u <\infty \ \ \text{almost surely},
	\end{align*}
	which ends the proof.
\end{proof}

\subsection{The ${\mathbb{R}}^{T}_1$  Measure for Multidimensional Brownian Motions}\label{app:Rforward}

Consider the Radon-Nikod\'{y}m derivative processes $\LHso$ and $\Ls$ in Lemma \ref{lemma:reeo} and denote them as $\LHsoP$ and $\LsP$, respectively, for the special case when the numeraire is the pure discount bond, or $\Y = P(\cdot,T)$. For this special case, $\Reeo = \RTo$, as defined in Corollary \ref{coro:three2}. We know that $\LsP$ is given as
\begin{align*}
	\LsP \triangleq\frac{\D \QT}{\D \pp}\Big|_{\F_s}=\frac{\D \QT}{\D \q}\cdot \frac{\D \q}{\D \pp}\Big|_{\F_s} = \frac{P(s,T)B_0}{P(0,T)B_s}\LsB,
\end{align*}
where $B$ is the money market account, $\LsB$ is defined in equation \eqref{eq:LSB} for one-dimensional Brownian motions, or equation \eqref{eq:LsBmul} for multidimensional Brownian motions.

Assume that the numeraire bond $P(\cdot,T)$'s volatility is $v(\cdot, T)$\footnote{For example, under $\q$, the pure discount bond follows
	\begin{align*}
		\frac{\D P(s,T)}{P(s,T)} = r_s\D s - v(s,T)\D \Wq.
\end{align*}}, then the process of $\LsP$ satisfies
\begin{align*}
	\frac{\D \LsP}{\LsP} = \left( - \gamma_s - v(s, T)\right)\D \Wp.
\end{align*}

Then, we have the following lemma:
\begin{lemma}
	In the economy coupled with stochastic interest rates, for a fixed $H$ with $0\leq H\leq T$, define a process $\LHsoP$ as
	\begin{align*}
		\begin{aligned}
			\LHsoP & =\left\{\begin{array}{ll}
				\displaystyle \frac{P(H,T)}{\E_0^\p\left[P(H,T)\right]}\cdot \frac{P(s,T)B_H}{P(H,T)B_s}\frac{\LsB}{\LHB} =\frac{B_HP(s,T)}{B_s\E_0^\p\left[P(H,T)\right]}\cdot \frac{\LsB}{\LHB}, &  \text{if} \ \  {H\leq s\leq T}, \\
				\vspace{-0.8em}& \\
				\displaystyle  \frac{\E_s^\p\left[P(H,T)\right]}{\E_0^\p\left[P(H,T)\right]}, &    \text{if } \ \ {0\leq s<H}.
			\end{array}\right.
		\end{aligned}
	\end{align*}
	Let
	\begin{align*}%
		\RTo(A) \triangleq \int_A \mathcal{L}_{1T}^P(H; \omega)\D \pp(\omega) \ \ \text{for all $A \in \mathcal{F}_T$},
	\end{align*}
	then $\RTo$ is a probability measure equivalent to $\pp$, and $\LHsoP$ is the Radon-Nikod\'{y}m derivative process of $\RTo$ with respect to $\pp$.
	
	Moreover, assume that
	\begin{align*}
		\frac{\D \LHsoP}{\LHsoP} = \left( - \GsH - v(s, T, H )\right)\D \Wp,
	\end{align*}
	then define
	\begin{align*}%
		\WT &\triangleq \Wp + \int_0^s \GuH \D u  + \int_0^s v(u,T,H)\D u \nonumber \\
		&= \Wp + \int_0^s\I_{\left\{u\geq H\right\}}{\gamma}_u \D u + \int_0^s v(u,T,H)\D u.
	\end{align*}
	The process $\WT$ is a Brownian motion under the $\RTo$ measure.
\end{lemma}
\begin{proof}
	The first part is an application of Lemma \ref{lemma:reeo}, and the second part follows from Proposition 26.4 in \cite{bjork2009arbitrage} with $Q^1 = \RTo$ and $Q^0 = \pp$.
\end{proof}

Notably, when $s\geq H$, $v(s,T,H)$ is equal to $v(s,T)$; when $s<H$, $v(s,T,H)$ is generally a term that is adjusted from $v(s,T)$. For example, under the Vasicek model defined in Internet Appendix Section \ref{app:someATSM}.\ref{app:Vas} (see equation \eqref{eq:vasprocess}), we have
\begin{align*}
	\begin{aligned}
		v(s,T,H) & =\left\{\begin{array}{ll}
			\displaystyle \sigma_r \frac{1-\ue^{-\alpha_r (T-s)}}{\alpha_r}, &  \text{if} \ \  {H\leq s\leq T}, \\
			\vspace{-0.8em}& \\
			\displaystyle  \sigma_r \frac{\ue^{-\alpha_r (H-s)}-\ue^{-\alpha_r (T-s)}}{\alpha_r}, &    \text{if } \ \ {0\leq s<H}.
		\end{array}\right.
	\end{aligned}
\end{align*}

\begin{center}
\section{Expected Future Bond Prices under Specific ATSMs}\label{app:someATSM} \vspace{-1em}
\end{center}
\subsection{Expected Future Bond Price under the Vasicek Model}\label{app:Vas}
For the ease of exposition, this section illustrates the solution of the expected future price of a pure discount bond under the classical one-factor model of \cite{Vasicek1977An}. The Ornstein-Uhlenbeck process for the short rate under the Vasicek model is given as follows:
\begin{align}\label{eq:vasprocess}
	\begin{aligned}
		\D r_s &= \alpha_r \left(m_r - r_s\right)\D s + \sigma_r  \D W_{s}^{\p},
	\end{aligned}
\end{align}
where $\alpha_r$ is the speed of mean reversion, $m_r$ is the long-term mean of the short rate, and $\sigma_r$ is the volatility of the changes in the short rate. Assuming $\Wq = \Wp + \int_0^s{\gamma_r} \D u$, the short rate process under the $\Ro$ measure can be derived using Lemma \ref{lemma:Ree} and Internet Appendix Section \ref{app:DiffEEM}.\ref{app:multiBM}, as follows:
\begin{align*}
	\begin{aligned}
		\D r_s &= [\alpha_r \left(m_r - r_s\right) -\I_{\{s\geq H\}} \sigma_r \gamma_r]\D s + \sigma_r \D \Wr \\
		& = \left[\alpha_r \left(m_r - r_s\right)\I_{\{s< H\}}  +  \alpha_r^* \left(m_r^* - r_s\right)\I_{\{s\geq H\}}  \right]\D s + \sigma_r  \D W_{s}^{\R}, %
	\end{aligned}
\end{align*}
where $ \alpha_r^* = \alpha_r $, $m_r^* = m_r - \sigma_r\gamma_r/\alpha_r$.

The expected future price of a $T$-maturity pure discount bond at time $H$ is given under the Vasicek model as follows: For all $t\leq H\leq T$,
\begin{align}\label{eq:vasbondprice}
	\begin{aligned}
		\E_t \left[P(H,T)\right] &= \E_t^\R\left[\exp\left(-\int_{H}^{T}r_u\D u\right)\right]  \\
		&=\exp\left(- A^{(b^*, c^*)}_{\alpha_r^*, m_r^*}(T-H)- A^{(b, {c})}_{\alpha_r, m_r}(H-t)  -B^{(b, c)}_{\alpha_r, m_r}(H-t)\cdot{r_t} \right),
	\end{aligned}
\end{align}
where $b^* = 0$, $c^* = 1$, $b = B^{(b^*,c^*)}_{\alpha_r^*, m_r^*}(T-H)$, $c=0$.
For any well-behaved $b$, $c$, $\alpha$, and $m$, $A^{(b,c)}_{\alpha, m}\left(\tau\right)$  and $B^{(b,c)}_{\alpha, m}\left(\tau\right)$ are given by
\begin{align*}
	\begin{aligned}
		A^{(b,c)}_{\alpha, m}\left(\tau\right)& =  \alpha m\left(bB_{\alpha}{(\tau)}  + c\frac{\tau - B_{\alpha}{(\tau)} }{\alpha}\right) \\ & -\frac{1}{2}\sigma_r^2\left( b^2\frac{1-\ue^{-2\alpha \tau}}{2\alpha} + bc {B_{\alpha}{(\tau)}}^2 + c^2\frac{2\tau-2B_{\alpha}{(\tau)} - \alpha {B_{\alpha}{(\tau)}}^2}{{2\alpha}^2} \right),\\
		B^{(b,c)}_{\alpha, m}\left(\tau\right) &= b\ue^{-\alpha \tau} + c B_{\alpha}{(\tau)},
	\end{aligned}
\end{align*}
with $B_{\alpha} (\tau)=\left(1-\ue^{-\alpha \tau}\right)/\alpha$.
When $H=t$, equation \eqref{eq:vasbondprice} reduces to the Vasicek bond price solution \cite[see][equation 27]{Vasicek1977An}.

\subsection{Expected Future Bond Price under the CIR Model} \label{app:CIR}
We now illustrate the expected future price of a pure discount bond under the classic one-factor Cox-Ingersoll-Ross \cite[CIR,][]{Cox1985A} model. The square-root process for the short rate under the CIR model is given as follows:
\begin{align*}
	\begin{aligned}
		\D r_s &= \alpha_r \left(m_r - r_s\right)\D s + \sigma_r \sqrt{r_s} \D W_{s}^{\p}.
	\end{aligned}
\end{align*}
Assume that the market price of risk $\gamma_{s}$ satisfies
\begin{align*}
	\begin{aligned}
		\gamma_{s} &= \gamma_r \sqrt{r_s}.
	\end{aligned}
\end{align*}
Thus, $\Wq = \Wp + \int_0^s{\gamma_r \sqrt{r_u}} \D u$, and the short rate process under the $\Ro$ measure can be derived using Lemma \ref{lemma:Ree} and Internet Appendix Section \ref{app:DiffEEM}.\ref{app:multiBM}, as follows:
\begin{align*}
	\D r_s &= \left[\alpha_r \left(m_r - r_s\right) - \I_{\left\{s\geq H\right\}}\gamma_r\sigma_r r_s\right]\D s + \sigma_r \sqrt{r_s} \D \Wr \\
	& = \left[\alpha_r \left(m_r - r_s\right)\I_{\{s< H\}} +  \alpha_r^* \left(m_r^* - r_s\right)\I_{\{s\geq H\}} \right]\D s + \sigma_r \sqrt{r_s} \D W_{s}^{\R},
\end{align*}
where $ \alpha_r^* = \alpha_r + \gamma_r\sigma_r$, $m_r^* = \alpha_r m_r/\left(\alpha_r + \gamma_r\sigma_r\right)$.

The expected future price of a $T$-maturity pure discount bond at time $H$ is given under the CIR model as follows: For all $t\leq H\leq T$,
\begin{align}\label{eq:CIRbondprice}
	\E_t\left[ P(H, T)\right] &=\E_t^{\R}\left[\exp\left(- \int_H^T r_u\D u \right) \right] \nonumber \\
	&=\exp\left(- A^{(b^*, c^*)}_{\alpha_r^*, m_r^*}(T-H)- A^{(b, {c})}_{\alpha_r, m_r}(H-t)  -B^{(b, c)}_{\alpha_r, m_r}(H-t)\cdot{r_t} \right),
\end{align}
where $b^* = 0$, $c^* = 1$, $b = B^{(b^*,c^*)}_{\alpha_r^*, m_r^*}(T-H)$, $c=0$. For any well-behaved $b$, $c$, $\alpha$, and $m$, $A^{(b,c)}_{\alpha, m}\left(\tau\right)$  and $B^{(b,c)}_{\alpha, m}\left(\tau\right)$ are given as
\begin{align*}
	\begin{aligned}
		A^{(b,c)}_{\alpha, m}\left(\tau\right)& = -\frac{2\alpha m}{\sigma_r^2}\ln\left(\frac{2\beta\ue^{\frac{1}{2}\left(\beta+\alpha\right)\tau}}
		{b\sigma_r^2\left(\ue^{\beta \tau}-1\right)+\beta-\alpha +\ue^{\beta \tau}\left(\beta+\alpha\right)}\right),\\
		B^{(b,c)}_{\alpha, m}\left(\tau\right) &= \frac{b\left(\beta+\alpha +\ue^{\beta \tau}\left(\beta-\alpha\right)\right)+2c\left(\ue^{\beta \tau}-1\right)}{b\sigma_r^2\left(\ue^{\beta \tau}-1\right)+\beta-\alpha +\ue^{\beta \tau}\left(\beta+\alpha\right)},
	\end{aligned}
\end{align*}
with $ \beta = \sqrt{{\alpha}^2+2\sigma_r^2c}$.
When $H=t$, equation \eqref{eq:CIRbondprice} reduces to the CIR bond price solution \cite[see][equation 23]{Cox1985A}.

\subsection{Expected Future Bond Price under the ${A_{1r}(3)}$ Model}
\label{sec:A1rmodel}

Here, we employ the commonly used $A_{1r}(3)$ form of three-factor ATSMs, which has been developed in various models, including \cite{balduzzi1996simple}, \cite{dai2000specification}. Specifically, we use the $A_{1r}(3)_{\text{MAX}}$ model to compute the expected bond price. Under the physical measure $\p$, the state variable processes under the $A_{1r}(3)_{\text{MAX}}$ model are defined as follows:
\begin{align*}
	\begin{aligned}
		\D v_s &= {\alpha}_v\left({m}_v - v_s\right) \D s + \eta \sqrt{v_s} \D {W}_{1s}^{\p}, \\
		\D \theta_s &={\alpha}_{\theta}\left( {m}_{\theta} - \theta_s \right)\D s + \sigma_{\theta v}\eta \sqrt{v_s} \D {W}_{1s}^{\p} + \sqrt{\zeta^2 + \beta_{\theta}v_s}\D {W}_{2s}^{\p} + \sigma_{\theta r} \sqrt{\delta_r + v_s} \D {W}_{3s}^{\p}, \\
		\D r_s &= {\alpha}_{r v }\left( {m}_{v} -v_s \right)\D s +{\alpha}_{r}\left( \theta_s -r_s \right)\D s +\sigma_{r v}\eta \sqrt{v_s} \D {W}_{1s}^{\p} + \sigma_{r \theta}\sqrt{\zeta^2 + \beta_{\theta}v_s}\D {W}_{2s}^{\p} \\
 & + \sqrt{\delta_r + v_s} \D {W}_{3s}^{\p},
	\end{aligned}
\end{align*}
where $W_{1}^{\p}$, $W_{2}^{\p}$ and $W_{3}^{\p}$ are independent Brownian motions under the physical measure $\p$. We assume the market price of risks (MPRs), $\gamma_s= \left(\gamma_{1s}, \gamma_{2s}, \gamma_{3s}\right)$, are given by
\begin{align*}
	\begin{aligned}
		\gamma_{1s} &= \gamma_1\sqrt{v_s}, \\
		\gamma_{2s} &= \gamma_2\sqrt{\zeta^2 + \beta_{\theta}v_s},  \\
		\gamma_{3s} &= \gamma_3\sqrt{\delta_r + v_s}.
	\end{aligned}
\end{align*}

Therefore, using Lemma \ref{lemma:Ree} and Internet Appendix Section \ref{app:DiffEEM}.\ref{app:multiBM}, the state variable processes under the $A_{1r}(3)_{\text{MAX}}$ model under the $\RR$ measure are given by
\begin{align*}
	\begin{aligned}
		\D v_s &= \left[{\alpha}_v\left({m}_v - v_s\right)-\I_{\left\{s\geq H\right\}}\gamma_1\eta v_s \right]\D s + \eta \sqrt{v_s} \D {W}_{1s}^{\h}, \\
		\D \theta_s &=\left[{\alpha}_{\theta}\left( {m}_{\theta} - \theta_s \right)- \I_{\left\{s\geq H\right\}}\gamma_1\sigma_{\theta v}\eta v_s - \I_{\left\{s\geq H\right\}}\gamma_2\left(\zeta^2 + \beta_{\theta}v_s\right) \right.\\
		& \bigg.-\I_{\left\{s\geq H\right\}}\gamma_3\sigma_{\theta r}\left(\delta_r + v_s\right)  \bigg]\D s + \sigma_{\theta v}\eta \sqrt{v_s} \D {W}_{1s}^{\h} + \sqrt{\zeta^2 + \beta_{\theta}v_s}\D {W}_{2s}^{\h} + \sigma_{\theta r} \sqrt{\delta_r + v_s} \D {W}_{3s}^{\h}, \\
		\D r_s &= \left[{\alpha}_{r v }\left( {m}_{v} -v_s \right)+{\alpha}_{r}\left( \theta_s -r_s \right)- \I_{\left\{s\geq H\right\}}\gamma_1\sigma_{r v}\eta v_s -\I_{\left\{s\geq H\right\}}\gamma_2\sigma_{r\theta}\left(\zeta^2 + \beta_{\theta}v_s\right) \right. \\
		&\bigg.-\I_{\left\{s\geq H\right\}}\gamma_3\left(\delta_r + v_s\right)  \bigg]\D s +\sigma_{r v}\eta \sqrt{v_s} \D {W}_{1s}^{\h} + \sigma_{r \theta}\sqrt{\zeta^2 + \beta_{\theta}v_s}\D {W}_{2s}^{\h} + \sqrt{\delta_r + v_s} \D {W}_{3s}^{\h}.
	\end{aligned}
\end{align*}

Hence, using equation \eqref{eq:rep1} in Corollary \ref{coro:three}, the time $t$ expectation of the future price of a $T$-maturity default-free zero-coupon bond at time $H$ under the $A_{1r}(3)_{\text{MAX}}$ model is given  as follows. For all  $t\leq H\leq T$,
\begin{align*}
	\E_t\left[ P(H, T)\right]&=\E_t^{\R}\left[\exp\left(- \int_H^T r_u\D u \right) \right]  \nonumber \\
	&=\exp\left( -A^{\!\left( \!\begin{smallmatrix} \lambda_1^*,\mu_1^*\\ \lambda_2^*,\mu_2^* \\ \lambda_3^*,\mu_3^*\end{smallmatrix} \!\right)}_{\chi^*}\!(T-H)-A_{\chi}^{\!\left( \!\begin{smallmatrix} \lambda_1,\mu_1\\ \lambda_2,\mu_2 \\ \lambda_3,\mu_3\end{smallmatrix} \!\right)}\!(H-t) -B_{\chi}^{\!\left( \!\begin{smallmatrix} \lambda_1,\mu_1\\ \lambda_2,\mu_2 \\ \lambda_3,\mu_3\end{smallmatrix} \!
		\!\right)}\!(H-t)v_t \right. \nonumber \\
	&\hspace{4em} \left. -C_{\chi}^{\!\left( \!\begin{smallmatrix} \lambda_1,\mu_1\\ \lambda_2,\mu_2 \\ \lambda_3,\mu_3\end{smallmatrix} \! \!\right)}\!(H-t)\theta_t -D_{\chi}^{\!\left( \!\begin{smallmatrix} \lambda_1,\mu_1\\ \lambda_2,\mu_2 \\ \lambda_3,\mu_3\end{smallmatrix} \! \!\right)}\!(H-t)r_t\right),
\end{align*}
where
\begin{align*}
	\lambda_1^* &  = 0, \ \ \lambda_2^* = 0, \ \ \lambda_3^* = 0, \ \ \mu_1^* = 0, \ \ \mu_2^* = 0, \ \ \mu_3^* = 1, \ \ 	 \mu_1 = 0, \ \ \mu_2= 0, \ \ \mu_3 = 0, \\
	\lambda_1 & = B^{\!\left( \!\begin{smallmatrix} \lambda_1^*,\mu_1^*\\ \lambda_2^*,\mu_2^* \\ \lambda_3^*,\mu_3^*\end{smallmatrix} \!\right)}_{\chi^*}\!(T-H), \ \ \lambda_2 = C^{\!\left( \!\begin{smallmatrix} \lambda_1^*,\mu_1^*\\ \lambda_2^*,\mu_2^* \\ \lambda_3^*,\mu_3^*\end{smallmatrix} \!\right)}_{\chi^*}\!(T-H), \ \ \lambda_3= D^{\!\left( \!\begin{smallmatrix} \lambda_1^*,\mu_1^*\\ \lambda_2^*,\mu_2^* \\ \lambda_3^*,\mu_3^*\end{smallmatrix} \!\right)}_{\chi^*}\!(T-H),
\end{align*}
and
\begin{align*}
	{\chi}&=\left(\alpha_v, \alpha_\theta, \alpha_r, \alpha_{rv},  m_v, m_\theta, \sigma_{\theta v}, \sigma_{\theta r}, \sigma_{rv}, \sigma_{r\theta}, \eta, \zeta, \beta_\theta, \delta_r, 0, 0 \right), \\
	{\chi}^*&=\left(\alpha_v^*, \alpha_\theta^*, \alpha_r^*, \alpha_{rv}^*, m_v^*, m_\theta^*, \sigma_{\theta v}, \sigma_{\theta r}, \sigma_{rv}, \sigma_{r\theta}, \eta, \zeta, \beta_\theta, \delta_r,   \alpha_{\theta v}^*, m_r^* \right),
\end{align*}
with
\begin{align*}
	\alpha_v^* &= \alpha_v + \gamma_1\eta, \ \ m_v^* = \frac{\alpha_v m_v}{\alpha_v + \gamma_1\eta}, \ \ \alpha_\theta^* = \alpha_\theta, \ \ m_\theta^* = \frac{\alpha_\theta m_\theta - \gamma_2\zeta^2 -\gamma_3\sigma_{\theta r} \delta_r}{\alpha_\theta}, \\
	\alpha_r^* &= \alpha_r, \ \ \alpha_{rv}^* = \alpha_{rv} + \gamma_1\sigma_{r v}\eta + \gamma_2\sigma_{r\theta}\beta_{\theta} + \gamma_3, \ \
	\alpha_{\theta v}^* = - \gamma_1\sigma_{\theta v}\eta - \gamma_2\beta_{\theta} - \gamma_3\sigma_{\theta r}, \\
	m_r^* &=  \alpha_{rv} m_v - \alpha_{rv}^* m_v^* -  \gamma_2\sigma_{r\theta}\zeta^2 - \gamma_3\delta_r.
\end{align*}
For any well-defined $\lambda_i$, $\mu_i$ ($i=1,2,3$) and the characteristic $\chi$, $A_{\chi}(\tau) \triangleq A_{\chi}^{\!\!\left( \!\!\begin{smallmatrix} \lambda_1,\mu_1\\ \lambda_2,\mu_2 \\ \lambda_3,\mu_3\end{smallmatrix} \!
	\!\right)}\!\!(\tau)$, $B_{\chi}(\tau) \triangleq B_{\chi}^{\!\!\left( \!\!\begin{smallmatrix} \lambda_1,\mu_1\\ \lambda_2,\mu_2 \\ \lambda_3,\mu_3\end{smallmatrix} \!
	\!\right)}\!\!(\tau)$, $C_{\chi}(\tau)\triangleq C_{\chi}^{\!\!\left( \!\!\begin{smallmatrix} \lambda_1,\mu_1\\ \lambda_2,\mu_2 \\ \lambda_3,\mu_3\end{smallmatrix} \!
	\!\right)}\!\!(\tau)$, $D_{\chi}(\tau)\triangleq D_{\chi}^{\!\!\left( \!\!\begin{smallmatrix} \lambda_1,\mu_1\\ \lambda_2,\mu_2 \\ \lambda_3,\mu_3\end{smallmatrix} \!
	\!\right)}\!\!(\tau)$ satisfy the following ODEs
\begin{align*}
	\begin{aligned}
		\frac{\partial A_{\chi}(\tau)}{\partial \tau} =& \ \  B_{\chi}(\tau){\alpha}_v {m}_v + C_{\chi}(\tau){\alpha}_{\theta}{m}_{\theta}+D_{\chi}(\tau)\left({\alpha}_{r v}{m}_v + m_r\right)- \frac{1}{2}C_{\chi}^2(\tau)\left(\sigma_{\theta r}^2\delta_r + \zeta^2\right) \\
&-\frac{1}{2}D_{\chi}^2(\tau)\left(\sigma_{r\theta}^2\zeta^2 + \delta_r \right) - C_{\chi}(\tau)D_{\chi}(\tau)\left(\sigma_{r\theta}\zeta^2 + \sigma_{\theta r}\delta_r \right),  \\
\frac{\partial B_{\chi}(\tau)}{\partial \tau}=&-B_{\chi}(\tau){\alpha}_v +C_{\chi}(\tau)\alpha_{\theta v}-D_{\chi}(\tau){\alpha}_{r v} -\frac{1}{2} B_{\chi}^2(\tau) \eta ^2 - \frac{1}{2} C_{\chi}^2(\tau) \left( \beta_{\theta}  + \sigma_{\theta v}^2\eta^2 + \sigma_{\theta r}^2 \right)
	\end{aligned}
\end{align*}
\begin{align}\label{eq:a1r3}
	\begin{aligned}
		&- \frac{1}{2} D_{\chi}^2(\tau)\left( \sigma_{r v}^2\eta^2 + \sigma_{r \theta}^2 \beta_{\theta} +1 \right) -B_{\chi}(\tau)C_{\chi}(\tau)\sigma_{\theta v}\eta^2 -  B_{\chi}(\tau)D_{\chi}(\tau) \sigma_{r v}\eta^2 \\
		&-C_{\chi}(\tau) D_{\chi}(\tau)\left(\sigma_{\theta v}\sigma_{r v}\eta^2+\sigma_{r \theta}\beta_{\theta}+ \sigma_{\theta r}\right) + \mu_1, \\
		\frac{\partial C_{\chi}(\tau)}{\partial \tau}=& -C_{\chi}(\tau) {\alpha}_{\theta} +  D_{\chi}(\tau) {\alpha}_r  + \mu_2, \\
		\frac{\partial D_{\chi}(\tau)}{\partial \tau}=& -D_{\chi}(\tau)   {\alpha}_r  + \mu_3,
	\end{aligned}
\end{align}
subject to the boundary condition $A_{\chi}=0$, $B_{\chi}(0)=\lambda_1$, $C_{\chi}(0)=\lambda_2$, $D_{\chi}(0)=\lambda_3$. $A_{\chi}(\tau)$ and $B_{\chi}(\tau)$ can be solved easily using numerical integration, and $C_{\chi}(\tau)$ and $D_{\chi}(\tau)$ are given as
\begin{align*}
	C_{\chi}(\tau)=&\lambda_2\ue^{-{\alpha}_{\theta}\tau} + \frac{ \mu_2+\mu_3 }{{\alpha}_{\theta} }\left(1-\ue^{{-\alpha}_{\theta} \tau } \right) +\frac{\lambda_3 {\alpha}_r - \mu_3}{{\alpha}_{\theta} - {\alpha}_r }\left(\ue^{-{\alpha}_r \tau}-\ue^{-{\alpha}_{\theta} \tau}\right),    \\
	D_{\chi}(\tau)=& \lambda_3 \ue^{-{\alpha}_r \tau} + \frac{\mu_3 }{{\alpha}_r }\left(1-\ue^{-{\alpha}_r \tau}\right).
\end{align*}
Similarly,  $A_{\chi^*}(\tau) \triangleq A_{\chi^*}^{\!\!\left( \!\!\begin{smallmatrix} \lambda_1^*,\mu_1^*\\ \lambda_2^*,\mu_2^* \\ \lambda_3^*,\mu_3^*\end{smallmatrix} \!
	\!\right)}\!\!(\tau)$, $B_{\chi^*}(\tau) \triangleq B_{\chi^*}^{\!\!\left( \!\!\begin{smallmatrix} \lambda_1^*,\mu_1^*\\ \lambda_2^*,\mu_2^* \\ \lambda_3^*,\mu_3^*\end{smallmatrix} \!
	\!\right)}\!\!(\tau)$, $C_{\chi^*}(\tau)\triangleq C_{\chi^*}^{\!\!\left( \!\!\begin{smallmatrix} \lambda_1^*,\mu_1^*\\ \lambda_2^*,\mu_2^* \\ \lambda_3^*,\mu_3^*\end{smallmatrix} \!
	\!\right)}\!\!(\tau)$, $D_{\chi^*}(\tau)\triangleq D_{\chi^*}^{\!\!\left( \!\!\begin{smallmatrix} \lambda_1^*,\mu_1^*\\ \lambda_2^*,\mu_2^* \\ \lambda_3^*,\mu_3^*\end{smallmatrix} \!
	\!\right)}\!\!(\tau)$ are obtained as in equation \eqref{eq:a1r3} with the characteristic $\chi$ replaced with the characteristic $\chi^*$, and the parameters $\lambda_i$, $\mu_i$ ($i=1,2,3$) replaced with  $\lambda_i^*$, $\mu_i^*$ ($i=1,2,3$).

\begin{center}
\section{Expected Future Bond Price under the Quadratic Term Structure Model}\label{sec:qtsm}\vspace{-1em}
\end{center}

As in \cite{Ahn2002Quadratic}, an $N$-factor quadratic term structure model (QTSM) satisfies the following assumptions. First, the nominal instantaneous interest rate is a quadratic function of the state variables
\begin{align*}
	r_s = \alpha + \mymat{\beta}'\mymat{Y_s} + \mymat{Y_s}'\mymat{\Psi} \mymat{Y_s},
\end{align*}
where $\alpha$ is a constant, $\mymat{\beta}$ is an $N$-dimensional vector, and $\mymat{\Psi}$ is an $N\times N$ matrix of constants. We assume that $\alpha - \frac{1}{4}\mymat{\beta}'\mymat{\Psi}^{-1}\mymat{\beta}\geq 0$, and $\mymat{\Psi}$ is a positive semidefinite matrix.

The SDEs of the state variables $\mymat{Y}$ are characterized as multivariate Gaussian processes with mean reverting properties:
\begin{align*}
	\D \mymat{Y_s} = \left(\mymat{\mu} + \mymat{\xi}\mymat{ Y_s}\right)\D s + \mymat{\Sigma}\D \mymat{\Wp},
\end{align*}
where $\mymat{\mu}$ is an $N$-dimensional vector of constants and $\mymat{\xi}$ and $\mymat{\Sigma}$ are $N$-dimensional square matrices. We assume that $\mymat{\xi}$ is ``diagonalizable'' and has negative real components of eigenvalues. $\mymat{W}^\p$ is an $N$-dimensional vector of standard Brownian motions that are mutually independent.

Similar to \cite{Ahn2002Quadratic}, the market price of risks $\mymat{\gamma_s}$ is assumed to be
\begin{align*}
	\mymat{\gamma_s} =\inv{\mymat{\Sigma}}\left(\mymat{\gamma_0} + \mymat{\gamma_1}\mymat{Y_s} \right),
\end{align*}
where $\mymat{\gamma_0}$ is an $N$-dimensional vector, $\mymat{\gamma_1}$ is an $N \times N$ matrix satisfying that $\inv{\mymat{\Sigma}}'\inv{\mymat{\Sigma}}\mymat{\gamma_1}$ is symmetric. Hence, using Lemma \ref{lemma:Ree} and Internet Appendix Section \ref{app:DiffEEM}.\ref{app:multiBM}, the state variables $\mymat{Y}$ under the $\RR$ measure satisfy the following SDE:
\begin{align*}
	\D \mymat{Y_s} &= \left[\mymat{\mu} + \mymat{\xi}\mymat{ Y_s}-\I_{\left\{s\geq H\right\}}\left(\mymat{\gamma_0} + \mymat{\gamma_1}\mymat{Y_s} \right)\right]\D s + \mymat{\Sigma}\D \mymat{\Wr} \\
	& =  \left[ \left(\mymat{\mu} + \mymat{\xi}\mymat{ Y_s}\right)\I_{\left\{s<H\right\}} + \left(\mymat{\mu^*} + \mymat{\xi^*}\mymat{ Y_s}\right)\I_{\left\{s\geq H\right\}}  \right]\D s + \mymat{\Sigma}\D \mymat{\Wr},
\end{align*}
where $\mymat{\mu^*} = \mymat{\mu} - \mymat{\gamma_0}, \mymat{\xi^*} = \mymat{\xi} - \mymat{\gamma_1}$.

Therefore, the time $t$ expectation of the future price of a $T$-maturity default-free zero-coupon bond at time $H$ under the $N$-factor QTSM model is given as follows. For all $t\leq H\leq T$,
\begin{align*}
	\E_t\left[ P(H, T)\right] &=\E_t^{\R}\left[\exp\left(- \int_H^T r_u\D u \right) \right] = \E_t^{\R}\left[\exp\left(  - \int_H^T \left( \alpha + \mymat{\beta}'\mymat{Y_u} + \mymat{Y_u}'\mymat{\Psi} \mymat{Y_u}\right)\D u \right) \right] \nonumber \\
	&=\exp\left(-\alpha(T-H) - A^{\!\left( \!\begin{smallmatrix}{\mymat{b^*}}, \mymat{d^*}\\ {\mymat{c^*}}, \mymat{q^*} \end{smallmatrix} \!\right)}_{\mymat{\mu^*}, \mymat{\xi^*}}\!(T-H) -A^{\!\left( \!\begin{smallmatrix}{\mymat{b}}, \mymat{d}\\ {\mymat{c}}, \mymat{q} \end{smallmatrix} \!\right)}_{\mymat{\mu}, \mymat{\xi}}\!(H-t)-\mymat{B^{{\!\left( \!\begin{smallmatrix}{\mymat{b}}, \mymat{d}\\ {\mymat{c}}, \mymat{q} \end{smallmatrix} \!\right)}\!}_{\mymat{\mu}, \mymat{\xi}}(H-t)'\mymat{Y_t}} \right.  \nonumber \\
	&\hspace{4em}  \left. -  \mymat{\mymat{Y_t}' C^{{\!\left( \!\begin{smallmatrix}{\mymat{b}}, \mymat{d}\\ {\mymat{c}}, \mymat{q} \end{smallmatrix} \!\right)}\!}_{\mymat{\mu}, \mymat{\xi}}(H-t)\mymat{Y_t}}\right),
\end{align*}
where
\begin{align*}
	\mymat{b^*} &= \mymat{0}, \ \   \mymat{c^*} = \mymat{0}, \ \ \mymat{d^*} = \mymat{\beta}, \ \ \mymat{q^*} = \mymat{\Psi}, \\ \mymat{b}&=\mymat{B^{\!\left( \!\begin{smallmatrix}{\mymat{b^*}}, \mymat{d^*}\\ {\mymat{c^*}}, \mymat{q^*} \end{smallmatrix} \!\right)}_{\mymat{\mu^*}, \mymat{\xi^*}}\!(T-H)}, \ \ \mymat{c}=\mymat{C^{\!\left( \!\begin{smallmatrix}{\mymat{b^*}}, \mymat{d^*}\\ {\mymat{c^*}}, \mymat{q^*} \end{smallmatrix} \!\right)}_{\mymat{\mu^*}, \mymat{\xi^*}}\!(T-H)}, \ \ \mymat{d} = \mymat{0}, \ \ \mymat{q} = \mymat{0}.
\end{align*}

For any well-defined matrices  $\mymat{b}, \mymat{c}, \mymat{d}, \mymat{q}$, $\mymat{\mu}$ and $\mymat{\xi}$, $A_{\mymat{\mu}, \mymat{\xi}}(\tau)\triangleq A^{\!\left( \!\begin{smallmatrix}{\mymat{b}}, \mymat{d}\\ {\mymat{c}}, \mymat{q} \end{smallmatrix} \!\right)}_{\mymat{\mu}, \mymat{\xi}}\!(\tau) $, $\mymat{B_{\mymat{\mu}, \mymat{\xi}}(\tau)}\triangleq \mymat{B^{\!\left( \!\begin{smallmatrix}{\mymat{b}}, \mymat{d}\\ {\mymat{c}}, \mymat{q} \end{smallmatrix} \!\right)}_{\mymat{\mu}, \mymat{\xi}}\!(\tau)} $, and $\mymat{C_{\mymat{\mu}, \mymat{\xi}}(\tau)}\triangleq \mymat{C_{\mymat{\mu}, \mymat{\xi}}^{\!\left( \!\begin{smallmatrix}{\mymat{b}}, \mymat{d}\\ {\mymat{c}}, \mymat{q} \end{smallmatrix} \!\right)}\!(\tau)} $ satisfy the following ODEs
\begin{align}\label{eq:qtsmODE}
	\begin{aligned}
		\frac{\D A_{\mymat{\mu}, \mymat{\xi}}(\tau)}{\D \tau} &= \mymat{\mu}'\mymat{B_{\mymat{\mu}, \mymat{\xi}}(\tau)} - \frac{1}{2}\mymat{B_{\mymat{\mu}, \mymat{\xi}}(\tau)}' \mymat{\Sigma} \mymat{\Sigma}'\mymat{B_{\mymat{\mu}, \mymat{\xi}}(\tau)} + \text{tr}\left(  \mymat{\Sigma} \mymat{\Sigma}'\mymat{C_{\mymat{\mu}, \mymat{\xi}}(\tau)}\right),  \\
		\frac{\D \mymat{B_{\mymat{\mu}, \mymat{\xi}}(\tau)}}{\D \tau} &=  \mymat{\xi}'\mymat{B_{\mymat{\mu}, \mymat{\xi}}(\tau)} + 2\mymat{C_{\mymat{\mu}, \mymat{\xi}}(\tau)}\mymat{\mu}-2\mymat{C_{\mymat{\mu}, \mymat{\xi}}(\tau)}\mymat{\Sigma} \mymat{\Sigma}'\mymat{B_{\mymat{\mu}, \mymat{\xi}}(\tau)} + \mymat{d}, \\
		\frac{\D \mymat{C_{\mymat{\mu}, \mymat{\xi}}(\tau)}}{\D \tau} &= -2\mymat{C_{\mymat{\mu}, \mymat{\xi}}(\tau)}\mymat{\Sigma} \mymat{\Sigma}'\mymat{C_{\mymat{\mu}, \mymat{\xi}}(\tau)} + \mymat{C_{\mymat{\mu}, \mymat{\xi}}(\tau)}\mymat{\xi} + \mymat{\xi}'\mymat{C_{\mymat{\mu}, \mymat{\xi}}(\tau)}
		+ \mymat{q},
	\end{aligned}
\end{align}
with the terminal conditions $A_{\mymat{\mu}, \mymat{\xi}}(0) = 0$, $\mymat{B_{\mymat{\mu}, \mymat{\xi}}(0)} =\mymat{b}$, $\mymat{C_{\mymat{\mu}, \mymat{\xi}}(0)} =\mymat{c}$. These ODEs can be easily solved numerically. $A_{\mymat{\mu^*}, \mymat{\xi^*}}(\tau)\triangleq A^{\!\left( \!\begin{smallmatrix}{\mymat{b^*}}, \mymat{d^*}\\ {\mymat{c^*}}, \mymat{q^*} \end{smallmatrix} \!\right)}_{\mymat{\mu^*}, \mymat{\xi^*}}\!(\tau) $, $\mymat{B_{\mymat{\mu^*}, \mymat{\xi^*}}(\tau)}\triangleq \mymat{B^{\!\left( \!\begin{smallmatrix}{\mymat{b^*}}, \mymat{d^*}\\ {\mymat{c^*}}, \mymat{q^*} \end{smallmatrix} \!\right)}_{\mymat{\mu}, \mymat{\xi}}\!(\tau)} $, and $\mymat{C_{\mymat{\mu^*}, \mymat{\xi^*}}(\tau)}\triangleq \mymat{C_{\mymat{\mu^*}, \mymat{\xi^*}}^{\!\left( \!\begin{smallmatrix}{\mymat{b^*}}, \mymat{d^*}\\ {\mymat{c^*}}, \mymat{q^*} \end{smallmatrix} \!\right)}\!(\tau)} $ are obtained similarly as in equation \eqref{eq:qtsmODE} with the drift coefficients $({\mymat{\mu}, \mymat{\xi}})$ replaced with drift coefficient $({\mymat{\mu^*}, \mymat{\xi^*}})$, and the parameters $(\mymat{b}, \mymat{c}, \mymat{d}, \mymat{q})$ replaced with  $(\mymat{b^*}, \mymat{c^*}, \mymat{d^*}, \mymat{q^*})$.

\textit{Example: The QTSM3 model.}
Since the empirical literature suggests that three factors are required to describe the term structure \citep[see, e.g.,][]{Knez1996Explorations, dai2000specification}, we fix $N=3$ and explore the particular QTSM3 model defined in \cite{Ahn2002Quadratic}. According to the parameter restrictions by \cite{Ahn2002Quadratic}, in this case, $\mymat{\mu}\geq 0$, $\alpha >0$, $\mymat{\beta}=\mymat{0}$, and $\mymat{\Sigma}$ is a diagonal matrix. In addition, $\mymat{\xi}$ and $\mymat{\gamma_1}$ are diagonal, and this assumption results in orthogonal state variables under the $\pp$ measure as well as the $\RR$ measure. Furthermore, $\mymat{\Psi}$ is the $3\times 3$ identity matrix $\mymat{I_3}$, such that there are no interactions among the state variables in the determination of the interest rate. The number of parameters in QTSM3 is 16. An important advantage is that the QTSM3 model allows for a fully closed-form solution for the expected bond price:
\begin{align}
	\E_t\left[ P(H, T)\right] &=\E_t^{\R}\left[\exp\left(- \int_H^T r_u\D u \right) \right] = \E_t^{\R}\left[\exp\left(  - \int_H^T \left( \alpha +  \mymat{Y_u}'\mymat{\Psi} \mymat{Y_u}\right)\D u \right) \right] \nonumber \\
	&=\exp\left(-\alpha(T-H) - \sum_{i=1}^3A^{\!\left( \!\begin{smallmatrix}{{b_i^*}}, {0}\\ {{c_i^*}}, {q_i^*} \end{smallmatrix} \!\right)}_{{\mu^*_i}, {\xi^*_{ii}}}\!(T-H) -\sum_{i=1}^3A^{\!\left( \!\begin{smallmatrix}{{b_i}}, {0}\\ {{c_i}}, {q_i} \end{smallmatrix}  \!\right)}_{{\mu_i}, {\xi_{ii}}}\!(H-t) \right. \nonumber \\
	&\hspace{4em}  \left. -\sum_{i=1}^3{B^{{\!\left( \!\begin{smallmatrix}{{b_i}}, {0}\\ {{c_i}}, {q_i} \end{smallmatrix}  \!\right)}\!}_{{\mu_i}, {\xi_{ii}}}(H-t){Y_{it}}}- \sum_{i=1}^3 {C^{{\!\left( \!\begin{smallmatrix}{{b_i}}, {0}\\ {{c_i}}, {q_i} \end{smallmatrix} \!\right)}\!}_{{\mu_i}, {\xi_{ii}}}(H-t){Y_{it}^2}}\right),
\end{align}
where $	{\mu^*_i} = {\mu_i} - {\gamma_{0i}}, \ \ {\xi^*_{ii}} = {\xi_{ii}} - {\gamma_{1ii}}$, and
\begin{align*}
	b_i^* & = 0, \ \ c_i^* = 0, \ \ q_i^* = 1, \ \	{{b}_i} = {B^{\!\left( \!\begin{smallmatrix}{{b_i^*}}, {0}\\ {{c_i^*}}, {q_i^*} \end{smallmatrix} \!\right)}_{{\mu^*_i}, {\xi^*_{ii}}}\!(T-H)}, \ \ {{c}_i}={C^{\!\left( \!\begin{smallmatrix}{{b_i^*}}, {0}\\ {{c_i^*}}, {q_i^*} \end{smallmatrix} \!\right)}_{{\mu^*_i}, {\xi^*_{ii}}}\!(T-H)}, \ \ q_i = 0. \end{align*}           						

For any well-defined ${b_i}, {c_i}, {q_i}$, and the coefficients ${\mu_i}$,  ${\xi_{ii}}$, we have
\begin{align}\label{eq:qtsm3}
	\begin{aligned}
		A^{\!\left( \!\begin{smallmatrix}{{b_i}}, {0}\\ {{c_i}}, {q_i} \end{smallmatrix} \!\right)}_{{\mu_i}, {\xi_{ii}}}\!(\tau)
		&=\left(\frac{\mu_i}{\beta_i}\right)^2q_i\tau-\frac{b_i^2\beta_i\Sigma_{ii}^2\left(\ue^{2\beta_i\tau}-1\right)-2b_i\mu_i\left(\ue^{\beta_i\tau}-1\right)\left[\left(\beta_i-\xi_{ii}\right)
			\left(\ue^{\beta_i\tau}-1\right)+2\beta_i\right]}
		{2\beta_i\left[\left(2c_i\Sigma_{ii}^2+\beta_i-\xi_{ii}\right)\left(\ue^{2\beta_i\tau}-1\right)+2\beta_i\right]}  \\
		&-\frac{\mu_i^2\left(\ue^{\beta_i\tau}-1\right)\left[\left(\left(\beta_i-2\xi_{ii}\right)q_i-2c_i\xi_{ii}^2\right)
			\left(\ue^{\beta_i\tau}-1\right)+2\beta_iq_i\right]}
		{\beta_i^3\left[\left(2c_i\Sigma_{ii}^2+\beta_i-\xi_{ii}\right)\left(\ue^{2\beta_i\tau}-1\right)+2\beta_i\right]}\\
		&-\frac{1}{2}\ln
		\frac{2\beta_i\ue^{\left(\beta_i-\xi_{ii}\right)\tau}}{\left(2c_i\Sigma_{ii}^2+\beta_i-\xi_{ii}\right)\left(\ue^{2\beta_i\tau}-1\right)+2\beta_i},\\
		B^{\!\left( \!\begin{smallmatrix}{{b_i}}, {0}\\ {{c_i}}, {q_i} \end{smallmatrix} \!\right)}_{{\mu_i}, {\xi_{ii}}}\!(\tau) &=\frac{2b_i\beta_i^2\ue^{\beta_i\tau}+2\mu_i\left[\left(q_i+c_i\left(\xi_{ii}+\beta_i\right)\right)
			\left(\ue^{\beta_i\tau}-1\right)^2+2c_i\beta_i\left(\ue^{\beta_i\tau}-1\right)\right]}
		{\beta_i\left[\left(2c_i\Sigma_{ii}^2+\beta_i-\xi_{ii}\right)\left(\ue^{2\beta_i\tau}-1\right)+2\beta_i\right]}, \\
		C^{\!\left( \!\begin{smallmatrix}{{b_i}}, {0}\\ {{c_i}}, {q_i} \end{smallmatrix} \!\right)}_{{\mu_i}, {\xi_{ii}}}\!(\tau) &=\frac{c_i\left[2\beta_i+\left(\ue^{2\beta_i\tau}-1\right)\left(\beta_i+\xi_{ii}\right)\right]+q_i\left(\ue^{2\beta_i\tau}-1\right)}
		{\left(2c_i\Sigma_{ii}^2+\beta_i-\xi_{ii}\right)\left(\ue^{2\beta_i\tau}-1\right)+2\beta_i},
	\end{aligned}
\end{align}
where $\beta_i = \sqrt{\xi_{ii}^2+2\Sigma_{ii}^2q_i}$. $A^{\!\left( \!\begin{smallmatrix}{{b_i^*}}, {0}\\ {{c_i^*}}, {q_i^*} \end{smallmatrix} \!\right)}_{{\mu_i^*}, {\xi_{ii}^*}}\!(\tau)$, $B^{\!\left( \!\begin{smallmatrix}{{b_i^*}}, {0}\\ {{c_i^*}}, {q_i^*} \end{smallmatrix} \!\right)}_{{\mu_i^*} {\xi_{ii}^*}}\!(\tau)$ and $C^{\!\left( \!\begin{smallmatrix}{{b_i^*}}, {0}\\ {{c_i^*}}, {q_i^*} \end{smallmatrix} \!\right)}_{{\mu_i^*}, {\xi_{ii}^*}}\!(\tau)$ are obtained similarly as in equation \eqref{eq:qtsm3} with the coefficients $({\mu_i}, {\xi_{ii}})$ replaced with the coefficients $({\mu_i^*}, {\xi_{ii}^*})$, and the parameters $({b_i}, {c_i}, {q_i})$ replaced with the parameters $({b_i^*}, {c_i^*}, {q_i^*})$.

\vspace{-3em}
\begin{center}
\item \section{Expected Future Price of Forward-Start Option using the Forward-Start $R$-Transform}\label{app:fsoexpected} \vspace{-1em}
\end{center}

As an extension of the $R$-transforms in Section \ref{sec:Rmeasure}.\ref{sec:Rtransforms}, we propose the following \emph{forward-start $R$-transform}, which can be used to obtain the expected future price of a forward-start option with the terminal date $T$ payoff given as $(S_T - kS_{T_0})^+$, with $T_0$ as the intermediate future date on which the strike price $kS_{T_0}$ is determined:
\begin{align*}
	\psi^R(a_1, a_2, \mymat{z}) &\triangleq \E_t^\RR\left[\exp{\left(-\int_H^T r_u\D u+ a_1 \ln S_T + a_2 \ln S_{T_0}+z\left(\ln S_T-\ln S_{T_0}\right)\right)}\right].
\end{align*}
The analytical solution of the above forward-start $R$-transform can be derived under the affine-jump diffusion processes given in \cite{Duffie2000Transform} using the $\Ro$ measure.\footnote{This solution also extends the work of \cite{kruse2005pricing} for computing the \emph{expected future price} of the forward-start call option under the affine-jump diffusion processes given in \cite{Duffie2000Transform}. There exist two solutions of $\psi^R(a_1, a_2, \mymat{z})$ depending upon whether $t\leq H\leq T_0\leq T$ or $t\leq T_0\leq H\leq T$.} The solution of $\psi^R(a_1, a_2, \mymat{z})$ can be used to obtain the expected future price of the forward-start call option as follows:
\begin{align*}
	\E_t \left[ C_H \right] &=\E^\R_t \left[\ue^{-\int_H^T r_u\D u} \left( S_T - k S_{T_0} \right)^+\right] \nonumber \\
	&=   \E_t^{\R}\left[\ue^{-\int_H^T r_u\D u} S_T\I_{\left\{\ln S_T-\ln S_{T_0}>\ln k \right\}}\right] - k\E_t^{\R}\left[  \ue^{-\int_H^T r_u\D u}S_{T_0} \I_{\left\{\ln S_T-\ln S_{T_0}>\ln k \right\}}\right] \\
	&=\E_t^{\R}\left[\ue^{-\int_H^T r_u \D u }S_T\right]\Pi_1 -k\E_t^\R\left[\ue^{-\int_H^T r_u \D u } S_{T_0}\right]\Pi_2.
\end{align*}

Using the Fourier inversion method outlined in \cite{bakshi2000spanning} and \cite{Duffie2000Transform}, the expected future price of the forward-start option can be solved as
\begin{align*}
	\E_t \left[ C_H \right] = \E_t^{\R}\left[\ue^{-\int_H^T r_u \D u }S_T\right]\Pi_1 -k\E_t^\R\left[\ue^{-\int_H^T r_u \D u } S_{T_0}\right]\Pi_2,
\end{align*}
where
\begin{align*}
	\E_t^{\R}\left[\ue^{-\int_H^T r_u \D u }S_T\right] =\psi^R \left(1,0,0\right)  \quad \text{and} \quad
	\E_t^\R\left[\ue^{-\int_H^T r_u \D u } S_{T_0}\right]=\psi^R \left(0,1,0\right),
\end{align*}
and for $j=1,2$, the two probabilities $\Pi_j$ are calculated by using the Fourier inversion formula:
\begin{align*}
	\Pi_j = \frac{1}{2} +\frac{1}{\pi} \int_0^{\infty}\operatorname{Re}\left[\frac{\ue^{-iu\ln k} \Phi^{j}(u)}{iu}\right]\D u,
\end{align*}
with the corresponding characteristic functions $\Phi^{1}(u)$ and $\Phi^{2}(u)$, given as
\begin{align*}
	\Phi^{1}(u) = \frac{\psi^R \left(1, 0, iu\right)}{\psi^R \left(1,0,0\right)} \quad \text{and} \quad \Phi^{2}(u) = \frac{ \phi^R\left(0,1,iu\right)}{\phi^R\left(0,1,0\right)}.
\end{align*}

\vspace{-3em}
\begin{center}
\item \section{Expected Future Option Price under the Merton Model with Vasicek Short Rate Process} \label{app:MertonVasicek} \vspace{-1em}
\end{center}

Consider a special case of equations \eqref{eq:MertonS} and \eqref{eq:MertonP} with $\sigma(s) = \sigma$, $\mu(s)$ = $r_s$ + $\gamma\sigma$, $\rho(s) = \rho$, with the parameters $\sigma$ and $\rho$, given as constants, and $\gamma = \left(\mu(s) - r_s\right)/\sigma$ (defined as the asset's risk premium per unit of asset return volatility), also given as a constant. Assume that the bond price process is given by the single-factor Vasicek model (see Internet Appendix Section \ref{app:someATSM}.\ref{app:Vas}). The asset price process, the short rate process, and the bond price process under these specifications are given as:
\begin{align}\label{eq:assetprocess1}
	\begin{aligned}
		\frac{\D S_s}{S_s} &= \left(r_s + \gamma\sigma \right)\D s + \sigma \D W_{1s}^{\p}, \\ %
		\D r_s &= \alpha_r \left(m_r - r_s\right)\D s + \sigma_r  \D W_{2s}^{\p}, \\
		\frac{\D P(s,T)}{P(s,T)} &= \left(r_s - \gamma_r \sigma_r B_{\alpha_r}(T-s)\right)\D s - \sigma_r B_{\alpha_r} (T-s)\D W_{2s}^{\p},
	\end{aligned}
\end{align}
where $B_{\alpha_r}(\tau)=\left(1-\ue^{-\alpha_r \tau}\right)/\alpha_r$, and $\gamma_r$ is the market price of interest rate risk in the \cite{Vasicek1977An} model.

The full solution of the expected call price formula given by equation \eqref{eq:Mertonoptexp} can be obtained by solving $\E_t\left[S_H\right]$, $\E_t\left[P(H,T)\right]$, and $v_p$ for the above processes.
Using Lemma \ref{lemma:reeo} and Internet Appendix Section \ref{app:DiffEEM}.\ref{app:Rforward}, the normalized asset price $V=S/P(\cdot,T)$ follows
\begin{align*}
	\begin{aligned}
		\frac{\D V_s}{V_s} & = \I_{\{s< H\}}\left[\gamma\sigma +\gamma_r\sigma_r B_{\alpha_r}(T-s)  +  \left(\rho\sigma +\sigma_r B_{\alpha_r}(T-s) \right)\sigma_rB_{\alpha_r}(H-s)\right]\D s \\
		&+ \sigma\D W_{1s}^\RTo +\sigma_r B_{\alpha_r}(T-s) \D W_{2s}^\RTo.
	\end{aligned}
\end{align*}
We can see that $V_T$ is also lognormally distributed under the $\RTo$ measure. The above equation can be derived by using the relationship between the Brownian motions under $\RTo$  and the Brownian motions under  $\pp$ (as shown in Internet Appendix Section \ref{app:DiffEEM}.\ref{app:Rforward}), as follows:
\begin{align*}
	\D W_{1s}^\RTo&=\D W_{1s}^\p + \left[\I_{\left\{s\geq H\right\}}\gamma + \rho\sigma_r B_{\alpha_r}{(T-s)}  - \I_{\{s< H\}}\rho\sigma_r B_{\alpha_r}{(H-s)}  \right]\D s,  \\
	\D W_{2s}^\RTo&=\D W_{2s}^\p + \left[\I_{\left\{s\geq H\right\}}{\gamma}_r + \sigma_r B_{\alpha_r}{(T-s)}  - \I_{\{s< H\}}\sigma_r B_{\alpha_r}{(H-s)}  \right]\D s.
\end{align*}

Therefore, under the above specifications, the expected call option price $\E_t[C_H]$ is given by equation \eqref{eq:Mertonoptexp} with that $\E_t[P(H,T)]$ is derived in equation \eqref{eq:vasbondprice}, and
\begin{align}\label{eq:ESvas}
	\begin{aligned}
		\E_t[S_H] &= S_t\exp\left( \left(m_r + \gamma\sigma + \frac{\sigma_r^2}{2\alpha_r^2} + \frac{\rho\sigma\sigma_r}{\alpha_r} \right)\left(H-t\right) + \left(r_t-m_r\right)B_{\alpha_r}\left(H-t\right) \right. \\
		&\quad \quad \left.  + \frac{\sigma_r^2}{2\alpha_r^2}\left[B_{2\alpha_r}(H-t) - 2B_{\alpha_r}(H-t)\right] - \frac{\rho\sigma\sigma_r}{\alpha_r}B_{\alpha_r}(H-t) \right),
	\end{aligned}
\end{align}
and
\begin{align*}
	v_p = \sqrt{\left(\sigma^2 + \frac{\sigma_r^2}{\alpha_r^2} + \frac{2\rho\sigma\sigma_r}{\alpha_r}\right)\left(T-t\right) +  \frac{\sigma_r^2}{\alpha_r^2}\left[B_{2\alpha_r}(T-t) - 2B_{\alpha_r}(T-t)\right] - \frac{2\rho\sigma\sigma_r}{\alpha_r}B_{\alpha_r}(T-t)}.
\end{align*}

The assumption of a constant value for $\gamma$ is not inconsistent with various multifactor asset pricing models, such as Merton's (\citeyear{merton1973intertemporal}) ICAPM and \citeauthor{connor1989intertemporal}'s (\citeyear{connor1989intertemporal}) intertemporal APT. For example, under the Merton's (\citeyear{merton1973intertemporal}) ICAPM framework, the market portfolio return and the projection of the asset return on the market portfolio return and the short rate process are given as follows:
\begin{align}\label{eq:projection}
	\begin{aligned}
		\frac{\D S_s^M}{S_s^M} &= \left(r_s + \gamma_M\sigma_M\right)\D s + \sigma_M \D W_{s,M}^{\p}, \\
		\frac{\D S_s}{S_s} &= \E\left[ \frac{\D S_s}{S_s}\right] +   \beta_1  \left[ \frac{\D S_s^M}{S_s^M} - \E\left[ \frac{\D S_s^M}{S_s^M}  \right] \right]+  \beta_2 \left[\D r_s - \E\left[ \D r_s  \right] \right] + \sigma_e \D W_{s,e}^{\p},
	\end{aligned}
\end{align}
where $\gamma_M$ is the Sharpe ratio or the market risk premium per unit of market return volatility, $\beta_1$ and $\beta_2$ are the systematic risk exposures of the asset with respect to mean-zero market return factor and mean-zero short rate change factor, $\sigma_e \D W_{s,e}^{\p}$ is the idiosyncratic risk component of the asset return, and $ W_{s,e}^{\p}$ is assumed to be independent of both $ W_{s,M}^{\p}$ and $ W_{2s}^{\p}$. %
The coefficients $\beta_1$ and $\beta_2$ are obtained as constants since all volatilities and correlations that these coefficents depend on are assumed to be constants in equations \eqref{eq:assetprocess1} and \eqref{eq:projection}. Under these assumptions, Merton's (\citeyear{merton1973intertemporal}) ICAPM implies that the asset risk premium in equation \eqref{eq:assetprocess1} is:
\begin{align*}%
	\mu(s) - r_s = \gamma\sigma = \beta_1 \gamma_M\sigma_M + \beta_2 \gamma_r\sigma_r,
\end{align*}
which implies that $\gamma$ is a constant equal to:
\begin{align*}%
	\gamma= \beta_1 \gamma_M\sigma_M/\sigma + \beta_2 \gamma_r\sigma_r/\sigma.
\end{align*}
This relation is also implied by the intertemporal APT of \cite{connor1989intertemporal} using the argument of diversification in large asset markets based upon the assumption of weakly correlated idiosyncratic terms, and the short rate process given by the \cite{Vasicek1977An} model.

\vspace{-3em}
\begin{center}
\item \section{Expected Future Price of an Asset Exchange Option}\label{iapp:margrabe} \vspace{-1em}
\end{center}

While all options contracts written on tradable assets give the right to exchange one asset for another, the term ``asset exchange options" is typically used for those option contracts whose strike price is a stochastic asset value instead of a constant. The asset $S$ to be exchanged (as the strike price) is used as the numeraire for pricing these asset exchange options (e.g., see  \cite{margrabe1978value}, \cite{rady1997option}, and examples in \cite{vecer2011stochastic}). We define the EMM corresponding to the numeraire $S$ as the $\QS$ measure, and the corresponding numeraire-specific EEM as the $\RSo$ measure, given as a special case of the $\Reeo$ measure with the numeraire $G$ = $S$. The definition of the EEM $\RSo$ assumes that the asset $S$ is any tradable asset other than the money market account and the $T$-maturity pure discount bond.

For simplicity of exposition, we use the \cite{margrabe1978value} model to demonstrate the use of EEM $\RSo$ for obtaining the expected future price of a call option that gives the option holder the right to exchange an asset $S=S_1$ for another asset $S_2$. We extend the \cite{margrabe1978value} model by allowing a stochastic short rate, and assume that the physical asset price processes $S_1$ and $S_2$, and the physical short rate process are given as follows: \begin{align*}
	\begin{aligned}
		\frac{\D S_{1s}}{S_{1s}} &= \left(r_s +  \gamma_1 \sigma_1\right) \D s +  \sigma_1  \D W_{1s}^{\p},  \\
		\frac{\D S_{2s}}{S_{2s}} &= \left(r_s +  \gamma_2 \sigma_2\right) \D s +  \sigma_2  \D W_{2s}^{\p}, \\
		\D r_s &= \alpha_r \left(m_r - r_s\right)\D s + \sigma_r  \D W_{rs}^{\p},
	\end{aligned}
\end{align*}
where $W_{1}^{\p}$ and $W_{2}^{\p}$ are correlated Brownian motions with a correlation coefficient equal to $\rho_{12}$; $W_r^{\p}$ is correlated with $W_{1}^{\p}$ and $W_{2}^{\p}$ with correlation coefficients equal to $\rho_{1r}$ and $\rho_{2r}$, respectively; and $\gamma_1$ and $\gamma_2$ are defined as the risk premium per unit of volatility for $S_1$ and $S_2$, respectively.

It is convenient to price the exchange option using equation \eqref{eq:priceQIT1}, with $G$ = $S_1$ as the numeraire asset, $\Reeo$ = $\RSo$ using this numeraire asset, $\X_T = C_T = (S_{2T} - S_{1T})^+$ as the terminal payoff, and the numeraire values given as $G_T$ = $S_{1T}$, and $G_H$ = $S_{1H}$, as follows:
\begin{align*}
	\begin{aligned}
		\E_t[C_H]  &= \E_t^\p\left[S_{1H}\right]\E_t^\RSo\left[\frac{(S_{2T}-S_{1T})^+}{S_{1T}}\right]  \\
		&=\E_t^\p\left[S_{1H}\right]\E_t^{\RSo}\left[(V_{T}-1)^+\right], \\
	\end{aligned}
\end{align*}
where $V = S_2/S_1$ is the asset price ratio. Therefore, it suffices for us to obtain the process of $V$ under the $\RSo$ measure, which can be derived using Lemma \ref{lemma:reeo} as follows:
\begin{align*}
	\frac{\D V_s}{V_s} = \I_{\{s< H\}}\left[\gamma_2\sigma_2-\gamma_1\sigma_1 +\left(\rho_{2r}\sigma_2 - \rho_{1r}\sigma_1\right)\sigma_r B_{\alpha_r}(H-s) \right]\D s + \sigma_2\D W_{2s}^\RSo - \sigma_1\D W_{1s}^\RSo,
\end{align*}
where $B_{\alpha_r}(\tau)=\left(1-\ue^{-\alpha_r \tau}\right)/\alpha_r$.

Note that since all volatilities, correlations, and drifts are either constant or deterministic in the above equation, $V_T$ is lognormally distributed under the $\RSo$ measure. Thus, using equation \eqref{eq:priceQIT1R} and similar derivation logic of the expected future price of a call option under the \cite{Merton1973Theory} model in Section \ref{sec:R1Tmeasure}.\ref{sec:R1TmeasureMerton} (given by equation \eqref{eq:Mertonoptexp}), the solution of the expected future price of the exchange option is given as follows:
\begin{align*}
	\E_t[C_H] = \E_t[S_{2H}]\mathcal{N}\left(\hat{d_1}\right) - \E_t[S_{1H}]\mathcal{N}\left(\hat{d_2}\right),
\end{align*}
where
\begin{align*}
	\hat{d}_1 &= \frac{1}{v_p}\ln\frac{\E_t\left[S_{2H}\right]}{\E_t\left[S_{1H}\right] } + \frac{v_p}{2}, \ \  \quad
	\hat{d}_2 = \frac{1}{v_p}\ln\frac{\E_t\left[S_{2H}\right]}{\E_t\left[S_{1H}\right]} - \frac{v_p}{2},
\end{align*}
with $v_p  = \sqrt{\left(\sigma_1^2 + \sigma_2^2 - 2\rho_{12}\sigma_1\sigma_2\right)\left(T-t\right)}$, %
$\E_t[S_{jH}]$ is given similarly as in equation \eqref{eq:ESvas} with ($S_t$, $\gamma$, $\sigma$, $\rho$) replaced with ($S_{jt}$, $\gamma_j$, $\sigma_j$, $\rho_{jr}$) for $j=1,2$.

It is important to note that while the valuation of an exchange option does not depend upon interest rates under the \cite{margrabe1978value} model, the expected future price of the exchange option as given by the above equation \emph{does depend upon the parameters of the interest rate process.}. This dependence arises not from the hedge portfolio's integrated variance $v_p^2$, but through the expected future prices $\E_t[S_{1H}]$ and $\E_t[S_{2H}]$.

\vspace{-3em}
\begin{center}
\item \section{A Procedure to Extract the Expected FSPD}\label{sec:extractFSPD}\vspace{-1em}
\end{center}

To obtain the expected FSPD from expected future prices of options using Theorem \ref{thm:Breedenext}, we consider extending the methods used for obtaining the SPD (defined as the second derivative of the call price function with respect to the strike price in \cite{breeden1978prices}) from current option prices. Both parametric estimation methods and non-parametric estimation methods can be used for extracting the SPD from option prices \citep{Figlewski2018Risk}.  %
The parametric methods include %
composite distributions based on the normal/lognormal \cite[see, e.g.,][]{Jarrow1982Approximate, Madan1994Contingent} and  mixture models \cite[see, e.g.,][]{Melick1997Recovering, Soderlind1997New}. The non-parametric estimation methods can be classified  into three categories: maximum entropy \cite[see, e.g.,][]{Stutzer1996A, Buchen1996The}, kernel \cite[see, e.g.,][]{Ait-Sahalia1998Non, Pritsker1998Non}, and curve-fitting \cite[see, e.g.,][]{jackwerth2000recovering, jackwerth2004option, rosenberg2002empirical, Figlewski2009Estimating} methods.
We extend the fast and stable curve-fitting method of \cite{jackwerth2004option} that fits the implied volatilities using the \cite{black1973pricing} model, and then numerically approximates the SPD.\footnote{For an extensive review of the curve-fitting methods, see \cite{jackwerth2004option} and \cite{Figlewski2018Risk}.}. %

We briefly recall the main steps for using the \cite{jackwerth2004option} method. %
The first step inverts the option prices into Black-Scholes implied volatilities, $\bar{\sigma}(K/S_t)$, using a finite number of options with different strike prices. This allows greater stability since implied volatilities for different strike prices are much closer in size than the corresponding option prices. Let $\left\{K_1, K_2, ..., K_N\right\}$ represent a set of strike prices ordered from the lowest to the highest. Using a ``smoothness" parameter as in \cite{jackwerth2004option}, a smooth BS implied volatility (IV) curve  $\hat{\sigma}(K/S_t)$ is fitted by optimization. In the next step, each optimized IV $\hat{\sigma}(K_i/S_t)$ is converted back into a call price using the Black-Scholes equation, that is, $C_i = C_{\text{BS}}\left(S_t, K_i,t, T, r; \hat{\sigma}(K_i/S_t)\right)$, and the SPD (discounted $\QT$ density) is estimated numerically as follows: %
For any $1<i<N$,
\begin{align*}
	\begin{aligned}
		f_t(S_T = K_i) &= P(t,T) p_t^\QT(S_T=K_i)\\
		&= \frac{\partial^2 C_{\text{BS}}\left(S_t, K_i,t, T, r; \hat{\sigma}(K_i/S_t)\right)}{\partial K_i^2} \\
		& \approx \frac{C_{i+1}-2C_i + C_{i-1}}{\left(\Delta K\right)^2}.
	\end{aligned}
\end{align*}

We extend the fast and stable curve-fitting method of \cite{jackwerth2004option} for obtaining the expected FSPD from the expected future option prices by not only fitting implied volatilities using the current call price formula of the \cite{black1973pricing} model, but also the \emph{implied drifts} using the expected future call price formula of the \cite{black1973pricing} model given in equation \eqref{eq:BSexp}, and then numerically approximate the expected FSPD using Theorem \ref{thm:Breedenext} (equation \eqref{eq:AD-RTo}). This extension is outlined in the following steps. The expected returns of options are first estimated using historical data, and then expected option prices can be obtained for a finite number of strike prices given by the set $\left\{K_1, K_2, ..., K_N\right\}$. As in \cite{jackwerth2004option}, the current option prices are inverted into BS IV $\bar{\sigma}(K_i/S_t)$ and using a ``smoothness" parameter, a smooth BS IV curve  $\hat{\sigma}(K/S_t)$ is fitted by optimization. Then, using fixed optimized IV $\hat{\sigma}(K/S_t)$, we invert the expected future option price into the BS implied drift $\bar{\mu}(K/S_t)$ using the Black-Scholes expected future call price in equation \eqref{eq:BSexp}. Next, following the optimization method of \cite{jackwerth2004option}, a smooth ID (implied drift) curve $\hat{\mu}(K/S_t)$ is fitted by solving a similar optimization problem that balances smoothness against the fit of the IDs. For each strike price, the optimized IV $\hat{\sigma}(K_i/S_t)$ and ID $\hat{\mu}(K_i/S_t)$, are converted back into the expected call price using the Black-Scholes expected future price in equation \eqref{eq:BSexp}, that is, $EC_i= \hat{C}_{\text{BS}}\left(S_t, K_i, t, T, H, r; \hat{\sigma}(K_i/S_t), \hat{\mu}(K_i/S_t)\right)$. Finally, the expected FSPD is numerically approximated using Theorem \ref{thm:Breedenext} (equation \eqref{eq:AD-RTo}), as follows: For any $1<i<N$,
\begin{align*}
	\begin{aligned}
		g(H, S_T = K_i)&=\E_t^\p\left[P(H,T)\right]p_t^\RTo(S_T=K_i) \\
		&= \frac{\partial^2\hat{C}_{\text{BS}}\left(S_t, K_i, t, T, H, r; \hat{\sigma}(K_i/S_t), \hat{\mu}(K_i/S_t)\right)  }{\partial K_i^2}  \\
		&\approx  \frac{EC_{i+1}-2EC_i + EC_{i-1}}{\left(\Delta K\right)^2}.
	\end{aligned}
\end{align*}

\clearpage
\bibliographystyle{jf}

\begin{thebibliography}{101}
\expandafter\ifx\csname natexlab\endcsname\relax\def\natexlab#1{#1}\fi

\bibitem[Adrian et~al.(2013)Adrian, Crump, and Moench]{adrian2013pricing}
Adrian, Tobias, Richard~K Crump, and Emanuel Moench, 2013, Pricing the term
  structure with linear regressions, {\em Journal of Financial Economics\/}
  110, 110--138.

\bibitem[Ahn et~al.(2002)Ahn, Dittmar, and Gallant]{Ahn2002Quadratic}
Ahn, Dong~Hyun, Robert~F. Dittmar, and A.~Ronald Gallant, 2002, Quadratic term
  structure models: Theory and evidence, {\em The Review of Financial
  Studies\/} 15, 243--288.

\bibitem[Amin and Jarrow(1991)]{AMIN1991Pricing}
Amin, Kaushik~I., and Robert~A. Jarrow, 1991, Pricing foreign currency options
  under stochastic interest rates, {\em Journal of International Money and
  Finance\/} 10, 310 -- 329.

\bibitem[Baele et~al.(2019)Baele, Driessen, Ebert, Londono, and
  Spalt]{baele2019cumulative}
Baele, Lieven, Joost Driessen, Sebastian Ebert, Juan~M Londono, and Oliver~G
  Spalt, 2019, Cumulative prospect theory, option returns, and the variance
  premium, {\em The Review of Financial Studies\/} 32, 3667--3723.

\bibitem[Bai et~al.(2019)Bai, Bali, and Wen]{bai2019common}
Bai, Jennie, Turan~G Bali, and Quan Wen, 2019, Common risk factors in the
  cross-section of corporate bond returns, {\em Journal of Financial
  Economics\/} 131, 619--642.

\bibitem[Bakshi et~al.(1997)Bakshi, Cao, and Chen]{bakshi1997empirical}
Bakshi, Gurdip, Charles Cao, and Zhiwu Chen, 1997, Empirical performance of
  alternative option pricing models, {\em The Journal of Finance\/} 52,
  2003--2049.

\bibitem[Bakshi and Madan(2000)]{bakshi2000spanning}
Bakshi, Gurdip, and Dilip Madan, 2000, Spanning and derivative-security
  valuation, {\em Journal of Financial Economics\/} 55, 205--238.

\bibitem[Bakshi et~al.(2010)Bakshi, Madan, and Panayotov]{bakshi2010returns}
Bakshi, Gurdip, Dilip Madan, and George Panayotov, 2010, Returns of claims on
  the upside and the viability of {U}-shaped pricing kernels, {\em Journal of
  Financial Economics\/} 97, 130--154.

\bibitem[Balduzzi et~al.(1996)Balduzzi, Das, Foresi, and
  Sundaram]{balduzzi1996simple}
Balduzzi, Pierluigi, Sanjiv~Ranjan Das, Silverio Foresi, and Rangarajan~K
  Sundaram, 1996, A simple approach to three-factor affine term structure
  models, {\em The Journal of Fixed Income\/} 6, 43--53.

\bibitem[Bali et~al.(2021)Bali, Subrahmanyam, and Wen]{bali2021long}
Bali, Turan~G, Avanidhar Subrahmanyam, and Quan Wen, 2021, Long-term reversals
  in the corporate bond market, {\em Journal of Financial Economics\/} 139,
  656--677.

\bibitem[Bansal and Zhou(2002)]{bansal2002term}
Bansal, Ravi, and Hao Zhou, 2002, Term structure of interest rates with regime
  shifts, {\em The Journal of Finance\/} 57, 1997--2043.

\bibitem[Bates(1996)]{bates1996jumps}
Bates, David~S, 1996, Jumps and stochastic volatility: Exchange rate processes
  implicit in {D}eutsche mark options, {\em The Review of Financial Studies\/}
  9, 69--107.

\bibitem[Bates(2000)]{bates2000post}
Bates, David~S, 2000, Post-'87 crash fears in the {S\&P} 500 futures option
  market, {\em Journal of Econometrics\/} 94, 181--238.

\bibitem[Becherer(2001)]{becherer2001numeraire}
Becherer, Dirk, 2001, The numeraire portfolio for unbounded semimartingales,
  {\em Finance and Stochastics\/} 5, 327--341.

\bibitem[Becker and Ivashina(2015)]{becker2015reaching}
Becker, Bo, and Victoria Ivashina, 2015, Reaching for yield in the bond market,
  {\em The Journal of Finance\/} 70, 1863--1902.

\bibitem[Bj{\"o}rk(2009)]{bjork2009arbitrage}
Bj{\"o}rk, Tomas, 2009, {\em Arbitrage Theory in Continuous Time - 3rd
  edition\/} (Oxford University Press).

\bibitem[Black and Cox(1976)]{Black1976VALUING}
Black, Fischer, and John~C. Cox, 1976, Valuing corporate securities: Some
  effects of bond indenture provisions, {\em The Journal of Finance\/} 31,
  351--367.

\bibitem[Black and Scholes(1973)]{black1973pricing}
Black, Fischer, and Myron Scholes, 1973, The pricing of options and corporate
  liabilities, {\em Journal of Political Economy\/} 81, 637--654.

\bibitem[Bondarenko(2003)]{bondarenko2003statistical}
Bondarenko, Oleg, 2003, Statistical arbitrage and securities prices, {\em The
  Review of Financial Studies\/} 16, 875--919.

\bibitem[Bondarenko(2014)]{bondarenko2014put}
Bondarenko, Oleg, 2014, Why are put options so expensive?, {\em The Quarterly
  Journal of Finance\/} 4, 1450015.

\bibitem[Borodin and Salminen(2002)]{borodin2002handbook}
Borodin, Andrei~N, and Paavo Salminen, 2002, {\em Handbook of Brownian Motion -
  Facts and Formulae, 2nd edition\/} (Birkh{\"a}user).

\bibitem[Brace et~al.(1997)Brace, G{\c{}}~atarek, and Musiela]{brace1997market}
Brace, Alan, Dariusz G{\c{}}~atarek, and Marek Musiela, 1997, The market model
  of interest rate dynamics, {\em Mathematical Finance\/} 7, 127--155.

\bibitem[Breeden and Litzenberger(1978)]{breeden1978prices}
Breeden, Douglas~T, and Robert~H Litzenberger, 1978, Prices of state-contingent
  claims implicit in option prices, {\em Journal of Business\/} 51, 621--651.

\bibitem[Broadie et~al.(2009)Broadie, Chernov, and
  Johannes]{broadie2009understanding}
Broadie, Mark, Mikhail Chernov, and Michael Johannes, 2009, Understanding index
  option returns, {\em The Review of Financial Studies\/} 22, 4493--4529.

\bibitem[B{\"u}chner and Kelly(2022)]{buchner2019latent}
B{\"u}chner, Matthias, and Bryan~T Kelly, 2022, A factor model for option
  returns, {\em Journal of Financial Economics\/} 143, 1140--1161.

\bibitem[Carr et~al.(2002)Carr, Geman, Madan, and Yor]{carr2002fine}
Carr, Peter, H{\'e}lyette Geman, Dilip~B Madan, and Marc Yor, 2002, The fine
  structure of asset returns: An empirical investigation, {\em The Journal of
  Business\/} 75, 305--332.

\bibitem[Carr and Wu(2003)]{Carr2003The}
Carr, Peter, and Liuren Wu, 2003, The finite moment log stable process and
  option pricing, {\em The Journal of Finance\/} 58, 753--777.

\bibitem[Chacko and Das(2002)]{Chacko2002Pricing}
Chacko, George, and Sanjiv Das, 2002, Pricing interest rate derivatives: A
  general approach, {\em The Review of Financial Studies\/} 15, 195--241.

\bibitem[Chambers et~al.(2014)Chambers, Foy, Liebner, and
  Lu]{chambers2014index}
Chambers, Donald~R, Matthew Foy, Jeffrey Liebner, and Qin Lu, 2014, Index
  option returns: Still puzzling, {\em The Review of Financial Studies\/} 27,
  1915--1928.

\bibitem[Cheng(2019)]{cheng2019vix}
Cheng, Ing-Haw, 2019, The {VIX} premium, {\em The Review of Financial
  Studies\/} 32, 180--227.

\bibitem[Chung et~al.(2019)Chung, Wang, and Wu]{chung2019volatility}
Chung, Kee~H, Junbo Wang, and Chunchi Wu, 2019, Volatility and the
  cross-section of corporate bond returns, {\em Journal of Financial
  Economics\/} 133, 397--417.

\bibitem[Cochrane(1996)]{cochrane1996cross}
Cochrane, John~H, 1996, A cross-sectional test of an investment-based asset
  pricing model, {\em Journal of Political Economy\/} 104, 572--621.

\bibitem[Cochrane and Piazzesi(2005)]{cochrane2005bond}
Cochrane, John~H, and Monika Piazzesi, 2005, Bond risk premia, {\em American
  Economic Review\/} 95, 138--160.

\bibitem[Collin-Dufresne and Goldstein(2001)]{collin2001credit}
Collin-Dufresne, Pierre, and Robert~S Goldstein, 2001, Do credit spreads
  reflect stationary leverage ratios?, {\em The Journal of Finance\/} 56,
  1929--1957.

\bibitem[Collin-Dufresne et~al.(2008)Collin-Dufresne, Goldstein, and
  Jones]{collin2008identification}
Collin-Dufresne, Pierre, Robert~S Goldstein, and Christopher~S Jones, 2008,
  Identification of maximal affine term structure models, {\em The Journal of
  Finance\/} 63, 743--795.

\bibitem[Collin-Dufresne et~al.(2009)Collin-Dufresne, Goldstein, and
  Jones]{collin2009can}
Collin-Dufresne, Pierre, Robert~S Goldstein, and Christopher~S Jones, 2009, Can
  interest rate volatility be extracted from the cross section of bond yields?,
  {\em Journal of Financial Economics\/} 94, 47--66.

\bibitem[Connor and Korajczyk(1989)]{connor1989intertemporal}
Connor, Gregory, and Robert~A Korajczyk, 1989, An intertemporal equilibrium
  beta pricing model, {\em The Review of Financial Studies\/} 2, 373--392.

\bibitem[Constantinides et~al.(2013)Constantinides, Jackwerth, and
  Savov]{constantinides2013puzzle}
Constantinides, George~M, Jens~Carsten Jackwerth, and Alexi Savov, 2013, The
  puzzle of index option returns, {\em Review of Asset Pricing Studies\/} 3,
  229--257.

\bibitem[Coval and Shumway(2001)]{Coval2001Expected}
Coval, Joshua~D., and Tyler Shumway, 2001, Expected option returns, {\em The
  Journal of Finance\/} 56, 983--1009.

\bibitem[Cox et~al.(1985)Cox, Ingersoll, and Ross]{Cox1985A}
Cox, John~C., Jonathan~E. Ingersoll, and Stephen~A. Ross, 1985, A theory of the
  term structure of interest rates, {\em Econometrica\/} 53, 385--407.

\bibitem[Cox and Ross(1976)]{cox1976valuation}
Cox, John~C, and Stephen~A Ross, 1976, The valuation of options for alternative
  stochastic processes, {\em Journal of Financial Economics\/} 3, 145--166.

\bibitem[Cox et~al.(1979)Cox, Ross, and Rubinstein]{cox1979option}
Cox, John~C, Stephen~A Ross, and Mark Rubinstein, 1979, Option pricing: A
  simplified approach, {\em Journal of Financial Economics\/} 7, 229--263.

\bibitem[Dai and Singleton(2000)]{dai2000specification}
Dai, Qiang, and Kenneth~J Singleton, 2000, Specification analysis of affine
  term structure models, {\em The Journal of Finance\/} 55, 1943--1978.

\bibitem[Dai and Singleton(2002)]{Dai2002Expectation}
Dai, Qiang, and Kenneth~J. Singleton, 2002, Expectation puzzles, time-varying
  risk premia, and dynamic models of the term structure, {\em Journal of
  Financial Economics\/} 63, 415--441.

\bibitem[Dew-Becker et~al.(2017)Dew-Becker, Giglio, Le, and
  Rodriguez]{dew2017price}
Dew-Becker, Ian, Stefano Giglio, Anh Le, and Marius Rodriguez, 2017, The price
  of variance risk, {\em Journal of Financial Economics\/} 123, 225--250.

\bibitem[Driessen and Maenhout(2007)]{driessen2007empirical}
Driessen, Joost, and Pascal Maenhout, 2007, An empirical portfolio perspective
  on option pricing anomalies, {\em Review of Finance\/} 11, 561--603.

\bibitem[Duffee(2002)]{duffee2002term}
Duffee, Gregory~R, 2002, Term premia and interest rate forecasts in affine
  models, {\em The Journal of Finance\/} 57, 405--443.

\bibitem[Duffie and Kan(1996)]{Duffie1996A}
Duffie, Darrell, and Rui Kan, 1996, A yield-factor model of interest rates,
  {\em Mathematical Finance\/} 6, 379--406.

\bibitem[Duffie et~al.(2000)Duffie, Pan, and Singleton]{Duffie2000Transform}
Duffie, Darrell, Jun Pan, and Kenneth Singleton, 2000, Transform analysis and
  asset pricing for affine jump-diffusions, {\em Econometrica\/} 68,
  1343--1376.

\bibitem[Eraker et~al.(2015)Eraker, Shaliastovich, and Wang]{eraker2015durable}
Eraker, Bj{\o}rn, Ivan Shaliastovich, and Wenyu Wang, 2015, Durable goods,
  inflation risk, and equilibrium asset prices, {\em The Review of Financial
  Studies\/} 29, 193--231.

\bibitem[Eraker and Wu(2017)]{eraker2017explaining}
Eraker, Bj{\o}rn, and Yue Wu, 2017, Explaining the negative returns to
  volatility claims: An equilibrium approach, {\em Journal of Financial
  Economics\/} 125, 72--98.

\bibitem[Garman and Kohlhagen(1983)]{garman1983foreign}
Garman, Mark~B, and Steven~W Kohlhagen, 1983, Foreign currency option values,
  {\em Journal of International Money and Finance\/} 2, 231--237.

\bibitem[Gebhardt et~al.(2005)Gebhardt, Hvidkjaer, and
  Swaminathan]{gebhardt2005cross}
Gebhardt, William~R, Soeren Hvidkjaer, and Bhaskaran Swaminathan, 2005, The
  cross-section of expected corporate bond returns: Betas or characteristics?,
  {\em Journal of Financial Economics\/} 75, 85--114.

\bibitem[Geman(1989)]{geman1989importance}
Geman, Helyette, 1989, The importance of the forward neutral probability in a
  stochastic approach of interest rates, Technical report, Working paper,
  ESSEC.

\bibitem[Geman et~al.(1995)Geman, El~Karoui, and Rochet]{geman1995changes}
Geman, Helyette, Nicole El~Karoui, and Jean-Charles Rochet, 1995, Changes of
  numeraire, changes of probability measure and option pricing, {\em Journal of
  Applied Probability\/} 32, 443--458.

\bibitem[Grabbe(1983)]{Grabbe1983The}
Grabbe, J.~Orlin, 1983, The pricing of call and put options on foreign
  exchange, {\em Journal of International Money and Finance\/} 2, 239--253.

\bibitem[Hansen(1982)]{hansen1982large}
Hansen, Lars~Peter, 1982, Large sample properties of generalized method of
  moments estimators, {\em Econometrica\/} 50, 1029--1054.

\bibitem[Hansen and Richard(1987)]{hansen1987role}
Hansen, Lars~Peter, and Scott~F Richard, 1987, The role of conditioning
  information in deducing testable restrictions implied by dynamic asset
  pricing models, {\em Econometrica\/}  587--613.

\bibitem[Harrison and Kreps(1979)]{harrison1979martingales}
Harrison, J~Michael, and David~M Kreps, 1979, Martingales and arbitrage in
  multiperiod securities markets, {\em Journal of Economic Theory\/} 20,
  381--408.

\bibitem[Heath et~al.(1992)Heath, Jarrow, and Morton]{Heath1992Bond}
Heath, David, Robert Jarrow, and Andrew Morton, 1992, Bond pricing and the term
  structure of interest rates: A new methodology for contingent claims
  valuation, {\em Econometrica\/} 60, 77--105.

\bibitem[Heston(1993)]{heston1993closed}
Heston, Steven~L, 1993, A closed-form solution for options with stochastic
  volatility with applications to bond and currency options, {\em The Review of
  Financial Studies\/} 6, 327--343.

\bibitem[Hilliard et~al.(1991)Hilliard, Madura, and
  Tucker]{hilliard1991currency}
Hilliard, Jimmy~E, Jeff Madura, and Alan~L Tucker, 1991, Currency option
  pricing with stochastic domestic and foreign interest rates, {\em Journal of
  Financial and Quantitative Analysis\/}  139--151.

\bibitem[Hong and Li(2005)]{hong2005nonparametric}
Hong, Yongmiao, and Haitao Li, 2005, Nonparametric specification testing for
  continuous-time models with applications to term structure of interest rates,
  {\em The Review of Financial Studies\/} 18, 37--84.

\bibitem[Hull and White(1987)]{hull1987pricing}
Hull, John, and Alan White, 1987, The pricing of options on assets with
  stochastic volatilities, {\em The Journal of Finance\/} 42, 281--300.

\bibitem[Israelov and Kelly(2017)]{israelov2017forecasting}
Israelov, Roni, and Bryan~T Kelly, 2017, Forecasting the distribution of option
  returns, {\em Available at SSRN: https://ssrn.com/abstract=3033242\/} .

\bibitem[Jacka(1991)]{jacka1991optimal}
Jacka, S.D., 1991, Optimal stopping and the {A}merican put, {\em Mathematical
  Finance\/} 1, 1--14.

\bibitem[Jackwerth(2004)]{jackwerth2004option}
Jackwerth, Jens~Carsten, 2004, Option-implied risk-neutral distributions and
  risk aversion, {\em Charlotteville: Research Foundation of AIMR\/} .

\bibitem[Jackwerth and Menner(2020)]{jackwerth2020does}
Jackwerth, Jens~Carsten, and Marco Menner, 2020, Does the {R}oss recovery
  theorem work empirically?, {\em Journal of Financial Economics\/} 137,
  723--739.

\bibitem[Jamshidian(1989)]{jamshidian1989exact}
Jamshidian, Farshid, 1989, An exact bond option formula, {\em The Journal of
  Finance\/} 44, 205--209.

\bibitem[Jamshidian(1997)]{jamshidian1997libor}
Jamshidian, Farshid, 1997, Libor and swap market models and measures, {\em
  Finance and Stochastics\/} 1, 293--330.

\bibitem[Jarrow et~al.(1997)Jarrow, Lando, and Turnbull]{Jarrow1997A}
Jarrow, R.~A., D.~Lando, and S.~M. Turnbull, 1997, A {M}arkov model for the
  term structure of credit risk spreads, {\em The Review of Financial
  Studies\/} 10, 481--523.

\bibitem[Jensen et~al.(2019)Jensen, Lando, and Pedersen]{jensen2019generalized}
Jensen, Christian~Skov, David Lando, and Lasse~Heje Pedersen, 2019, Generalized
  recovery, {\em Journal of Financial Economics\/} 133, 154--174.

\bibitem[Johnson(2017)]{johnson2017risk}
Johnson, Travis~L, 2017, Risk premia and the vix term structure, {\em Journal
  of Financial and Quantitative Analysis\/} 52, 2461--2490.

\bibitem[Jones(2006)]{jones2006nonlinear}
Jones, Christopher~S, 2006, A nonlinear factor analysis of {S\&P} 500 index
  option returns, {\em The Journal of Finance\/} 61, 2325--2363.

\bibitem[Joslin et~al.(2011)Joslin, Singleton, and Zhu]{joslin2011new}
Joslin, Scott, Kenneth~J Singleton, and Haoxiang Zhu, 2011, A new perspective
  on gaussian dynamic term structure models, {\em The Review of Financial
  Studies\/} 24, 926--970.

\bibitem[Karatzas and Kardaras(2007)]{karatzas2007numeraire}
Karatzas, Ioannis, and Constantinos Kardaras, 2007, The num{\'e}raire portfolio
  in semimartingale financial models, {\em Finance and Stochastics\/} 11,
  447--493.

\bibitem[Kruse and N{\"o}gel(2005)]{kruse2005pricing}
Kruse, Susanne, and Ulrich N{\"o}gel, 2005, On the pricing of forward starting
  options in {H}eston's model on stochastic volatility, {\em Finance and
  Stochastics\/} 9, 233--250.

\bibitem[Leippold and Wu(2003)]{leippold2003design}
Leippold, Markus, and Liuren Wu, 2003, Design and estimation of quadratic term
  structure models, {\em Review of Finance\/} 7, 47--73.

\bibitem[Leland and Toft(1996)]{Leland1996Optimal}
Leland, Hayne~E., and Klaus~Bjerre Toft, 1996, Optimal capital structure,
  endogenous bankruptcy, and the term structure of credit spreads, {\em The
  Journal of Finance\/} 51, 987--1019.

\bibitem[Lin et~al.(2011)Lin, Wang, and Wu]{lin2011liquidity}
Lin, Hai, Junbo Wang, and Chunchi Wu, 2011, Liquidity risk and expected
  corporate bond returns, {\em Journal of Financial Economics\/} 99, 628--650.

\bibitem[Long(1990)]{long1990numeraire}
Long, John~B, 1990, The numeraire portfolio, {\em Journal of Financial
  Economics\/} 26, 29--69.

\bibitem[Longstaff et~al.(2005)Longstaff, Mithal, and
  Neis]{longstaff2005corporate}
Longstaff, Francis~A, Sanjay Mithal, and Eric Neis, 2005, Corporate yield
  spreads: Default risk or liquidity? new evidence from the credit default swap
  market, {\em The Journal of Finance\/} 60, 2213--2253.

\bibitem[Longstaff and Schwartz(1995)]{longstaff1995simple}
Longstaff, Francis~A, and Eduardo~S Schwartz, 1995, A simple approach to
  valuing risky fixed and floating rate debt, {\em The Journal of Finance\/}
  50, 789--819.

\bibitem[Margrabe(1978)]{margrabe1978value}
Margrabe, William, 1978, The value of an option to exchange one asset for
  another, {\em The Journal of Finance\/} 33, 177--186.

\bibitem[Martin and Wagner(2019)]{martin2019expected}
Martin, Ian, and Christian Wagner, 2019, What is the expected return on a
  stock?, {\em The Journal of Finance\/} 74, 1887--1929.

\bibitem[Merton(1973{\natexlab{a}})]{merton1973intertemporal}
Merton, Robert~C, 1973{\natexlab{a}}, An intertemporal capital asset pricing
  model, {\em Econometrica\/} 41, 867--887.

\bibitem[Merton(1973{\natexlab{b}})]{Merton1973Theory}
Merton, Robert~C., 1973{\natexlab{b}}, Theory of rational option pricing, {\em
  Bell Journal of Economics\/} 4, 141--183.

\bibitem[Merton(1974)]{Merton1974ON}
Merton, Robert~C., 1974, On the pricing of corporate debt: The risk structure
  of interest rates, {\em The Journal of Finance\/} 29, 449--470.

\bibitem[Merton(1976)]{merton1976option}
Merton, Robert~C, 1976, Option pricing when underlying stock returns are
  discontinuous, {\em Journal of Financial Economics\/} 3, 125--144.

\bibitem[Miltersen et~al.(1997)Miltersen, Sandmann, and
  Sondermann]{miltersen1997closed}
Miltersen, Kristian~R, Klaus Sandmann, and Dieter Sondermann, 1997, Closed form
  solutions for term structure derivatives with log-normal interest rates, {\em
  The Journal of Finance\/} 52, 409--430.

\bibitem[Mueller(2000)]{mueller2000simple}
Mueller, Clemens, 2000, A simple multi-factor model of corporate bond prices,
  Working paper, University of Wisconsin-Madison.

\bibitem[Nawalkha and Zhuo(2022)]{nawalkha2022sharpe}
Nawalkha, Sanjay~K, and Xiaoyang Zhuo, 2022, Risk and return dynamics of
  optioned portfolios, Working paper, University of Massachusetts, Amherst.

\bibitem[Pan(2002)]{pan2002jump}
Pan, Jun, 2002, The jump-risk premia implicit in options: Evidence from an
  integrated time-series study, {\em Journal of Financial Economics\/} 63,
  3--50.

\bibitem[Rady(1997)]{rady1997option}
Rady, Sven, 1997, Option pricing in the presence of natural boundaries and a
  quadratic diffusion term, {\em Finance and Stochastics\/} 1, 331--344.

\bibitem[Ross(1976)]{ross1976arbitrage}
Ross, Stephen~A, 1976, The arbitrage theory of capital asset pricing, {\em
  Journal of Economic Theory\/} .

\bibitem[Ross(2015)]{ross2015recovery}
Ross, Steve, 2015, The recovery theorem, {\em The Journal of Finance\/} 70,
  615--648.

\bibitem[Rubinstein(1984)]{rubinstein1984simple}
Rubinstein, Mark, 1984, A simple formula for the expected rate of return of an
  option over a finite holding period, {\em The Journal of Finance\/} 39,
  1503--1509.

\bibitem[Rubinstein(1991)]{rubinstein1991pay}
Rubinstein, Mark, 1991, Pay now, choose later, {\em Risk\/} 4, 13.

\bibitem[Shreve(2004)]{shreve2004stochastic}
Shreve, Steven~E, 2004, {\em Stochastic Calculus for Finance II:
  Continuous-Time Models\/}, volume~11 (Springer Science \& Business Media).

\bibitem[Vasicek(1977)]{Vasicek1977An}
Vasicek, Oldrich, 1977, An equilibrium characterization of the term structure,
  {\em Journal of Financial and Quantitative Analysis\/} 12, 627--627.

\bibitem[Vecer(2011)]{vecer2011stochastic}
Vecer, Jan, 2011, {\em Stochastic Finance: A Numeraire Approach\/} (CRC Press).

\end{thebibliography}


\begin{thebibliography}{31}
\expandafter\ifx\csname natexlab\endcsname\relax\def\natexlab#1{#1}\fi

\bibitem[Ahn et~al.(2002)Ahn, Dittmar, and Gallant]{Ahn2002Quadratic}
Ahn, Dong~Hyun, Robert~F. Dittmar, and A.~Ronald Gallant, 2002, Quadratic term
  structure models: Theory and evidence, {\em The Review of Financial
  Studies\/} 15, 243--288.

\bibitem[A{\"{i}}t-Sahalia and Lo(1998)]{Ait-Sahalia1998Non}
A{\"{i}}t-Sahalia, Yacine, and Andrew~W. Lo, 1998, {Nonparametric estimation of
  state-price densities implicit in financial asset prices}, {\em The Journal
  of Finance\/} 53, 499--547.

\bibitem[Bakshi and Madan(2000)]{bakshi2000spanning}
Bakshi, Gurdip, and Dilip Madan, 2000, Spanning and derivative-security
  valuation, {\em Journal of Financial Economics\/} 55, 205--238.

\bibitem[Balduzzi et~al.(1996)Balduzzi, Das, Foresi, and
  Sundaram]{balduzzi1996simple}
Balduzzi, Pierluigi, Sanjiv~Ranjan Das, Silverio Foresi, and Rangarajan~K
  Sundaram, 1996, A simple approach to three-factor affine term structure
  models, {\em The Journal of Fixed Income\/} 6, 43--53.

\bibitem[Bj{\"o}rk(2009)]{bjork2009arbitrage}
Bj{\"o}rk, Tomas, 2009, {\em Arbitrage Theory in Continuous Time - 3rd
  edition\/} (Oxford University Press).

\bibitem[Black and Scholes(1973)]{black1973pricing}
Black, Fischer, and Myron Scholes, 1973, The pricing of options and corporate
  liabilities, {\em Journal of Political Economy\/} 81, 637--654.

\bibitem[Breeden and Litzenberger(1978)]{breeden1978prices}
Breeden, Douglas~T, and Robert~H Litzenberger, 1978, Prices of state-contingent
  claims implicit in option prices, {\em Journal of Business\/} 51, 621--651.

\bibitem[Buchen and Kelly(1996)]{Buchen1996The}
Buchen, Peter~W., and Michael Kelly, 1996, {The maximum entropy distribution of
  an asset inferred from option prices}, {\em Journal of Financial and
  Quantitative Analysis\/} 31, 143--159.

\bibitem[Connor and Korajczyk(1989)]{connor1989intertemporal}
Connor, Gregory, and Robert~A Korajczyk, 1989, An intertemporal equilibrium
  beta pricing model, {\em The Review of Financial Studies\/} 2, 373--392.

\bibitem[Cox et~al.(1985)Cox, Ingersoll, and Ross]{Cox1985A}
Cox, John~C., Jonathan~E. Ingersoll, and Stephen~A. Ross, 1985, A theory of the
  term structure of interest rates, {\em Econometrica\/} 53, 385--407.

\bibitem[Dai and Singleton(2000)]{dai2000specification}
Dai, Qiang, and Kenneth~J Singleton, 2000, Specification analysis of affine
  term structure models, {\em The Journal of Finance\/} 55, 1943--1978.

\bibitem[Duffie et~al.(2000)Duffie, Pan, and Singleton]{Duffie2000Transform}
Duffie, Darrell, Jun Pan, and Kenneth Singleton, 2000, Transform analysis and
  asset pricing for affine jump-diffusions, {\em Econometrica\/} 68,
  1343--1376.

\bibitem[Figlewski(2009)]{Figlewski2009Estimating}
Figlewski, Stephen, 2009, Estimating the implied risk neutral density, in Tim
  Bollerslev, Jeffrey~R. Russell, and Mark Watson, eds., {\em Volatility and
  Time Series Econometrics: Essays in Honor of Robert F. Engle\/},  323--353
  (Oxford University Press, Oxford).

\bibitem[Figlewski(2018)]{Figlewski2018Risk}
Figlewski, Stephen, 2018, Risk-neutral densities: {A} review, {\em Annual
  Review of Financial Economics\/} 10, 329--359.

\bibitem[Jackwerth(2000)]{jackwerth2000recovering}
Jackwerth, Jens~Carsten, 2000, Recovering risk aversion from option prices and
  realized returns, {\em The Review of Financial Studies\/} 13, 433--451.

\bibitem[Jackwerth(2004)]{jackwerth2004option}
Jackwerth, Jens~Carsten, 2004, Option-implied risk-neutral distributions and
  risk aversion, {\em Charlotteville: Research Foundation of AIMR\/} .

\bibitem[Jarrow and Rudd(1982)]{Jarrow1982Approximate}
Jarrow, Robert, and Andrew Rudd, 1982, Approximate option valuation for
  arbitrary stochastic processes, {\em Journal of Financial Economics\/} 10,
  347--369.

\bibitem[Knez et~al.(1996)Knez, Litterman, and
  Scheinkman]{Knez1996Explorations}
Knez, Peter~J., Robert Litterman, and Jos\'{e} Scheinkman, 1996, Explorations
  into factors explaining money market returns, {\em The Journal of Finance\/}
  49, 1861--1882.

\bibitem[Kruse and N{\"o}gel(2005)]{kruse2005pricing}
Kruse, Susanne, and Ulrich N{\"o}gel, 2005, On the pricing of forward starting
  options in {H}eston's model on stochastic volatility, {\em Finance and
  Stochastics\/} 9, 233--250.

\bibitem[Madan and Milne(1994)]{Madan1994Contingent}
Madan, Dilip~B., and Frank Milne, 1994, Contingent claims valued and hedged by
  pricing and investing in a basis, {\em Mathematical Finance\/} 4, 223--245.

\bibitem[Margrabe(1978)]{margrabe1978value}
Margrabe, William, 1978, The value of an option to exchange one asset for
  another, {\em The Journal of Finance\/} 33, 177--186.

\bibitem[Melick and Thomas(1997)]{Melick1997Recovering}
Melick, William~R., and Charles~P. Thomas, 1997, Recovering an asset's implied
  pdf from option prices: An application to crude oil during the gulf crisis,
  {\em Journal of Financial and Quantitative Analysis\/} 32, 91--115.

\bibitem[Merton(1973{\natexlab{a}})]{merton1973intertemporal}
Merton, Robert~C, 1973{\natexlab{a}}, An intertemporal capital asset pricing
  model, {\em Econometrica\/} 41, 867--887.

\bibitem[Merton(1973{\natexlab{b}})]{Merton1973Theory}
Merton, Robert~C., 1973{\natexlab{b}}, Theory of rational option pricing, {\em
  Bell Journal of Economics\/} 4, 141--183.

\bibitem[Pritsker(1998)]{Pritsker1998Non}
Pritsker, Matt, 1998, {Nonparametric density estimation and tests of continuous
  time interest rate models}, {\em The Review of Financial Studies\/} 11,
  449--487.

\bibitem[Rady(1997)]{rady1997option}
Rady, Sven, 1997, Option pricing in the presence of natural boundaries and a
  quadratic diffusion term, {\em Finance and Stochastics\/} 1, 331--344.

\bibitem[Rosenberg and Engle(2002)]{rosenberg2002empirical}
Rosenberg, Joshua~V, and Robert~F Engle, 2002, Empirical pricing kernels, {\em
  Journal of Financial Economics\/} 64, 341--372.

\bibitem[S{o}derlind and Svensson(1997)]{Soderlind1997New}
S{o}derlind, Paul, and Lars Svensson, 1997, New techniques to extract market
  expectations from financial instruments, {\em Journal of Monetary
  Economics\/} 40, 383 -- 429.

\bibitem[Stutzer(1996)]{Stutzer1996A}
Stutzer, Michael, 1996, {A simple nonparametric approach to derivative security
  valuation}, {\em The Journal of Finance\/} 51, 1633--1652.

\bibitem[Vasicek(1977)]{Vasicek1977An}
Vasicek, Oldrich, 1977, An equilibrium characterization of the term structure,
  {\em Journal of Financial and Quantitative Analysis\/} 12, 627--627.

\bibitem[Vecer(2011)]{vecer2011stochastic}
Vecer, Jan, 2011, {\em Stochastic Finance: A Numeraire Approach\/} (CRC Press).

\end{thebibliography}
\putbib

\end{bibunit}

\end{appendices}

\end{document}